\PassOptionsToPackage{unicode}{hyperref}
\PassOptionsToPackage{hyphens}{url}
\PassOptionsToPackage{dvipsnames,svgnames,x11names}{xcolor}
\documentclass[12pt]{article}

\usepackage{papermacros}
\usepackage{tikz}
\usepackage{amsthm}

\usepackage{xr}

\newtheorem{example}{Example} 
\newtheorem{theorem}{Theorem}
\newtheorem{lemma}{Lemma} 
\newtheorem{proposition}{Proposition} 
\newtheorem{remark}{Remark}
\newtheorem{corollary}{Corollary}
\newtheorem{definition}{Definition}
\newtheorem{assumption}{Assumption}

\usepackage{amsmath,amssymb}
\usepackage{iftex}
\ifPDFTeX
  \usepackage[T1]{fontenc}
  \usepackage[utf8]{inputenc}
  \usepackage{textcomp} %
\else %
  \usepackage{unicode-math}
  \defaultfontfeatures{Scale=MatchLowercase}
  \defaultfontfeatures[\rmfamily]{Ligatures=TeX,Scale=1}
\fi
\usepackage{lmodern}
\ifPDFTeX\else  
\fi
\IfFileExists{upquote.sty}{\usepackage{upquote}}{}
\IfFileExists{microtype.sty}{%
  \usepackage[]{microtype}
  \UseMicrotypeSet[protrusion]{basicmath} %
}{}
\makeatletter
\@ifundefined{KOMAClassName}{%
  \IfFileExists{parskip.sty}{%
    \usepackage{parskip}
  }{%
    \setlength{\parindent}{0pt}
    \setlength{\parskip}{6pt plus 2pt minus 1pt}}
}{%
  \KOMAoptions{parskip=half}}
\makeatother
\usepackage{xcolor}
\makeatletter
\ifx\paragraph\undefined\else
  \let\oldparagraph\paragraph
  \renewcommand{\paragraph}{
    \@ifstar
      \xxxParagraphStar
      \xxxParagraphNoStar
  }
  \newcommand{\xxxParagraphStar}[1]{\oldparagraph*{#1}\mbox{}}
  \newcommand{\xxxParagraphNoStar}[1]{\oldparagraph{#1}\mbox{}}
\fi
\ifx\subparagraph\undefined\else
  \let\oldsubparagraph\subparagraph
  \renewcommand{\subparagraph}{
    \@ifstar
      \xxxSubParagraphStar
      \xxxSubParagraphNoStar
  }
  \newcommand{\xxxSubParagraphStar}[1]{\oldsubparagraph*{#1}\mbox{}}
  \newcommand{\xxxSubParagraphNoStar}[1]{\oldsubparagraph{#1}\mbox{}}
\fi
\makeatother

\usepackage{longtable,booktabs,array}
\usepackage{calc} %
\usepackage{etoolbox}
\makeatletter
\patchcmd\longtable{\par}{\if@noskipsec\mbox{}\fi\par}{}{}
\makeatother
\IfFileExists{footnotehyper.sty}{\usepackage{footnotehyper}}{\usepackage{footnote}}
\makesavenoteenv{longtable}
\usepackage{graphicx}
\makeatletter
\def\maxwidth{\ifdim\Gin@nat@width>\linewidth\linewidth\else\Gin@nat@width\fi}
\def\maxheight{\ifdim\Gin@nat@height>\textheight\textheight\else\Gin@nat@height\fi}
\makeatother
\setkeys{Gin}{width=\maxwidth,height=\maxheight,keepaspectratio}
\makeatletter
\def\fps@figure{htbp}
\makeatother

\makeatletter
\@ifpackageloaded{caption}{}{\usepackage{caption}}
\AtBeginDocument{%
\ifdefined\contentsname
  \renewcommand*\contentsname{Table of contents}
\else
  \newcommand\contentsname{Table of contents}
\fi
\ifdefined\listfigurename
  \renewcommand*\listfigurename{List of Figures}
\else
  \newcommand\listfigurename{List of Figures}
\fi
\ifdefined\listtablename
  \renewcommand*\listtablename{List of Tables}
\else
  \newcommand\listtablename{List of Tables}
\fi
\ifdefined\figurename
  \renewcommand*\figurename{Figure}
\else
  \newcommand\figurename{Figure}
\fi
\ifdefined\tablename
  \renewcommand*\tablename{Table}
\else
  \newcommand\tablename{Table}
\fi
}
\@ifpackageloaded{float}{}{\usepackage{float}}
\floatstyle{ruled}
\@ifundefined{c@chapter}{\newfloat{codelisting}{h}{lop}}{\newfloat{codelisting}{h}{lop}[chapter]}
\floatname{codelisting}{Listing}

\makeatother
\makeatletter
\@ifpackageloaded{caption}{}{\usepackage{caption}}
\@ifpackageloaded{subcaption}{}{\usepackage{subcaption}}
\makeatother

\ifLuaTeX
  \usepackage{selnolig}  %
\fi
\usepackage[]{natbib}
\usepackage{bookmark}

\IfFileExists{xurl.sty}{\usepackage{xurl}}{} %
\hypersetup{
  pdftitle={Title},
  pdfauthor={Author 1; Author 2},
  pdfkeywords={3 to 6 keywords, that do not appear in the title},
  colorlinks=true,
  linkcolor={blue},
  filecolor={Maroon},
  citecolor={Blue},
  urlcolor={Blue},
  pdfcreator={LaTeX via pandoc}}

\newcommand{\anon}{1}

\begin{document}

\def\spacingset#1{\renewcommand{\baselinestretch}%
  {#1}\small\normalsize} \spacingset{1}

\if1\anon
  {
    \title{\bf Minimax rates for the linear-in-means model reveal an identifiability-estimability gap}
    \author{Alex Hayes\thanks{
        We thank Karl Rohe, Ralph Trane, Felix Elwert, Hyunseung Kang, Sameer Deshpande, Edward McFowland, Lihua Lei, Liza Levina, Joshua Cape, Carey Priebe, Minh Tang, Robert Lunde, Laura Forastiere, Betsy Ogburn, Dean Eckles, Vincent Boucher, Nicholas Christakis, Johan Ugander, Paul Goldsmith-Pinkham and the attendees of the University of Wisconsin-Madison IFDS Ideas Seminar for their helpful comments and suggestions. Support for this research was provided by the University of Wisconsin-Madison Office of the Vice Chancellor for Research and Graduate Education with funding from the Wisconsin Alumni Research Foundation, as well as NSF grants DMS 2052918 and DMS 2023239. Support for this research was also provided by American Family Insurance through a research partnership with the University of Wisconsin-Madison's American Family Insurance Data Science Institute. We note the use of Claude Opus 4.1, a Generative AI tool, for feedback on the writing in this manuscript.}\hspace{.2cm}\\
      Department of Economics, Stanford University \\
      and \\
      Keith Levin \\
      Department of Statistics, University of Wisconsin-Madison}
    \maketitle
  } \fi

\if0\anon
  {
    \bigskip
    \bigskip
    \bigskip
    \begin{center}
      {\LARGE \bf Minimax rates for the linear-in-means model\newline \newline reveal an identifiability-estimability gap}
    \end{center}
    \medskip
  } \fi

\bigskip
\begin{abstract}
  The linear-in-means model is widely used to study peer influence in social networks. We consider estimation in the linear-in-means model when a randomized treatment is applied to nodes in a network. We show that even when peer effects are identified, they may not be estimable at standard rates, due to near-perfect collinearity. We prove a minimax lower bound on estimation error and show that estimation becomes more difficult as networks grow denser. In sufficiently dense networks, consistent estimation of peer effects is impossible. To address this challenge, we investigate network-dependent treatment assignment. Using random dot product graphs, we show that treatments depending on network structure can prevent asymptotic collinearity when there is sufficient degree heterogeneity. However, such dependence is not a panacea, as different dependence structures must be individually evaluated for estimability. These results suggest caution when using the linear-in-means model to estimate peer effects and highlight the importance of explicitly modeling the relationship between treatments and network structure.
\end{abstract}

\noindent%
\vfill

\newpage

\section{Introduction}
\label{sec:intro}

The linear-in-means model is a canonical approach to estimating peer influence in social networks \citep{bramoulle2020, blume2011}. Suppose there is a network with $n$ nodes, encoded by a symmetric adjacency matrix $\A \in \R^{n \times n}$. In binary networks, $\A_{ij} = 1$ if nodes $i$ and $j$ form an edge, and $\A_{ij} = 0$ otherwise. Each node $i \in [n]$ is associated with an outcome $Y_i \in \R$ and a covariate $T_i \in \R$. Letting $\mathcal{N}_i = \left\{ j \in [n]: \A_{ij} = 1 \right\}$ denote the neighbors of node $i$, the treatments and outcomes of the neighbors of node $i$ are modeled as influencing $i$'s outcome according to
\begin{equation}
  \label{eq:lim-avg}
  Y_i =
  \alpha +
  \frac{ \beta }{\abs{\mathcal{N}_i}}\sum_{j \in \mathcal{N}_i} Y_j +
  \gamma T_i +
  \frac{\delta}{\abs{\mathcal{N}_i}}\sum_{j \in \mathcal{N}_i} T_j +
  \varepsilon_i.
\end{equation}
The coefficient $\beta$ (the ``contagion term'') measures how peer outcomes $Y_j$ influence $Y_i$, while $\delta$ (the ``interference term'') measures how peer treatments $T_j$ influence $Y_i$. Linear-in-means models have been heavily used to investigate social effects in education, crime, health and social policy \citep{sacerdote2001, epple2011, soetevent2007, trogdon2008, duflo2003, bertrand2000, glaeser1996, patacchini2012a, carrell2013}.

The linear-in-means model has received considerable theoretical attention because peer effect terms can be perfectly collinear, causing identification failure \citep{bramoulle2020}. \cite{manski1993} famously showed that this perfect collinearity occurs in highly structured social networks. Despite Manski's negative result, later generalizations by \cite{bramoulle2009} and \cite{martellosio2022} showed that peer effects are in general identified, subject to mild constraints on network structure.

Researchers often assume that identification alone justifies inference from a model. While identification ensures that model parameters correspond to unique data distributions, it does not imply that parameters can be estimated consistently, even with infinite data \citep{maclaren2020a, hennig2024, tibshirani1988}. We show that in the linear-in-means model, parameters can be identified at every finite sample size, yet estimates may not be consistent, even in gold standard scenarios such as random experiments.

\begin{figure}[t]
  \centering
  \includegraphics{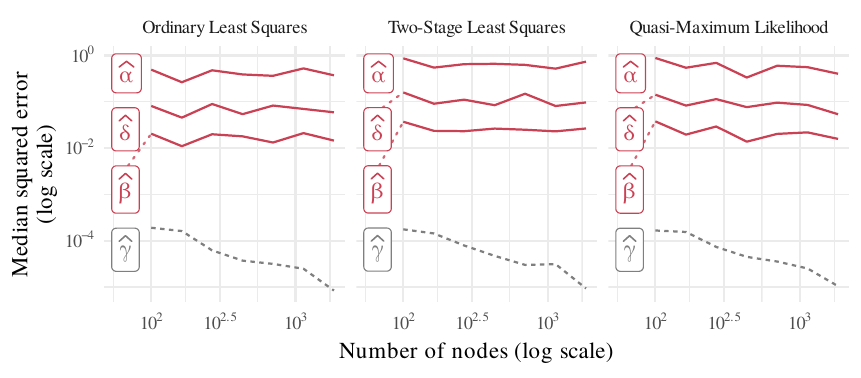}
  \caption{Median squared error of estimates. Each panel corresponds to a different estimator. Within a panel, the x-axis represents the sample size on a log scale, and the y-axis represents the Monte Carlo estimate of median squared error, also on a log scale. Each line corresponds to a single coefficient; solid red lines are asymptotically collinear, dashed gray lines are not.}
  \label{fig:mse-problem}
\end{figure}

Figure~\ref{fig:mse-problem} illustrates this phenomenon through Monte Carlo simulations considering the median squared error of three widely-used estimators. Despite all parameters being identified, none of the estimators recover the regression coefficients $\alpha, \beta$ and $\delta$; their median squared error fails to decrease with growing sample size. The underlying issue is a degeneracy in the design matrix: while its columns remain linearly independent for finite $n$, the contagion and interference columns converge to constants, a setting known as a ``nearly singular design,'' as constant columns are collinear with the intercept \citep{phillips2016, knight2008}.

This degeneracy arises from a simple but overlooked mechanism. Consider a randomized experiment where each node receives treatment independently with probability $\pi \in (0, 1)$. For each node, the interference column represents the proportion of treated peers, which is an average over $\abs*{\Ni}$ peers. As the network grows and degrees increase, this proportion converges uniformly to $\pi$, rendering the interference column asymptotically collinear with the intercept. The contagion term exhibits similar behavior through its dependence on neighborhood averages in the reduced form. The result is an asymptotically shrinking signal-to-noise ratio. While the columns of the design matrix are linearly independent for every finite sample size $n$, in Theorem~\ref{lem:indepcov:nonid} we show that the peer effect columns of the design matrix converge uniformly to constant multiples of the intercept column. In Theorem~\ref{thm:LIM:minimax} we show that, as a consequence of this collinearity, estimates of peer effects can converge at slower than $\sqrt{n}$ rates, leading to confidence intervals with incorrect coverage and potentially inconsistent point estimates.

Past work has examined asymptotic collinearity in limited contexts, showing it can cause inconsistency or slower than parametric convergence rates in quasi-maximum likelihood estimators \citep[][Theorem 5.2]{lee2004}. In concurrent work, \cite{wang2025j} derives lower bounds on the estimation error of ordinary least squares and two-stage least squares estimators in the linear-in-means model. Our own minimax results are substantially broader, as our lower bounds apply to all possible estimators.

These lower bounds are particularly relevant for new methods leveraging low-rank or graphon structure in network models \citep{bhadra2025, paul2022a, li2022f}. In the cited works, networks become increasingly dense as sample size grows, and nodal covariates are often modeled as randomized treatments independent of network structure---precisely the scenario that causes asymptotic collinearity. We emphasize that this modeling choice involves important tradeoffs in terms of estimability.

Given that asymptotic collinearity stems from independence between treatments and network structure, we also investigate whether dependence can prevent degeneracies in the design matrix. In Theorem~\ref{thm:rdpg}, we show that dependence between covariates and network structure in random dot product graphs can sometimes, but not always, prevent asymptotic collinearity. This result provides guidance for the growing literature investigating peer effects via random dot product graphs \citep{paul2024, paul2022a, bhadra2025, li2022f}, clarifying when modeling choices lead to estimability issues.

\textbf{Our contributions}. We make three primary contributions to understanding the linear-in-means model:

\begin{enumerate}
  \item We prove that peer effects in the linear-in-means model can be asymptotically collinear even when identified (Theorem~\ref{lem:indepcov:nonid}).
  \item We establish minimax lower bounds on estimation error for peer effects in the linear-in-means model in random experiments, showing that peer effects may not be estimable at the parametric rate, and become completely inestimable in sufficiently dense networks (Theorem~\ref{thm:LIM:minimax}).
  \item We demonstrate that covariate-network dependence may prevent asymptotic collinearity in random dot product graphs (Theorem~\ref{thm:rdpg}), showing that degree heterogeneity is necessary (but not sufficient) to ensure a nearly full rank design matrix in the asymptotic limit
\end{enumerate}

\subsection*{Notation} \label{notation} For a matrix $\A$, let $\norm*{\A}, \norm*{\A}_F$ and $\norm*{\A}_{2, \infty}$ denote the spectral, Frobenius, and two-to-infinity norms, respectively. For a matrix $\A$, we write $\A_{i \cdot}$ for its $i$-th row and $\A_{\cdot j}$ for its $j$-th column. We use standard Landau notation, e.g., $\bigoh{a_n}$ and $\littleoh{a_n}$ to denote growth rates, as well as the probabilistic variants $\Op{a_n}$ and $\op{a_n}$. For example, $g(n) = \bigoh{f(n)}$ means that for some constant $C>0$, $|g(n)| \le C f(n)$ for all suitably large $n$. In proofs, $C$ denotes a constant not depending on the number of vertices $n$, whose precise value may change from line to line, and occasionally within the same line. We say that an estimator $\widehat{\theta}$ converges at rate $h(n)$ when $h(n) \paren{\widehat{\theta} - \theta}$ is $\mathcal O_p(1)$.

\section{The linear-in-means model}

The linear-in-means model captures how outcomes spread via peer influence in social networks. In this section, we introduce the data generating process and describe conditions for identifiability and estimability.
To analyze the model formally, we express it in matrix form. Let the adjacency matrix $\A$ have non-negative real entries, representing either weighted or binary edges. Define the degree matrix $\D = \diag(d_1, d_2, \dots, d_n)$, where $d_i = \sum_j \A_{ij}$ is the degree of node $i$, and let $\G = \D^{-1} \A$ be the row-normalized adjacency matrix. Multiplication by $\G$ performs neighborhood averaging: $[\G \Y]_i = d_i^{-1} \sum_j \A_{ij} Y_j$ (see Figure~\ref{fig:averaging}).

\tikzset{every loop/.style={}}
\begin{figure}[t]
  \begin{minipage}{0.49\textwidth}
    \centering
    \begin{tikzpicture}
      \node[shape=circle,fill=Maroon,label=above left:{$T_A = 1$}] (A) at (0,1) {};
      \node[shape=circle,fill=Gray,label=below left:{$T_B = 0$}] (B) at (1,0) {};
      \node[shape=circle,fill=Maroon,label=above right:{$T_C = 1$}] (C) at (1.5,1.5) {};
      \node[shape=circle,fill=Maroon,label=above right:{$T_D = 1$}] (D) at (2.75,0.5) {};

      \draw (A) -- (B);
      \draw (A) -- (C);
      \draw (B) -- (C);
      \draw (C) -- (D);
    \end{tikzpicture}
  \end{minipage}
  \begin{minipage}{0.49\textwidth}
    \centering
    \begin{tikzpicture}
      \node[shape=circle,fill=Maroon,label=above left:{$[\G \T]_A = 1/2$}] (A) at (0,1) {};
      \node[shape=circle,fill=Gray,label=below left:{$[\G \T]_B = 1$}] (B) at (1,0) {};
      \node[shape=circle,fill=Maroon,label=above right:{$[\G \T]_C = 2/3$}] (C) at (1.5,1.5) {};
      \node[shape=circle,fill=Maroon,label=above right:{$[\G \T]_D = 1$}] (D) at (2.75,0.5) {};

      \draw (A) -- (B);
      \draw (A) -- (C);
      \draw (B) -- (C);
      \draw (C) -- (D);
    \end{tikzpicture}
  \end{minipage}
  \caption{Neighborhood averaging. Left: A binary covariate $\T$ on a small network. Red indicates one treatment value, gray another. Right: The average values of $\T$ in each node's neighborhood. For example, node $A$ is connected to nodes $B$ and $C$, the average value of $\T$ in the neighborhood centered on $A$ is $1/2$ (the value of $\T$ at node $A$ is excluded from this calculation). Similarly, the average value of $\T$ in the neighborhood centered on $B$ is $1$. Red continues to indicate on treatment value and gray another.}
  \label{fig:averaging}
\end{figure}
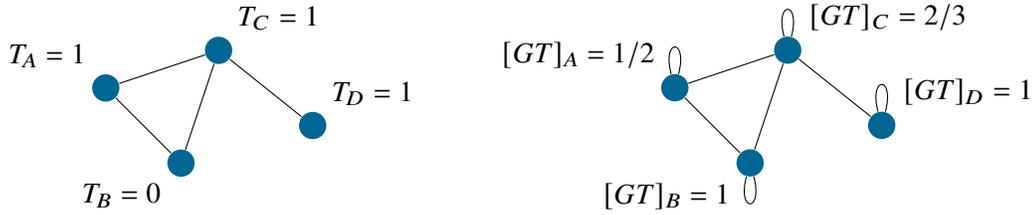

Using this notation, we rewrite Equation~\eqref{eq:lim-avg} as:
\begin{equation} \label{eq:lim-mv}
  \Y = \alpha \onevec_n + \beta \G \Y + \T \gamma + \G \T \delta + \bm \varepsilon,
\end{equation}
where $\G \Y \in \R^n$ and $\G \T \in \R^n$ encode neighborhood averages of outcomes and treatments, respectively. The design matrix becomes:
\begin{equation} \label{eq:def:Wn}
  \W_n = \begin{bmatrix} \onevec_n & \G \Y & \T & \G \T \end{bmatrix} \in \R^{n \times 4}.
\end{equation}

We assume, as is typical, that the errors $\varepsilon$ are mean-zero and independent of the network $\A$ and nodal covariates $\T$. Typically one does not impose any parametric assumptions on the distribution of the entries of $\bm \varepsilon$, although we will assume that $\bm \varepsilon$ is Gaussian in our minimax result (Theorem~\ref{thm:LIM:minimax}). Later, we will consider multiple nodal covariates, allowing vector-valued treatments and thus $\gamma$ and $\delta$ will be vector-valued, as well.

The simultaneity in this model---each $Y_i$ depending on all other outcomes---requires special treatment. Requiring $|\beta| < 1$ ensures that $(\I - \beta \G)$ is invertible, allowing us to solve for $\Y$ and apply the Neumann expansion:
\begin{equation} \label{eq:lim-red}
  \Y = (I - \beta \G)^{-1} (\onevec_n \alpha + \T \gamma + \G \T \delta + \bm \varepsilon) = \sum_{k=0}^\infty \beta^k \G^k (\onevec_n \alpha + \T \gamma + \G \T \delta + \bm \varepsilon).
\end{equation}

This \emph{reduced form} reveals an intuitive interpretation: outcomes reach equilibrium via repeated neighborhood averaging. From initial outcomes $\Y^{(0)} = \alpha \onevec_n + \gamma \T + \delta \G + \bm \varepsilon$, we iteratively compute $\Y^{(t+1)} = \beta \G \Y^{(t)} + \Y^{(0)}$, diffusing influence throughout the network. The constraint $|\beta| < 1$ guarantees convergence to a unique equilibrium \citep{besag1974, rue2005a}.
The diffusion interpretation of the model means that interpreting individual coefficients requires care. The typical marginal effect interpretation fails due to the model's non-linearity; \citet[][Chapter 2]{lesage2009} develop impact scores that preserve marginal interpretations, and \cite{bhadra2025} connect these to causal estimands. The contagion coefficient $\beta$ generally lacks a standard causal interpretation \citep{debarsy2025}, though some authors propose interpretations via perturbed equilibria \citep{shpitser2024, ogburn2020, tchetgentchetgen2021a}.

With these details about the linear-in-means model in hand, we turn to identification. Identification is a property of a statistical model, which states that parameters of the model uniquely determine the data generating process.

\begin{definition} \label{def:identified}
  Consider a statistical model $\{P_\theta : \theta \in \Theta\}$ where $P_\theta$ is the law of the data generating process, and $\theta$ are the parameters of that process. A parameter $\theta$ is \emph{identified} if and only if $\theta_1 \neq \theta_2$ implies $P_{\theta_1} \neq P_{\theta_2}$. A parameter $q(\theta)$ is \emph{identified} if and only if $q(\theta_1) \neq q(\theta_2)$ implies $P_{\theta_1} \neq P_{\theta_2}$.
\end{definition}

In the linear-in-means model, we are interested in $q(\theta) = (\alpha, \beta, \gamma, \delta)$ and we ignore identification of parameters governing the distributions of $\T$ and $\bm \varepsilon$, since they are irrelevant for our purposes (we thus abuse notation and write $\theta = (\alpha, \beta, \gamma, \delta)$ in our proofs for convenience). Importantly, we are interested in identifying peer effects from a single social network, and thus rely on recent identification results from \cite{martellosio2022} rather than the more widely known reduced-form identification results of \cite{bramoulle2009}, which require multiple networks or repeated observations of the same network.

\begin{proposition}[\citealt{martellosio2022}]
  \label{prop:martellioso2022}
  Consider a single network with a fixed number of nodes $n$. Suppose that $\bbE[\bm \varepsilon \mid \T, \G] = 0$ and let
  \begin{equation*}
    Y = \onevec_n \alpha + \G \Y \beta + \T \gamma + \G \T \delta + \bm \varepsilon
  \end{equation*}
  where $\abs{\beta} < 1$. Then $\alpha, \beta, \gamma$ and $\delta$ are generically identified if and only if $\onevec_n, \T, \G \T$ and $\G^2 \T$ are linearly independent.
\end{proposition}

Note that we do not need to specify the distributions of $\T$ or $\varepsilon$ to identify $\alpha, \beta, \gamma$ and $\delta$. It is typical to assume that $\onevec_n, \T$ and $\G \T$ are linearly independent, and the question of identifiability reduces to whether or not these vectors are collinear with $\G^2 \T$. In many probabilistic network models, such as random dot product graphs (see Definition~\ref{def:rdpg} below), $\G$ has three or more distinct eigenvalues, such that $\I, \G$ and $\G^2$ are linearly independent (see Proposition~\ref{prop:three-eig} in Appendix~\ref{app:proof:prop:three-eig}). This is almost, but not quite, sufficient for full linear independence of $\onevec_n, \T, \G \T$ and $\G^2 \T$. See Example 9 of \cite{martellosio2022} for detailed discussion of settings where identification fails.

Identification alone does not guarantee consistent estimation. As \cite{hennig2024} and \cite{maclaren2020a} emphasize, parameters can be identified yet inestimable. In the linear-in-means model, standard estimation approaches such as two-stage least squares, quasi-maximum likelihood, and ordinary least squares all require that $\W_n^\top \W_n / n$ be asymptotically non-singular (see Assumption 7b of \cite{kelejian1998}, Assumption 8 of \cite{lee2004} and Assumption 5 of \cite{lee2002}). Identification does not guarantee that this condition holds.

\section{Peer effects are asymptotically collinear in random experiments}

As discussed above, most work in the linear-in-means model has assumed that the design matrix covariance $\W_n^\top \W_n / n$ converges to a non-singular limit.
Unfortunately, as we will see, this condition can fail when the treatments $\T$ are independent of the network.

\begin{definition}
  Suppose $\W_n^\top \W_n / n$ converges to a limit $\bm \Sigma$. Two or more columns of the design matrix given in Equation~\eqref{eq:def:Wn} are \emph{asymptotically collinear} when the corresponding columns of $\bm \Sigma$ are linearly dependent.
\end{definition}

Asymptotic collinearity occurs under mild network conditions---even when Proposition~\ref{prop:martellioso2022}'s identifying conditions hold.

\begin{theorem} \label{lem:indepcov:nonid}
  Suppose that (1) the nodal covariates $T_1,T_2,\dots,T_n$ have shared mean $\tau \in \R$ and  the centered nodal covariates $\{ T_i - \tau : i \in 1,2,\dots, n \}$, are independent $(\nu,b)$-subgamma random variables; (2) $\T$ is independent of $\A$; (3) the regression errors $\varepsilon_1, \varepsilon_2, \dots, \varepsilon_n$ are independent subgamma random variables with parameters not depending on $n$, independent of $T_1,T_2,\dots, T_n$; and (4) the adjacency matrix $\A$ contains only non-negative entries and does not contain any self-loops, such that $A_{ii} = 0$ for all $i = 1,2,\dots, n$.

  If the degrees of the network grow such that
  \begin{equation} \label{eq:growth:nub}
    \max_{i \in [n] } \frac{1}{d_i^2} \sum_{j=1}^n \A_{ij}^2
    = o\left( \frac{ 1 }{ \nu \log^2 n } \right)
    ~\text{ and }
    \max_{i,j \in [n]} \frac{ \A_{ij} }{ d_i }
    = o\left( \frac{ 1 }{ b \log n } \right).
  \end{equation}
  then
  \begin{equation*}
    \max_{i \in [n]} \Big| [\G \T]_i - \tau \Big|
    = o(1) ~ \text{ almost surely }
  \end{equation*}
  and
  \begin{equation*}
    \max_{i \in [n]} \Big| [\G \Y]_i - \eta \Big|
    = o(1) ~ \text{ almost surely,}
  \end{equation*}
  where
  \begin{equation} \label{eq:def:eta}
    \eta = \frac{ \alpha + (\gamma+\delta)\tau }{ 1-\beta }.
  \end{equation}
\end{theorem}

A proof of this result can be found in Appendix~\ref{apx:proof:lem:indepcov:nonid}. The first condition of Theorem~\ref{lem:indepcov:nonid} requires that the nodal covariates are independent across nodes. The sub-gamma condition ensures that we can control averages of the $T_i$ using standard concentration inequalities. The class of subgamma distributions is broad, and includes as special cases the Bernoulli, Poisson, Exponential, Gamma, and Gaussian distributions, as well as any sub-Gaussian or squared sub-Gaussian distribution and all bounded distributions \citep{boucheron2013,vershynin2020}.

The second condition of Theorem~\ref{lem:indepcov:nonid} requires that nodal covariates $\T$ are independent of the network. This is the case for many controlled experiments on networks. The third condition, similar to the first, requires that the regression errors $\varepsilon$ are not too heavy-tailed. The fourth condition requires that the network has no self-loops. We anticipate that this requirement can be relaxed, but we do not pursue this here.

We note that weighted networks are allowed in Theorem~\ref{lem:indepcov:nonid}, so long as all edges have non-negative edges weights. The condition in Equation~\eqref{eq:growth:nub} requires that the size of each neighborhood is growing. When the network is binary, such that $\A_{ij} \in \set{0, 1}$, the condition on the degrees reduces to $\min_{i \in [n]} d_i = \omega(\log n)$. That is, the size of the smallest neighborhood must grow faster than $\log n$. This implies that no nodes in the network are isolated. In the more general case of a weighted, non-negative network, the condition in Equation~\eqref{eq:growth:nub} essentially requires that no individual edge accounts for too much of the total ``edge weight'' incident on any vertex. Although we present Theorem~\ref{lem:indepcov:nonid} in the context of scalar nodal covariates, the theorem can be extended to multiple nodal covariates, provided that the subgamma assumption holds for each covariate.

When the conditions of Theorem~\ref{lem:indepcov:nonid} hold, the interference term $\G \T$ and the contagion term $\G \Y$ become collinear with the intercept in the large-network limit (i.e., as the number of vertices $n$ grows). As discussed in the introduction, the intuition is that the neighborhood averages of $\T$ converge to the expected value of $\T$ when $\T$ is independent of the network. What is perhaps surprising is that the contagion term $\G \Y$ behaves similarly to the interference term $\G \T$. This can be better seen by multiplying through by $\G$ in Equation~\eqref{eq:lim-red}. When no node is isolated (i.e., all nodes have positive degree),
\begin{equation} \label{eq:GY:expand}
  \G \Y = \frac{\alpha}{1 - \beta} \onevec_n + \gamma \G \T
  + (\gamma \beta + \delta) \sum_{k=0}^\infty \beta^k \G^{k+2} \T
  + \sum_{k=0}^\infty \beta^k \G^{k + 1} \bm \varepsilon.
\end{equation}
The first right-hand term is a constant vector, and we have already argued that the second term converges to a constant. The third right-hand term expands to
\begin{equation*}
  (\gamma \beta + \delta) \G^2 \T + (\gamma \beta + \delta) \beta \G^3 \T + (\gamma \beta + \delta) \beta^2 \G^4 \T + \cdots .
\end{equation*}
When $\G \T$ is near-constant, we have $\G^2 \T = \G (\G \T) \approx \G \T$, and higher-order summands of this term behave, once again, similarly to a vector of constants. The final right-hand term in Equation~\eqref{eq:GY:expand} is zero in expectation, and thus irrelevant for identification pursues. Nonetheless, one can see that, by a similar argument as for the third term, the fourth term should converge to a column vector of zeroes.

Altogether, the implication is that the contagion term $\G \Y$ converges to a constant when the interference term $\G \T$ converges to a constant. This is particularly concerning, as $\G \T$ is simply a collection of averages with shared expectation, and we thus anticipate that the entries of $\G \T$ will all converge to $\tau$, the expectation of $\T$, under a wide variety of circumstances.

The consequence of this asymptotic collinearity is the peer effects may not be estimatable at the usual $\sqrt{n}$ rates. We show this by proving a minimax lower bound on estimation rates for the contagion, interference and intercept terms in the model. These minimax bounds characterize the best possible estimation rates that hold across the entire parameter space of the linear-in-means model.

\begin{theorem} \label{thm:LIM:minimax}
  Under the model in Equation~\eqref{eq:lim-mv}, suppose that $\varepsilon \sim \mathcal N(0, \sigmaeps^2 \I)$ and suppose that the entries of $\T$ are drawn i.i.d.~according to a distribution with mean $\tau \in \R$.
  Then there exists some $N$ such that for all $n > N$, there exist positive constants $c_\beta,c_\delta$ and $c_0$, such that
  \begin{equation} \label{eq:minimax:betadelta}
    \inf_{\thetahat} \sup_{\theta \in \ThetaLIM}
    \min\left\{
    \bbP_{\theta}\left[
      \left| \betahat - \beta \right| \ge \frac{c_\beta}{\| \G \|_F} \right] ,
    ~
    \bbP_{\theta}\left[
      \left| \deltahat - \delta \right| \ge \frac{c_\delta}{\| \G \|_F} \right]
    \right\}
    \ge c_0 ,
  \end{equation}
  where the infimum is over all estimators and $\ThetaLIM = \left\{ (\alpha,\beta,\gamma,\delta) : \alpha,\gamma,\delta \in \R, \beta \in (-1,1) \right\}$. Further, if $\tau \neq 0$, there exist positive constants $c_\alpha$ and $c_\alpha'$ such that
  \begin{equation*}
    \inf_{\thetahat} \sup_{\theta \in \ThetaLIM}
    \bbP_{\theta}\left[
      \left| \alphahat - \alpha \right| \ge \frac{c_\alpha}{\| \G \|_F} \right]
    \ge c_\alpha' .
  \end{equation*}
\end{theorem}

A proof is given in Appendix~\ref{apx:LIM:minimax}. Theorem~\ref{thm:LIM:minimax} establishes minimax lower bounds for estimating peer effects in the linear-in-means model. The theorem shows that every estimator must have error at least of $1/\| \G \|_F$ with positive probability, where $\| \G \|_F$ is the Frobenius norm of the row-normalized adjacency matrix.

To understand the implications of Theorem~\ref{thm:LIM:minimax}, we relate the Frobenius norm of the row-normalized adjacency matrix to the harmonic and minimum degrees of the network. In binary networks, $\| \G \|_F^2 = \sum_{i=1}^n 1 / d_i$, which equals $n/\bar{d}_{\text{har}}$ where
\begin{equation*}
  \bar{d}_{\text{har}} = \frac{n}{\sum_{i=1}^n 1/d_i}
\end{equation*}
is the harmonic mean of the degrees. Therefore, $\| \G \|_F = \sqrt{n/\bar{d}_{\text{har}}}$. Since the harmonic mean degree is trivially lower bounded by the minimum degree, we have $\| \G \|_F > \sqrt{n / \dmin}$, and Theorem~\ref{thm:LIM:minimax} can be interpreted in terms of the growth rates on the minimum degree with a slight loss of tightness.

In particular, Theorem~\ref{thm:LIM:minimax} states that the best possible rate of convergence for any estimator is $\sqrt{n}$ when $\bar{d}_{\text{har}} = O(1)$ (sparse networks with bounded degrees) and $n^{1/4}$ when $\bar{d}_{\text{har}} = \Theta(\sqrt{n})$ (moderately dense networks). When $\bar{d}_{\text{har}} = \Theta(n)$ (dense networks where most nodes are linked to a constant fraction of the network), consistent estimation is impossible.

We now consider some specific examples to build intuition. To do so, we must first introduce the stochastic blockmodel.

\begin{definition}[Poisson Degree-Corrected Stochastic Blockmodel]\label{def:sbm}
  The Poisson degree-corrected stochastic blockmodel \citep{rohe2018, karrer2011} is an undirected network model, in which each node, indexed by $i=1,2,\dots,n$, is assigned a community (i.e., block) $z_i \in \set{1,2,\dots, d}$ with probability $\Pr({z_i = k}) = \pi_k$, and a degree-correction parameter $\xi_i$, which describes the propensity of vertex $i$ to connect with other nodes. Conditional on block memberships and degree-correction parameters, edges are generated independently between every pair of vertices in the network according to a Poisson distribution. The expected number of edges between two vertices depends on their community memberships, their degree correction parameters, a positive semi-definite matrix $\B \in [0, 1]^{d \times d}$ of inter-block edge formation probabilities, and a scaling factor $\rho_n \in [0, 1]$ (which may vary with $n$), according to
  \begin{equation*}
    \E[z_i, z_j, \xi_i, \xi_j]{\A_{ij} = 1} = \rho_n \, \xi_i \B_{z_i, z_j} \xi_j.
  \end{equation*}
\end{definition}

\begin{example}[Bernoulli experiments on stochastic blockmodels]
  Suppose that the network $\A$ follows a stochastic blockmodel with or without degree correction, and the nodal covariate is independently and identically $\mathrm{Bernoulli}(p)$. Suppose that the average degree of the network grows at a $\log^2 n$ rate, which in turn implies that both the minimum and the maximum degrees of the network grow at $\log^2 n$ rates. Then $\onevec_n, \G \T$ and $\G \Y$ are asymptotically collinear by Theorem~\ref{lem:indepcov:nonid}. By Theorem~\ref{thm:LIM:minimax}, provided that $\tau \neq 0$, any estimators $\alphahat, \betahat$ and $\deltahat$ are at best consistent at $(n / \log^2 n)^{1/2}$ rates. If the average degree grows at a faster $n^{1/2}$ rate, $\alphahat, \betahat$ and $\deltahat$ are all at best consistent at $n^{1/4}$ rates. If the average degree grows linearly in $n$, there are no consistent estimators of $\alpha, \beta$ and $\delta$.
\end{example}

\begin{example}[Bernoulli experiments with partial interference]
  Suppose the vertices of network $\A$ are divided into $g$ different groups, indexed by $j = 1,2,\dots, g$, and each node within a group is connected to each other node in that same group. Let $n_j$ denote the number of nodes beloning to group $j \in [g]$. If the nodal covariates are independently and identically $\mathrm{Bernoulli}(p)$ and $\min_{j \in [g]} n_j = \omega( \sqrt{n} )$, then $\onevec_n, \G \T$ and $\G \Y$ are asymptotically collinear by Theorem~\ref{lem:indepcov:nonid}. By Theorem~\ref{thm:LIM:minimax}, any estimators $\alphahat, \betahat$ and $\deltahat$ are at best consistent at $n^{1/4}$ rates. If the minimum group size $\min_{j \in [g]} n_j$ grows linearly in $n$, there are no consistent estimators of $\alpha, \beta$ and $\delta$.
\end{example}

\begin{remark}[Non-random covariates] \label{rem:fixedcovars}
  Since we assume that $T_1, T_2, \dots, T_n$ are independent covariates with shared expectation, Theorem~\ref{thm:LIM:minimax} does not apply if the nodal covariates $\T$ are fixed and non-random. In observational contexts, it may be appropriate to assume that covariates are fixed. In the context of randomized experiments, however, it is not desirable to model treatment assignments as fixed.
\end{remark}

\section{Peer effects are partially collinear in random dot product graphs}

We now examine a model where dependence between nodal covariates and network structure prevents neighborhood averages from converging to constants \citep[see also][on the importance of the relationship between the network and nodal covariates]{case1991,martellosio2022}. As our above results illustrate, when the network is independent of nodal covariates, interference terms and contagion terms both become collinear with the intercept. When the network is dependent on nodal covariates, the story becomes more complicated. We first introduce a general probabilistic model for random networks. This model includes as special cases a number of widely-used network models that may be more familiar to practitioners, and we discuss these special-case models in more detail below.

\begin{definition}[Random Dot Product Graph, \citealt{young2007}]
  \label{def:rdpg}
  Let $F$ be a distribution on $\R^d$ such that $0 \le x^\top y$ for all $x,y \in \supp F$ and the convex cone of $\supp F$ is $d$-dimensional.
  Draw $\X_1, \X_2, \dots, \X_n$ independently according to $F$, and collect these in the rows of $\X \in \R^{n \times d}$ for ease of notation.
  Conditional on these $n$ vectors, which we call {\em latent positions}, generate edges by drawing $\{ \A_{ij} : 1 \le i < j \le n \}$ as independent $(\nu,b)$-subgamma random variables with $\bbE[ \A_{ij} \mid \X ] = \rho_n \X_i^\top \X_j$, where $\rho_n \in [0,1]$ can vary as a function of $n$.
  Then we say that $\A$ is distributed according to an $n$-vertex random dot product graph with latent position distribution $F$, $(\nu,b)$-subgamma edges and sparsity factor $\rho_n$.
  We write $(\A, \X) \sim \RDPG( F, n)$, with the subgamma and sparsity parameters made clear from the context.
\end{definition}

Definition~\ref{def:rdpg} is a slight generalization of the random dot product graph as it was originally introduced \citep{young2007,athreya2018}. The random dot product graph as defined in \cite{young2007} assumes that $0 \le x^\top y \le 1$ for all $x, y \in \supp F$, to ensure that these inner products yield edge probabilities and the resulting network is binary. We do not require this restriction for our results, and thus the above definition permits weighted edges, but we note that Definition~\ref{def:rdpg} recovers the binary setting as a special case.

Under the random dot product graph, each node in a network is associated with a latent vector, and these latent vectors characterize the propensities for pairs of vertices to form edges with one other. Specifically, nodes close to one another in latent space are more likely to form edges, and nodes far apart are unlikely to form edges. When nodes' latent positions cluster in the latent space, the result is that edges are more likely to form between nodes with similar latent characteristics. This manifests as homophily in the resulting network.

More concretely, and of import for practitioners, degree-corrected stochastic blockmodels (Definition~\ref{def:sbm}) are submodels of the random dot product graph in Definition~\ref{def:rdpg}. Mixed-membership stochastic blockmodels and overlapping stochastic blockmodels, with and without degree-correction, and with Poisson or Bernoulli edges are also special cases of our model in Definition~\ref{def:rdpg} \citep{latouche2011,airoldi2008,jin2024a,zhang2020b, rohe2023}.

The key feature of random dot product graphs is that the latent positions $\X$ and the network $\A$ are highly dependent on one another. Thus, if the latent positions $\X$ are incorporated into a linear-in-means models as nodal covariates, Lemma~\ref{lem:rdpg:bramoulleconverge} in the Appendix states that the neighborhood averages $\G \X$ and $\G \Y$ will not converge to constants, and there is thus the potential for $(\alpha, \beta, \gamma, \delta)$ to avoid the asymptotic collinearity issue highlighted in Theorem~\ref{lem:indepcov:nonid} and Theorem~\ref{thm:LIM:minimax}. Theorem~\ref{thm:rdpg} shows, however, that some regression terms are still collinear in the asymptotic limit.

\begin{theorem} \label{thm:rdpg}
  Suppose that $(\A, \X)$ are sampled from a random dot product graph where $\X \in \R^{n \times d}$ is rank $d$ with probability $1$.
  Let $\bm \varepsilon$ be a vector of mean zero, i.i.d.~$(\nueps,\beps)$-subgamma random variables, with $(\nueps,\beps)$ not depending on $n$,
  and let
  \begin{equation}
    \label{eq:rdpg-eq}
    \Y = \alpha \onevec_n + \beta \G \Y + \X \bm \gamma + \G \X \bm \delta + \bm \varepsilon
  \end{equation}
  for $\alpha, \beta \in \R$ and $\gamma, \delta \in \R^d$.
  Suppose that the conditions of Proposition~\ref{prop:martellioso2022} hold and that $\X$ has $k \ge 2d$ distinct rows. Let $\W_n = \begin{bmatrix} \onevec_n \; \G \Y \; \X \; \G \X \end{bmatrix}$.
  Then, under Assumptions~\ref{assum:growth:sparsity},~\ref{assum:Fsparse:interact},~\ref{assum:F:extremes:norho}~and~\ref{assum:F:momentratio}, deferred to the Appendix for space considerations, $\W_n^\top \W_n / n$ converges to a limit $\bm \Sigma \in \R^{2(d +1) \times 2(d+1)}$, where $\bm \Sigma$ has rank $2d$. In particular, the columns of the design matrix corresponding to $\onevec_n, \G \Y$ and $\G \X$ are asymptotically collinear.
\end{theorem}

A proof is given in Appendix~\ref{app:partial-id-details}. Informally, Assumptions~\ref{assum:growth:sparsity},~\ref{assum:Fsparse:interact},~\ref{assum:F:extremes:norho}~and~\ref{assum:F:momentratio}  state that the network is fairly  dense, that the expected minimum degree of the network is growing, that no individual latent position $\X_i$ has an overly large spectral norm, and that the $\X_i$ have a finite second moment. These assumptions are relatively standard for the random dot product graph literature.

Theorem~\ref{thm:rdpg} states that the linear-in-means model with latent positions as nodal covariates is typically asymptotically well-behaved provided that (1) there is a mild amount of variation in $\X$, which will induce degree heterogeneity, and (2) two of the intercept, contagion or interference columns are dropped. Two entries of $(\alpha, \beta, \delta_1, \delta_2, \dots, \delta_d)$ need to be set to zero in the data generating process to achieve a non-singular design matrix in the limit, since $\G \X$ is asymptotically collinear with both the intercept column $\onevec_n$ and the contagion column $\G \Y$.

\begin{remark}[Rotational ambiguity of $\gamma$]
  In random dot product graphs, the latent positions $\X$ are identified only up an unknown rotation,	such that $\bm \gamma$ is identified only up to unknown rotation. However, the joint contribution $\X \bm \gamma$ is fully identified, so this is rarely an issue unless one wants to specifically interpret the coefficients $\bm \gamma$ \citep{athreya2018}. In that case, varimax rotation can often identify the coefficients \citep{hayes2025, rohe2023}.
\end{remark}

\begin{example}[Block randomization in a stochastic blockmodel without degree correction]
  Suppose that the network $\A$ follows a stochastic blockmodel without degree correction, and the nodal covariate $\T$ is a one-hot representation of the block membership $z_i$ of each node $i$. Omit an intercept column, since $\onevec_n$ and $\T$ are collinear for any $n$. Under the assumptions of Theorem~\ref{thm:rdpg}, $\onevec_n, \G \T$ and $\G \Y$ are asymptotically collinear. Thus, block or cluster randomization in experiments can still lead to collinearity issues.
\end{example}

The $\G \X$ term in Equation~\ref{eq:rdpg-eq} is somewhat stylized, but various authors have considered the linear-in-means model on random dot product graphs when $\bm \delta = 0$. To handle the fact that $\X$ is latent and unobserved, the typical approach is to estimate $\X$ via the adjacency spectral embedding \citep{sussman2014,athreya2018} and use the estimated latent positions $\Xhat$ as plug-in replacements for $\X$ \citep{hayes2023,mcfowland2021,le2022a}.

\begin{example}[Network autoregression in degree-corrected stochastic blockmodels]
  \label{ex:nar-rdpg}
  Suppose that the network $\A$ follows a stochastic blockmodel with degree correction, and outcomes come from the following regression
  \begin{equation*}
    \Y = \alpha \onevec_n + \beta \G \Y + \X \bm \gamma_{\mathrm{x}} + \T \bm \gamma_{\mathrm{t}} + \bm \varepsilon,
  \end{equation*}
  and there is no interference (i.e., $\delta = 0$). Suppose that the assumptions of Theorem~\ref{thm:rdpg} hold and let $\W_n = \begin{bmatrix} \onevec_n \; \G \Y \; \X \; \T \end{bmatrix}$ be the design matrix. Then $\onevec_n, \T, \X$ and $\G \Y$ are linearly independent for every $n$ and $\W_n^\top \W_n / n$ converges to a non-singular limit, such that there is no asymptotic collinearity. This is because the latent positions $\X$ and the degree-normalized adjacency matrix $\G$ are dependent, such that $\G \Y$ does not converge to a column vector of constants. Briefly, the $\G \Y$ terms converges to a vector in the span of $\D^{-1} \X$ and $\onevec_n$, and so there will not be any asymptotic collinearity provided that $\D^{-1} \X$ and $\X$ are linearly independent, which is the case when there is sufficient degree heterogeneity in the network. In forthcoming work, we show several methods to estimate this model.

  \cite{paul2022} uses this model (albeit with the latent positions $\X = \U \S^{1/2}$ replaced by the unscaled eigenvectors $\U$) to investigate recidivism in a criminal context, and \cite{paul2024} proposes it as a method to account for spatial autocorrelation. Theorem~\ref{thm:rdpg} thus augments the identification theory of \cite{martellosio2022} for this model, clarifying when it may be subject to estimability issues. In particular, Theorem~\ref{thm:rdpg} shows that the estimator \cite{paul2024} considers in their simulation study is asymptotically full rank. The simulation study thus avoids estimability issues due to asymptotic collinearity.

  The details are intriguing, as the precise specification of the simulation model is somewhat non-standard. \cite{paul2024} nominally simulate from stochastic blockmodel without degree correction. Per Theorem~\ref{thm:rdpg}, the lack of degree correction would thus imply that simulations are subject to asymptotic collinearity, and thus potentially subject to slower than parametric convergence rates. However, \cite{paul2024} specify their stochastic blockmodel via a rank deficit mixing matrix $\B$, in contrast to the more typical choice of a full rank mixing matrix $\B$. More precisely, \cite{paul2024} specify a $\B \in [0, 1]^{4 \times 4}$ where $\rank \B = 2$. This vanilla stochastic blockmodel is equivalent to rank $d = 2$ degree corrected stochastic blockmodel where there are two distinct degree correction parameters in each block. In this rank $d = 2$ specification, there are $k = 2d$ distinct rows in $\X$, inducing exactly the minimum amount of degree heterogeneity needed in Theorem~\ref{thm:rdpg} to avoid asymptotic collinearity.
\end{example}

\begin{example}
  \cite{bhadra2025} consider models of the form
  \begin{equation*}
    Y = \G \Y \beta + \T \gamma + \G \T \delta + \bm \varepsilon
  \end{equation*}
  omitting an intercept term out of concern for asymptotic collinearity. Additionally, they assume that the network $\A$ is generated from a low-rank network model. Since the model for $\Y$ does not include $\X$, Theorems~\ref{lem:indepcov:nonid}~and~\ref{thm:LIM:minimax} apply and provide lower bounds on estimation error for $\beta$ and $\delta$. Since $\A$ follows Definition~\ref{def:rdpg}, $\dmin = \Omega(n \rho_n)$, and then $\beta$ and $\delta$ are, in the worst case, at best estimable at $\rho_n^{-1/2}$ rates (recall that $\rho_n \in [0, 1]$, with $\rho_n \to 0$ inducing sparsity).
\end{example}

\begin{example}
  \cite{li2022f} consider estimation of direct and indirect effects in graphons (which generalize random dot product graphs) under the assumption that
  \begin{equation*}
    Y = \alpha \onevec_n + \T \gamma + f(\G \T) + \bm \varepsilon,
  \end{equation*}
  where $f$ is a function with three bounded derivatives. In this model, Theorems~\ref{lem:indepcov:nonid}~and~\ref{thm:LIM:minimax} together with the continuity of $f$ imply a form of asymptotic collinearity, since $\G \T \to \onevec_n \pi$ implies $f(\G \T) \to f(\onevec_n \pi)$. This makes it possible to construct a point estimate of $\gamma$ without explicitly estimating $f$, although variance estimation is still challenging.
\end{example}

\section{Simulation study on finite-sample consequences of asymptotic collinearity} \label{sec:simulations}

We illustrate the consequences of Theorems~\ref{lem:indepcov:nonid},~\ref{thm:LIM:minimax}~and~\ref{thm:rdpg} via simulation. Our simulations demonstrate that, even though our results are asymptotic, estimators are negatively affected by multicollinearity in finite samples. This is the case even though all the conditions of Proposition~\ref{prop:martellioso2022} are satisfied such that $\alpha, \beta, \delta$ and $\gamma$ are identified for finite $n$.

All networks in our simulations below are generated from a Poisson degree-corrected stochastic blockmodel (Definition~\ref{def:sbm}) with $n$ nodes and four equally probable blocks.
The edge formation matrix $\B \in [0,1]^{4 \times 4}$ is taken to be
\begin{equation*}
  \B =
  \begin{bmatrix}
    0.5  & 0.05 & 0.05 & 0.05 \\
    0.05 & 0.5  & 0.05 & 0.05 \\
    0.05 & 0.05 & 0.5  & 0.05 \\
    0.05 & 0.05 & 0.05 & 0.5  \\
  \end{bmatrix},
\end{equation*}
and the sparsity parameter $\rho_n$ is set so that the expected mean degree of the network is $2 n^{0.7}$. At this density level, Theorem~\ref{thm:LIM:minimax} states that no consistent estimators of $\alpha, \beta$ and $\delta$ exist, if nodal covariates $T_i$ are independent on the network structure. Degree-correction parameters $\xi_1, \xi_2,\dots, \xi_n$ are sampled independently from a continuous uniform distribution supported on the interval $[1, 2]$. All parameters in the simulation models are identified.

We consider three distinct generative models for nodal outcomes $Y_1,Y_2,\dots, Y_n$. In the \emph{Bernoulli} model, there is a single nodal covariate $T_i \sim \Bern(0.5)$ sampled independently for all nodes, and independently of the network. The regression model is then
\begin{equation*}
  Y = \alpha \onevec_n + \beta \G \Y + \gamma \T + \delta \G \T + \bm \varepsilon,
\end{equation*}
and we fix $\alpha = 3, \beta = 0.2, \gamma = 4, \delta = 2$ and sample the entries of $\bm \varepsilon$ independently and identically according to a normal with mean zero and standard deviation $\sigma = 0.1$. Per our theoretical results, the columns of $\W_n$ corresponding to $\alpha, \beta$ and $\delta$ are asymptotically collinear in the Bernoulli model.

In the \emph{Unrestricted} model, the nodal covariates are the latent positions of the stochastic blockmodel. That is, $T_i =\X_i \in \R^4$ for all $i \in [n]$ where $\X = \U \S^{1/2}$ and $\U \S \U^\top$ is the eigendecomposition of $\E[z_1,z_2,\dots, z_n, \theta]{\A}$ (recall that $z_i$ is the block membership of node $i$; see Definition~\ref{def:sbm}). The nodal regression model is thus
\begin{equation*}
  Y = \alpha \onevec_n + \beta \G \Y + \X \bm \gamma + \G \X \bm \delta + \bm \varepsilon,
\end{equation*}
where we again set $\alpha = 3, \beta = 0.2$ and sample $\varepsilon_i \sim \mathcal N(0, \sigma^2)$ independently with $\sigma = 0.1$. Since $\X_i \in \R^4$, we have $\bm \gamma, \bm \delta \in \R^4$ and we fix $\bm \delta = (2, 2, 2, 2)$ and $\bm \gamma = (1.5, 2.5, 3.5, 4.5)$. Once again, our theory predicts that $\alpha, \beta$ and $\bm \delta$ are inestimable in this model.

Finally, the \emph{Restricted} model is the same as the \emph{Unrestricted} model, but with the additional constraint that $\bm \delta = (0, 0, 2, 2)$, such that the design matrix is full rank in the limit and there is no asymptotic collinearity, as per Theorem~\ref{thm:rdpg}.

In our experiments that follow, we vary the sample size $n$ (i.e., the number of vertices) on a logarithmic scale, considering $n \in \set{100, 163, 264, 430, 698, 1135, 1845, 3000}$, and replicate our experiments $100$ times for each simulation setting. In each setting, we estimate the parameters of the linear-in-means model using ordinary least squares \citep{lee2002, trane2023}, two-stage least squares \citep{kelejian1998,lee2003,piras2022}, and a quasi-maximum likelihood estimator \citep{lee2004, nath2023}.

Figure~\ref{fig:mse} shows the median squared error of the estimated coefficients as a function of the number of nodes. All estimators that we consider fail to recover the regression coefficients that correspond to asymptotically collinear columns of the design matrix. Median squared error decreases for all other coefficients (gray), as expected. The median squared error for asymptotically collinear coefficients (colored) either increases or remains approximately constant as sample size grows. This agrees with our theoretical results indicating that there is no consistent estimator for the asymptotically collinear coefficients. These simulations suggest that ordinary least squares, two-stage least squares, and quasi-maximum likelihood estimators are inconsistent for regression coefficients under asymptotic collinearity, or that at least that they fail to achieve expected parametric rates.

\begin{figure}[t]
  \centering
  \includegraphics{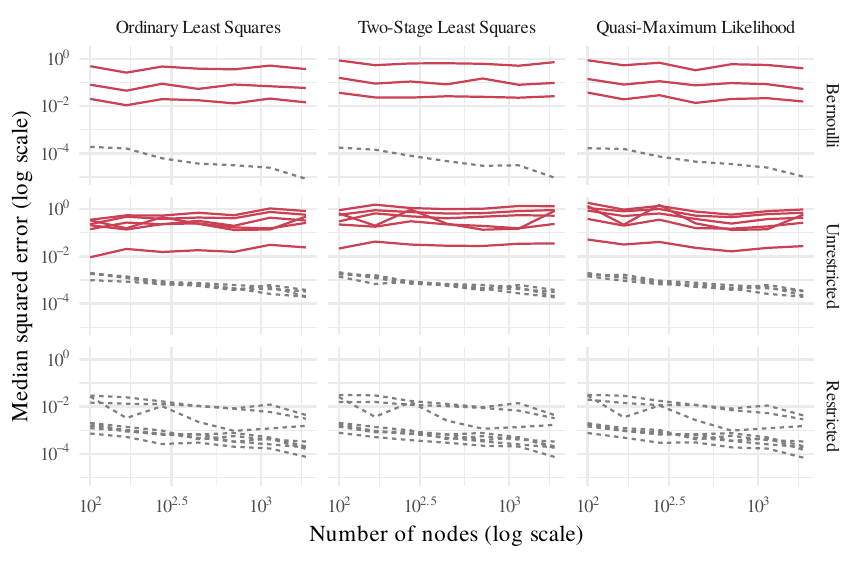}
  \caption{Median squared error of estimates. Each row of panels denotes a different simulation setting, and each column of panels corresponds to a different estimator. Within a panel, the x-axis represents the sample size on a log scale, and the y-axis represents the Monte Carlo estimate of median squared error, also on log scale. Each line corresponds to a single coefficient. Solid red lines are asymptotically collinear, dashed gray lines are not.}
  \label{fig:mse}
\end{figure}

Figure~\ref{fig:vif} demonstrates that the columns of the design matrix become more and more collinear as sample size increases. To measure collinearity, we use variance inflation factors. The variance inflation factor for the $j$-th estimated coefficient is defined as $1 / (1 - R_j^2)$ where $R_j$ is the coefficient of determination in the regression where the $j$-th covariate is predicted based on all other covariates, using ordinary least squares. It is typically considered inappropriate to compute variance inflation factors for the intercept, but here we are not interested in the variance multiplier interpretation of variance inflation factors, but rather a simple metric of collinearity \citep{fox1992}.

\begin{figure}[t]
  \centering
  \includegraphics{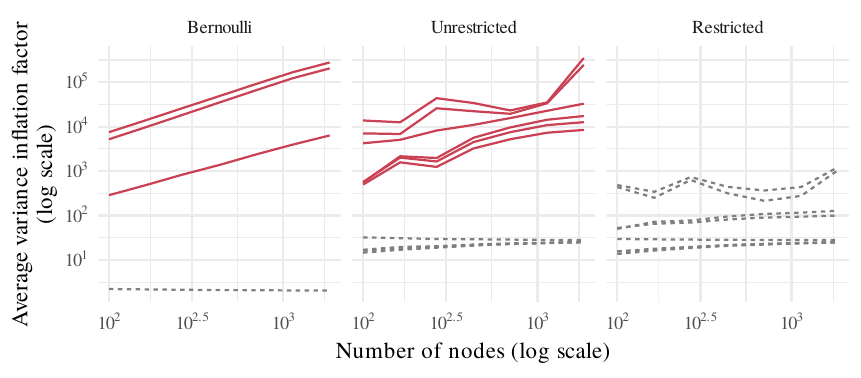}
  \caption{Average variance inflation factors. Each panel denotes a different simulation setting. The x-axis represents the sample size on a log scale, and the y-axis represents the Monte Carlo estimate of mean variance inflation factor, also on log scale. Each line corresponds to a single coefficient. Solid red lines are asymptotically collinear, dashed gray lines are not.}
  \label{fig:vif}
\end{figure}

Examining Fig.~\ref{fig:vif}, we see that in the \emph{Bernoulli} model, the contagion and interference terms are converging to constant multiples of the intercept, and in the \emph{Unrestricted} model, the intercept and contagion term are becoming closer and closer to the span of the interference terms. In the \emph{Restricted} model, where there is no asymptotic collineairty, there is still some collinearity of the intercept, interference and contagion terms, but not as much as in the \emph{Unrestricted} model. The variance inflation factors in the \emph{Restricted} model are either constant or growing slowly as a function of $n$. In the \emph{Bernoulli} model, the variance inflation factors for the intercept $\alpha$, interference effect $\delta$, and contagion effect $\beta$ are orders of magnitude larger than in the \emph{Restricted} model and increasing multiplicatively with sample size. In the \emph{Unrestricted} model, the variance inflation factor for the unidentified coefficients are also growing rapidly as a function of sample size. These results are exactly as expected based on Theorems~\ref{thm:LIM:minimax} and~\ref{thm:rdpg}.

\section{Discussion}

Collinearity has long plagued linear-in-means models. Historically, this was primarily understood as an identification problem, which \cite{manski1993} termed the \emph{reflection problem}. The reflection problem is generally considered resolved, as some relatively mild conditions on the network structure are sufficient to guarantee identifiability, or lack of perfect collinearity. Our results show that linear-in-means models can be subject to collinearity issues, even when the collinearity is not perfect. Identifying conditions are necessary but not sufficient to ensure reliable estimation of peer effects.

Whenever the minimum degree of a network diverges and nodal covariates are independent of network structure, the minimax estimation rate slows down, for all estimators, potentially nullifying coverage guarantees of confidence intervals (Theorem~\ref{thm:LIM:minimax}). In dense networks, estimates can become inconsistent. This issue occurs even in randomized experiments on networks, which may be surprising as random experiments are typically a setting where estimators obtain parametric rates of convergence \citep{aronow2025a}. The intercept, contagion and interference coefficients are potentially subject to this inconsistency, but direct effects remain estimable.

This parametric result mirrors a number of non-parametric results for network experiments, which show certain estimators may have infinite variance under Bernoulli designs \citep{basse2018b}. Bernoulli designs may yield very little information about contagion and interference effects, and thus modern designs for network experiments explicitly take network structure into account when assigning treatment \citep{ugander2013, saint-jacques2019, kandiros2024, viviano2024}. It may be possible to avoid asymptotic collinearity issues by analyzing only direct effects \citep{li2022c}, or by obtaining repeated measurements of the same network \citep{yu2022a}.

In practice, many network analyses consider sparse networks, and it may be reasonable to assume that node degrees are bounded. In these settings, our results do not necessarily indicate an identifiability-estimability gap, and so the linear-in-means model is still a reliable tool for inference. However, our results emphasize the importance of bounded degrees.

Bounded degrees are especially salient for theorists developing new estimators under graphons and random dot product graphs. Most asymptotic results for graphons and random dot product graphs rely on exogeneous networks becoming increasingly dense and interconnected as sample size increases. This is exactly the setting where we have shown the possibility of identifiability-estimability gaps. For instance, our results suggest that the estimators proposed in \cite{bhadra2025} may have infinite variance in the asymptotic limit. Similarly, our results characterize when the estimators proposed in \cite{paul2022a} may be subject to an identifiability-estimability gap. As an increasing number of methods in the causal inference literature leverage structure in random dot product graphs to improve estimation efficiency \citep[for instance,][]{li2022f}, we emphasize that this modeling choice comes with important tradeoffs and the potential for asymptotically vanishing signal-to-noise ratios.

Modeling nodal covariates as endogenous to network structure may prevent asymptotically vanishing signal-to-noise ratios. As we have shown in Theorem~\ref{thm:rdpg}, dependence between covariates and network structure can prevent asymptotic degeneracies in the design matrix of linear-in-means models. From a conceptual standpoint, the fact that homophily can improve estimability is quite interesting. From a pragmatic standpoint, though, our results are more mixed, and homophily only sometimes prevents asymptotic collinearity issues. In stochastic blockmodels, peer effects that are functions of block membership are asymptotically degenerate. Some amount of degree heterogenity in a degree-corrected stochastic blockmodels is necessary to avoid to this degeneracy. This indicates that dependence between nodal covariates and network structure is not a silver bullet, and estimability issues remain possible under homophily. In observational analyses of network data, we recommend explicitly modeling the dependence between nodal characteristics and network structure. Unfortunately, to the best of our knowledge, different dependence relations between covariates and networks need to be individually evaluated for potential collinearity problems.

\section{Disclosure statement} \label{disclosure-statement}

The authors declare that they have no conflicts of interest.

\section{Data Availability Statement}\label{data-availability-statement}

\if1\anon
  {
    A replication package for our simulations is available at \url{https://github.com/alexpghayes/asymptotic-collinearity-replication}.
  } \fi

\if0\anon
  {
    A replication package for our simulations will be posted to Github, but is not linked here for the sake of blinded peer review. We have attended a copy of the replication package as a supplemental file to our submission.
  } \fi

\newpage
\appendix

\section{Example of model that is identified but not estimable}
\label{app:toy-id-example}

\begin{example}
  Suppose that we we have a collection of independent and identically distributed observations $Y_i =\X_i \theta + \varepsilon_i$ where $\varepsilon_i \sim \mathcal N(0, \sigma^2)$. In linear Gaussian models, identification (Definition \ref{def:identified}) can be characterized by several equivalent conditions \citep{lewbel2019}: (1) the matrix $\X$ has full-rank (i.e., there is no perfect collinearity); (2) the covariance matrix $\X^\top \X / n$ has full-rank, (3) the log-likelihood
  \[
    -\frac{n}{2}\log(2\pi\sigma^2) - \frac{1}{2\sigma^2}\sum_{i=1}^n(y_i -\X_i\theta)^2
  \]
  \noindent has a unique maximizer.

  As a concrete of example of a model that is identified, asymptotically collinear and inestimable, consider the following linear regression model, where all predictors are the same except for the first observation:
  \begin{equation*}
    \begin{bmatrix}
      Y_1    \\
      Y_2    \\
      Y_3    \\
      \vdots \\
      Y_n
    \end{bmatrix}
    =
    \begin{bmatrix}
       & 1      & 2      & \\
       & 1      & 1      & \\
       & 1      & 1      & \\
       & \vdots & \vdots & \\
       & 1      & 1      &
    \end{bmatrix}
    \begin{bmatrix}
      \alpha \\
      \beta
    \end{bmatrix}
    +
    \begin{bmatrix}
      \varepsilon_1 \\
      \varepsilon_2 \\
      \varepsilon_3 \\
      \vdots        \\
      \varepsilon_n
    \end{bmatrix}.
  \end{equation*}
  Since the design matrix is full-rank, $\alpha$ and $\beta$ are identified, but since there is only a single observation that differentiates $\alpha$ from $\beta$ (namely, the first row of the design matrix), it will be impossible to estimate the coefficients even when the sample size diverges to infinity. %
\end{example}

\section{Linear Independence Under Random Dot Product Graphs}
\label{app:proof:prop:three-eig}

As mentioned in the main text, $\I, \G$ and $\G^2$ are linearly independent under a broad class of network models.
Proposition~\ref{prop:three-eig} establishes this for the random dot product graph.

\begin{proposition} \label{prop:three-eig}
  If $\G$ has three or more distinct eigenvalues, then $\I, \G$ and $\G^2$ are linearly independent.
\end{proposition}
\begin{proof}
  Let $c_1,c_2,c_3 \in \R$ be such that $c_1 \I + c_2 \G + c_3 \G^2 = 0$. The characteristic polynomial of this matrix must have $n$ roots all equal to zero. It follows that for any eigenvalue $\lambda$ of $\G$, we must have $c_1 + c_2 \lambda + c_3 \lambda^2 = 0$. This quadratic has at most two distinct solutions for $c_1,c_2,c_3$ not all zero. Since $\G$ has more than two distinct eigenvalues, it follows that we must have $c_1=c_2=c_3=0$.
\end{proof}

\section{Concentration inequalities and moment bounds} \label{app:subgamma}
Here we collect a few technical results regarding moment bounds and concentration inequalities to be used in our main results.

\begin{definition}[\citealt{boucheron2013}]
        Let $\Z$ be a mean-zero random variable with cumulant generating function $\psi_Z(t) = \log \E{e^{t Z}}$.
        $\Z$ is \emph{subgamma} with parameters $\nu \ge 0$ and $b \ge 0$ if
        \begin{equation*}
                \psi_Z(t) \le \frac{t^2 \nu}{2 (1 - b t)}
                ~\text{ and }~
                \psi_{-Z}(t) \le \frac{t^2 \nu}{2 (1 - b t)}
                ~\text{ for all }~ t < 1 / b.
        \end{equation*}
        If this is the case, we will often simply write that $\Z$ is $(\nu,b)$-subgamma.
\end{definition}

\begin{lemma}[\cite{boucheron2013} Chapter 2] \label{lem:sgbasic}
        Suppose that $\Z$ is a $(\nu,b)$-subgamma random variable.
        Then for all $ t > 0$,
        \begin{equation*}
                \Pr\left[ |Z| > \sqrt{2\nu t} + bt \right] \le \exp\{ -t \} .
        \end{equation*}
\end{lemma}

The following are basic results concerning subgamma random variables, which we prove for the sake of completeness.

\begin{lemma} \label{lem:sgmax}
        Let $\Z_1,Z_2,\dots,Z_n$ be a collection of independent $(\nu,b)$-subgamma random variables and let $c > 0$ be a constant.
        Then it holds with probability at least $1-Cn^{-c}$ that
        \begin{equation*}
                \max_{i \in [n]} |Z_i| \le C\sqrt{ \nu + b^2 } \log n,
        \end{equation*}
\end{lemma}
\begin{proof}
        For $t \ge 0$, applying a union bound followed by Lemma~\ref{lem:sgbasic},
        \begin{equation*}
                \Pr\left[ \max_i |Z_i| > \sqrt{2\nu t} + bt \right]
                \le \sum_{i=1}^n \Pr[ |Z_i| > \sqrt{2\nu t} + bt ]
                \le n \exp\{ -t \}.
        \end{equation*}
        Taking $t = C \log n$ for $C>0$ chosen suitably large, it holds that with probability at least $1-n^{-c}$,
        \begin{equation*}
                \max_i |Z_i|
                \le \sqrt{2C \nu \log n} + Cb\log n
                \le C\sqrt{ \nu + b^2 } \log n,
        \end{equation*}
        as we set out to show.
\end{proof}

\begin{lemma} \label{lem:sgsum}
        Let $\Z_1,Z_2,\dots,Z_n$ be a collection of independent $(\nu,b)$-subgamma random variables and let $\alpha_1,\alpha_2,\dots,\alpha_n \in \R$ be nonnegative.
        Then, defining $S_n = \sum_i \alpha_i Z_i$, for any $t > 0$, for any constant $c > 0$, it holds with probability at least $1 - 2n^{-c}$ that
        \begin{equation} \label{eq:firstclaim}
                \left| S_n \right| \le
                C (\nu^{1/2} + b)\left( \sum_{i=1}^n \alpha_i^2 \right)^{1/2} \log n
        \end{equation}
        and
        \begin{equation*}
                \left| S_n \right|
                = O\left(  (\nu^{1/2} + b)\left( \sum_{i=1}^n \alpha_i^2 \right)^{1/2} \log n
                \right)~~~\text{ almost surely.}
        \end{equation*}
\end{lemma}
\begin{proof}
        By a basic property of subgamma random variables \citep[see][Chapter 2]{boucheron2013}, $\alpha_i Z_i$ is $(\alpha_i^2 \nu, \alpha_i b)$-subgamma, and thus the random sum $S_n = \sum_i \alpha_i Z_i$ is a subgamma random variable with parameters
        \begin{equation*}
                \nubar = \nu \sum_{i=1}^n \alpha_i^2
                ~~~\text{ and }~~~
                \bbar  = b \max_{i \in [n]} \alpha_i
                \le b \sqrt{ \sum_{i=1}^n \alpha_i^2 }.
        \end{equation*}
        Thus, applying Lemma~\ref{lem:sgbasic}, for any $t > 0$,
        \begin{equation*}
                \Pr\left[ |S_n| > \sqrt{2\nubar t} + \bbar t \right] \le \exp\{ -t \}.
        \end{equation*}
        Taking $t= C\log n$ for suitably large $C > 0$ and noting that $\bbar \le C \nubar^{1/2}$ for suitably-chosen constant $C > 0$, it follows that
        \begin{equation*}
                \Pr\left[ |S_n| > C( \nubar^{1/2} \log^{1/2} n + \bbar \log n \right]
                \le 2n^{-c}.
        \end{equation*}
        Observing that
        \begin{equation*}
                \nubar^{1/2} \log^{1/2} n + \bbar \log n
                \le (\nu^{1/2} + b)\left( \sum_{i=1}^n \alpha_i^2 \right)^{1/2} \log n
        \end{equation*}
        establishes Equation~\eqref{eq:firstclaim}.
        Taking $c > 1$ and applying the Borel-Cantelli lemma then implies that
        \begin{equation*}
                |S_n| = O\left(  \nubar^{1/2} \log^{1/2} n + \bbar \log n \right)
                ~~~\text{ almost surely,}
        \end{equation*}                                                                 which completes the proof.
\end{proof}

\begin{lemma} \label{lem:vectorbernstein}
        Let $\X_1,\X_2,\dots,\X_n \in \R^d$ and let $\xi_1,\xi_2,\dots,\xi_n$ be conditionally independent $(\nu,b)$-subgamma random variables given $\X_1,\X_2,\dots,\X_n$.
        Then with probability at least $1-2n^{-3}$,
        \begin{equation*}
                \left\| \sum_{j=1}^n \xi_j \X_j \right\|
                \le
                C \sqrt{ \sum_{j=1}^n \| \X_j \|^2 } \sqrt{ \nu + b^2 } \log n.
        \end{equation*}
\end{lemma}
\begin{proof}
        Define the random vector $\xitilde \in \R^d$ according to
        \begin{equation*}
                \xitilde_k = \sum_{j=1}^n \xi_j \X_{jk}
                ~\text{ for }~ k \in [d] .
        \end{equation*}
        We observe that conditional on $\X_1,\X_2,\dots,\X_n$, $\xitilde_k$ is a sum of subgamma random variables.
        Applying Bernstein's inequality \citep[][Corollary 2.11]{boucheron2013}, for $t > 0$,
        \begin{equation*}
                \Pr\left[ |\xitilde_k| > t \right]
                \le
                2\exp\left\{ \frac{ -t^2 }{ 2(\nu_{ik} + b_{ik} t) } \right\},
        \end{equation*}
        where
        \begin{equation*}
                \nu_{ik} = \nu \sum_{j=1}^n \X_{jk}^2~~~\text{ and }~~~
                b_{ik} = b \max_{j \in [n]} |\X_{jk}|.
        \end{equation*}
        Taking $t= C\sqrt{\nu_{ik}+b_{ik}^2} \log n$ for $C>0$ chosen suitably large,
        it holds with probability at least $1-2n^{-3}$ that
        \begin{equation*} \begin{aligned}
                        | \xitilde_k |
                         & \le
                        C\left( \nu \sum_{j=1}^n \X_{jk}^2
                        + b^2 \max_{j\in[n]} |\X_{jk}|^2 \right)^{1/2} \log n \\
                         & \le
                        C \| \X_{\cdot k} \|^2 \sqrt{ \nu + b^2} \log n.
                \end{aligned} \end{equation*}
        A union bound over all $j \in [d]$ followed by taking square roots implies that with probability at least $1-2n^{-3}$,
        \begin{equation*}
                \| \xitilde \|
                \le
                C \sqrt{ \sum_{j=1}^n \| \X_j \|^2 } \sqrt{ \nu + b^2 } \log n,
        \end{equation*}
        completing the proof.
\end{proof}

\begin{lemma} \label{lem:Gsubgamma}
        Let $\A \in \R^{n \times n}$ be the adjacency matrix of a network with degree matrix $\D \in \R^{n \times n}$, and let $\bm \varepsilon \in \R^n$ be a vector of independent mean zero $(\nu,b)$-subgamma random variables with $\bm \varepsilon$ independent of $\A$.
        Letting $\G = \D^{-1} \A$, it holds with high probability that
        \begin{equation*}
                \max_{i \in [n]} \left| \left[ \G \bm \varepsilon \right]_i \right|
                \le
                C \max_{i \in[n] } \sqrt{ \nu_i \log^2 n } + C \max_{i \in [n]} b_i \log n,
        \end{equation*}
        where
        \begin{equation} \label{eq:def:nuibi}
                \nu_i = \nu \sum_{j=1}^n \frac{ \A_{ij}^2 }{ d_i^2 }
                ~~~\text{ and }~~~
                b_i = b \max_{j \in [n]} \frac{ \A_{ij} }{ d_i }.
        \end{equation}
\end{lemma}
\begin{proof}
        Unrolling the definition, for $i \in [n]$,
        \begin{equation*}
                \left[ \G \bm \varepsilon \right]_i
                = \frac{1}{d_i} \sum_{j=1}^n \A_{ij} \varepsilon_j,
        \end{equation*}
        which is a weighted sum of subgamma random variables.
        Lemma~\ref{lem:sgsum}, conditional on $\A$, implies that for $t > 0$,
        \begin{equation*}
                \Pr\left[ \left| \left[ \G \bm \varepsilon \right]_i \right|
                        > t \mid \A \right]
                \le
                2\exp\left\{ \frac{ -t^2 }{ 2(\nu_i + b_i t ) } \right\},
        \end{equation*}
        where $\nu_i$ and $b_i$ are as defined in Equation~\eqref{eq:def:nuibi}.
        Taking $t = C \sqrt{ \nu_i + b_i^2 } \log n$ for $C > 0$ suitably large, it holds with probability at least $1 - 2n^{-3}$ that
        \begin{equation*}
                \left| \left[ \G \bm \varepsilon \right]_i \right|
                \le
                C\sqrt{ \nu_i + b_i^2 } \log n.
        \end{equation*}
        Noting that $C$ can be chosen independently of the index $i$, a union bound over all $i \in [n]$ implies that with probability at least $1-2n^{-2}$,
        \begin{equation*}
                \max_{i \in [n]} \left| \left[ \G \bm \varepsilon \right]_i \right|
                \le
                C \max_{i \in[n] }\sqrt{ \nu_i + b_i^2 } \log n
                \le
                C \max_{i \in[n] } \sqrt{ \nu_i \log^2 n }
                +
                C \max_{i \in [n]} b_i \log n,
        \end{equation*}
        as we set out to show.
\end{proof}

\begin{lemma} \label{lem:Zquad:conc}
        Suppose that for all $n \ge 1$, $\Z \in \R^n$ is a vector of independent $\nu_Z$-subgaussian random variables with $\nu_Z$ constant with respect to $n$ and $\M = \\M_n \in \R^{n \times n}$ is a (possibly random) matrix such that $\Z$ is independent of $\M$.
        Then
        \begin{equation*}
                \left| \Z^\top \M \Z - \bbE \Z^\top \M \Z \right|
                = \Op{ \| \M \|_F }.
        \end{equation*}
\end{lemma}
\begin{proof}
        By the Hanson-Wright inequality \citep{rudelson2013,vershynin2020},
        \begin{equation*}
                \Pr\left[ \left| \Z^\top \M \Z - \bbE \Z^\top \M \Z \right|
                        > t \right]
                \le 2\exp\left\{ -c \min \left\{ \frac{ t^2 }{ \nu_Z^2 \| \M \|_F^2 } ,
                \frac{ t }{ \nu_Z \| \M \| } \right\} \right\}.
        \end{equation*}
        Setting $t = C\nu_Z \| \M \|_F$, we have
        \begin{equation*}
                \Pr\left[ \left| \frac{1}{n} \Z^\top \M \Z - \frac{1}{n} \bbE \Z^\top \M \Z \right|
                        > C\nu_Z \| \M \|_F \right]
                \le 2\exp\left\{ -C \min \left\{ 1, \frac{ \| \M \|_F }{ \| \M \| } \right\}
                \right\}
                = 2\exp\{ -C \}.
        \end{equation*}
        Choosing $C>0$ suitably large makes this right-hand probability arbitrarily small, and it follows that
        \begin{equation*}
                \left| \Z^\top \M \Z - \bbE \Z^\top \M \Z \right| = \Op{ \| \M \|_F },
        \end{equation*}
        as we set out to show.
\end{proof}

\begin{lemma} \label{lem:Zquad:moment}
        Suppose that $\Z \in \R^n$ is a vector of independent mean-zero random variables with shared variance $\sigma_Z^2$ and shared fourth moment $\zeta_4$ and let $\M \in \R^{n \times n}$ be a fixed matrix with $\Z$ not depending on $\M$.
        Then
        \begin{equation} \label{eq:ZMZ:expec}
                \bbE \Z^\top \M \Z = \sigma_Z^2 \trace \M
        \end{equation}
        and $\bbE ( \Z^\top \M \Z )^2 \le C\zeta_4\left[ \| \M \|_F^2 + ( \trace \M )^2 \right]$.
\end{lemma}
\begin{proof}
        Since $\Z$ is mean zero, $\bbE \Z^\top \M \Z = \sigma_Z^2 \trace \M$, and Equation~\eqref{eq:ZMZ:expec} is immediate.

        Expanding the quadratic,
        \begin{equation*}
                \bbE ( \Z^\top \M \Z )^2
                = \sum_{i=1}^n \sum_{j=1}^n \sum_{k=1}^n \sum_{\ell=1}^n
                \M_{ij} \M_{k\ell} \bbE Z_i Z_j Z_k Z_{\ell}.
        \end{equation*}
        Noting that the expectations disappear unless $|\{i,j,k,\ell\}|\in\{2,4\}$, we have
        \begin{equation*}
                \bbE ( \Z^\top \M \Z )^2
                = \zeta_4 \sum_{i=1}^n \M_{ii}^2
                + \sigma_Z^4 \sum_{i \neq j} ( \M_{ii} \M_{jj} + \M_{ij}^2 + \M_{ij} \M_{ji} ).
        \end{equation*}
        Applying Jensen's inequality and using the fact that $2ab \le (a^2 + b^2)$,
        \begin{equation*}
                \bbE ( \Z^\top \M \Z )^2
                \le C\zeta_4\left( \| \M \|_F^2 + \sum_{i \neq j} \M_{ii} \M_{jj} \right)
                \le C\zeta_4\left[ \| \M \|_F^2 + ( \trace \M )^2 \right],
        \end{equation*}
        completing the proof.
\end{proof}

\begin{lemma} \label{lem:Zbilinear:moment}
        Suppose that $\Z, \bm \Ztilde \in \R^n$ are independent mean-zero random vectors.
        Let the entries of $\Z$ have independent entries with shared variance $\sigma_Z^2$.
        Similarly, let the entries of $\bm \Ztilde$ have shared variance $\sigma_{\Ztilde}^2$.
        Let $\M \in \R^{n \times n}$ be a fixed matrix with $\Z$ and $\bm \Ztilde$ not depending on $\M$.
        Then
        \begin{equation*}
                \bbE \left( \Z^\top \M \bm \Ztilde \right)^2
                = \sigma_Z^2 \sigma_{\Ztilde}^2 \| \M \|_F^2.
        \end{equation*}
\end{lemma}
\begin{proof}
        Expanding the quadratic,
        \begin{equation*}
                \bbE \left( \Z^\top \M \bm \Ztilde \right)^2
                = \sum_{i=1}^n \sum_{j=1}^n \sum_{k=1}^n \sum_{\ell=1}^n
                \bbE \M_{ij} \M_{k\ell} Z_i Z_k \Ztilde_j \Ztilde_\ell.
        \end{equation*}
        Using the independence structure of $\Z$ and $\Ztilde$,
        \begin{equation*}
                \bbE \left( \Z^\top \M \bm \Ztilde \right)^2
                = \sum_{i=1}^n \sum_{j=1}^n \M_{ij}^2 \sigma_Z^2 \sigma_{\Ztilde}^2,
        \end{equation*}
        completing the proof.
\end{proof}

\section{Spectral results} \label{apx:spectral}

Here we collect results related to the spectral properties of the adjacency matrix $\A$, its expected value $\P =\rho \X \X^\top$ and the random walk Laplacian $\G$.

\begin{lemma} \label{lem:XcalDinvX:spectral}
	Let $(\A, \X) \sim \RDPG( F,n )$ with sparsity parameter $\rho$. Then, defining
	\begin{equation} \label{eq:def:calD}
		\calD = \diag( \delta_1, \delta_2, \dots, \delta_n ) \in \R^{n \times n},
	\end{equation}
	we have
	\begin{equation*}
		\left\| \rho \X^\top \calD^{-1} \X \right\| = 1.
	\end{equation*}
\end{lemma}
\begin{proof}
	We note first that since for a matrix $\B$, the matrices $\B^\top \B$ and $\B \B^\top$ have the same non-zero eigenvalues.
	As such, $\| \B^\top \B \| = \| \B \B^\top \|$.
	In particular, taking $\B = \calD^{-1/2} \X$,
	\begin{equation} \label{eq:XcalDinvX:spec1}
		\left\| \rho \X^\top \calD^{-1} \X \right\|
		=
		\| \rho \calD^{-1/2} \X \X^\top \calD^{-1/2} \|.
	\end{equation}

	Now, suppose that $u \in \R^n$ is an eigenvector of $\rho \calD^{-1/2} \X \X^\top \calD^{-1/2}$ with eigenvalue $\lambda$, so that
	\begin{equation*}
		\rho \calD^{-1/2} \X \X^\top \calD^{-1/2} u = \lambda u.
	\end{equation*}
	Then we have
	\begin{equation*}
		\rho \calD^{-1} \X \X^\top (\calD^{-1/2} u)
		= \calD^{-1/2} \left( \rho \calD^{-1/2} \X \X^\top \calD^{-1/2} u \right)
		= \lambda \calD^{-1/2} u.
	\end{equation*}
	In particular, $\lambda$ is also an eigenvalue of $\rho \calD^{-1} \X \X^\top$, albeit with a different eigenvector.
	It follows that, since $\calD^{-1} \X \X^\top$ is row-stochastic and thus has largest eigenvalue $1$,
	\begin{equation*}
		\left\| \rho \calD^{-1/2} \X \X^\top \calD^{-1/2} \right\|
		=
		\left\| \rho \calD^{-1} \X \X^\top \right\| = 1.
	\end{equation*}
	Combining this with Equation~\eqref{eq:XcalDinvX:spec1} above completes the proof.
\end{proof}

The following lemmas collect a few basic facts about the random walk Laplacian $\G$.

\begin{lemma} \label{lem:G:spectrum}
	Let $\G = \D^{-1} \A$ be the transition matrix of a network with adjacency matrix $\A$ and degree matrix $\D$.
	The eigenvalues of $\G$ are all real, and $\| \G \| \le 1$.
\end{lemma}
\begin{proof}
	Since $\G = \D^{-1} \A$ is row stochastic, $\| \G \| \le 1$ follows from the Perron-Frobenius theorem.
	To see that all eigenvalues of $\G$ are real, let $\lambda \in \R$ be such that $\D^{-1/2} \A \D^{-1/2} u = \lambda u$ for some eigenvector $u \in \R^n$.
	Then, multiplying by $\D^{-1/2}$ on both sides,
	$\D^{-1} \A \D^{-1/2} u = \lambda \D^{-1/2} u$, so that $\D^{-1/2} u$ is an eigenvector of $\G$ with eigenvalue $\lambda$.
	Thus, $\G$ and $\D^{-1/2} \A \D^{-1/2}$ have the same spectrum.
	Since $\D^{-1/2} \A \D^{-1/2}$ is symmetric, all its eigenvalues are real.
\end{proof}

\begin{lemma} \label{lem:IbgyG:invertible}
	Let $\G = \D^{-1} \A$ be the transition matrix of a network with adjacency matrix $\A$ and degree matrix $\D$.
	If $|\beta| < 1$ then $\I - \beta \G$ is invertible and has all eigenvalues in the interval $1 \pm \beta$.
	Further,
	\begin{equation} \label{lem:IbgyG:infty}
		\left\| \left( \I - \beta \G \right)^{-1} \right\|_{\infty}
		\le \frac{1}{1-|\beta|}.
	\end{equation}
\end{lemma}
\begin{proof}
	By Lemma~\ref{lem:G:spectrum}, all eigenvalues of $\beta \G$ have absolute value at most $|\beta|$, from which all eigenvalues of $\I - \beta \G$ have absolute value in $[1-\beta, 1+\beta]$.
	In particular, the eigenvalues of $\I - \beta \G$ are bounded away from zero, ensuring that the matrix is invertible.

	To prove Equation~\eqref{lem:IbgyG:infty}, note that by the Neumann expansion and the triangle inequality,
	\begin{equation*}
		\left\| \left( \I - \beta \G \right)^{-1} \right\|_{\infty}
		\le \sum_{q=0}^\infty |\beta|^q \left\| \G^q \right\|_{\infty}.
	\end{equation*}
	Since $\G^q$ is row-stochastic for any $q=0,1,2,\dots$, we have $\| \G^q \|_\infty \le 1$ trivially, and thus
	\begin{equation*}
		\left\| \left( \I - \beta \G \right)^{-1} \right\|_{\infty}
		\le \sum_{q=0}^\infty |\beta|^q = \frac{ 1 }{ 1 - |\beta| },
	\end{equation*}
	completing the proof.
\end{proof}

\begin{lemma} \label{lem:Xspec}
	Let $(\A, \X) \sim \RDPG( F, n )$. Then
	\begin{equation*}
		\| \X \| = O( \sqrt{n} )~~~\text{ almost surely.}
	\end{equation*}                                                                 \end{lemma}
\begin{proof}
	By the law of large numbers and the definition of $\Lambda$ in Equation \eqref{eq:def:secmm},
	\begin{equation*}
		\left\| \frac{ \X^\top \X }{ n } - \secmm \right\| = o( 1 ).
	\end{equation*}
	Multiplying through by $n$,
	\begin{equation*}
		\left\| \X^\top \X - n \secmm \right\| = o( n ).
	\end{equation*}
	Applying the triangle inequality and using this bound,
	\begin{equation*}
		\left\| \X \right\|^2
		= \left\| \X^\top \X \right\|
		\le \left\| n\secmm \right\| + \left\| \X^\top \X - n\secmm \right\|
		\le n \left\| \secmm \right\| + o( n ).
	\end{equation*}
	Taking square roots and recalling that $\secmm$ is fixed in $n$ by definition completes the proof.
\end{proof}

\begin{lemma} \label{lem:APspec}
	Let $(\A, \X) \sim \RDPG( F, n )$ with sparsity parameter $\rho$.
	Denoting $\P =\rho \X \X^\top$, with probability at least $1-O(n^{-2})$,
	\begin{equation*}
		\left\| \A - \P \right\| \le C \sqrt{ \nu + b^2 } \sqrt{n} \log n.
	\end{equation*}
\end{lemma}
\begin{proof}
	This is a special case of Lemma 5 in \cite{levin2022a}, obtained by setting $N=1$ and taking all the $\nu_{ij}$ and $b_{ij}$ parameters to be identically $\nu$ and $\B$, respectively.
\end{proof}

\section{Proof of Theorem~\ref{lem:indepcov:nonid}}
\label{apx:proof:lem:indepcov:nonid}

Here we provide proof details for Theorem~\ref{lem:indepcov:nonid}.
The proof relies on two basic lemmas.

\begin{lemma} \label{lem:indep:GTconverge}
	Under the setting of Theorem~\ref{lem:indepcov:nonid},
	\begin{equation*}
		\max_{i \in [n]} \left| [ \G \T ]_i - \tau \right|
		\rightarrow 0 ~~~\text{ almost surely.}
	\end{equation*}
\end{lemma}
\begin{proof}
	For each $i \in [n]$, define
	\begin{equation*}
		\Ttilde_i = \left[ \G \T \right]_i
		= \left[ \D^{-1} \A \T \right]_i = \frac{1}{d_i} \sum_{j=1}^n \A_{ij} T_j.
	\end{equation*}
	Recalling that $\bbE\left[ T_i \mid \A \right] = \tau$, Lemma~\ref{lem:Gsubgamma} implies that
	\begin{equation*}
		\max_{i \in [n]} \left| \Ttilde_i - \tau \right|
		\le
		C \left( \sqrt{\nu} \log n \right) \max_{i \in [n] }
		\sqrt{ \sum_{j=1}^n \frac{ \A_{ij}^2 }{ d_i^2 } }
		+
		C \left( b \log n \right)
		\max_{i \in [n]} \max_{j \in [n]} \frac{ \A_{ij} }{ d_i }.
	\end{equation*}
	Applying the Borel-Cantelli lemma conditional on any sequence of networks obeying our growth assumptions in Equation~\eqref{eq:growth:nub},
	\begin{equation*}
		\max_{i \in [n]} \left| \Ttilde_i - \tau \right|
		=
		o( 1 ) ~\text{ almost surely, }
	\end{equation*}
	as we set out to show.
\end{proof}

\begin{lemma} \label{lem:indep:contagion}
	Under the same setting as Lemma~\ref{lem:indep:GTconverge}, defining
	\begin{equation*}
		\xi = (\I - \beta \G)^{-1} \G^2 \T
		= \sum_{k=0}^\infty \beta^k \G^{k+2} \T
		\in \R^n,
	\end{equation*}
	we have
	\begin{equation*}
		\max_{i \in [n]}
		\left| \xi_i - \tau \right|
		\rightarrow 0~~~\text{ almost surely.}
	\end{equation*}
\end{lemma}
\begin{proof}
	We observe that trivially, since $\G$ is row-stochastic,
	\begin{equation*}
		(\I - \beta \G)^{-1} \G ( \tau \onevec_n) = \tau \onevec_n.
	\end{equation*}
	Thus, for $i \in [n]$,
	\begin{equation*} \begin{aligned}
			\xi_i - \tau
			 & =
			\left[ (\I - \beta \G)^{-1} \G^2 \T \right]_i - \tau
			=
			\left[ \left( \I - \beta \G \right)^{-1} \G
			\left( \G \T -\tau \onevec_n \right) \right]_i \\
			 & =
			\sum_{j=1}^n \left[ \left( \I - \beta \G \right)^{-1} \G \right]_{ij}
			[(\G \T)_j - \tau \onevec_n ]_j.
		\end{aligned} \end{equation*}
	Applying the triangle inequality and bounding the entries of $\G \T - \tau \onevec_n$ by their maximum,
	\begin{equation*}
		\left| \left( \I - \beta \G \right)^{-1} \G
		\left( \G \T - \tau \onevec_n \right) \right|_i
		\le
		\left( \max_{i \in [n]} \Big|[\G \T]_i - \tau \Big| \right)
		\sum_{j=1}^n \left| \left( \I - \beta \G \right)^{-1} \G \right|_{ij}.
	\end{equation*}
	By Lemma~\ref{lem:indep:GTconverge}, it will suffice for us to establish that
	\begin{equation} \label{eq:contagion:neumannrow}
		\max_{i \in [n]}
		\sum_{j=1}^n \left| \left( \I - \beta \G \right)^{-1} \G \right|_{ij}
		= O(1).
	\end{equation}

	We observe that if $\beta=0$, then for any $i \in [n]$,
	\begin{equation*}
		\sum_{j=1}^n \left| \left( \I - \beta \G \right)^{-1} \G \right|_{ij}
		=
		\sum_{j=1}^n \left| \G \right|_{ij}
		= \frac{1}{d_i} \sum_{j=1}^n \A_{ij} = 1,
	\end{equation*}
	and Equation~\eqref{eq:contagion:neumannrow} holds.
	If $\beta \neq 0$, then by the Neumann expansion,
	\begin{equation} \label{eq:neumannG}
		\sum_{j=1}^n \left| \left( \I - \beta \G \right)^{-1} \G \right|_{ij}
		= \sum_{j=1}^n \left| \sum_{q=0}^\infty \beta^q \\G^{q+1} \right|_{ij}
		\le \frac{1}{|\beta|}
		\sum_{j=1}^n \sum_{q=1}^\infty |\beta|^{q} \left| \G^q \right|_{ij}.
	\end{equation}
	Since each $\G^q$ is itself a transition matrix, we have for any $i \in [n]$ and any $q=1,2,3,\dots$,
	\begin{equation*}
		\sum_{j=1}^n [\G^q]_{ij} = 1.
	\end{equation*}
	It follows that
	\begin{equation*}
		\sum_{j=1}^n \left| \left( \I - \beta \G \right)^{-1} \G \right|_{ij}
		\le \frac{1}{|\beta|}
		\sum_{q=1}^\infty |\beta|^q = \frac{1}{|\beta|(1-|\beta|)}.
	\end{equation*}
	Since this right-hand side is constant with respect to $n$ and does not depend on $i \in [n]$, we have established Equation~\eqref{eq:contagion:neumannrow}, completing the proof.
\end{proof}

With the above two lemmas in hand, Theorem~\ref{lem:indepcov:nonid} follows straightforwardly.

\begin{proof}[Proof of Theorem~\ref{lem:indepcov:nonid}]
	By Lemma~\ref{lem:indep:GTconverge},
	\begin{equation} \label{eq:GT:as}
		\max_{i \in [n]} \left| [ \G \T ]_i - \tau \right| \rightarrow 0
		~~~\text{ almost surely, }
	\end{equation}
	where we remind the reader that $\tau = \bbE T_1$ is the shared expectation of the node-level covariates.
	Similarly, by Lemma~\ref{lem:indep:contagion},
	\begin{equation} \label{eq:G2T:as}
		\max_{i \in [n]}
		\left| \left[(\I - \beta \G)^{-1} \G^2 \T \right]_i
		- \frac{ \tau }{1-\beta} \right|
		\rightarrow 0
		~~~\text{ almost surely. }
	\end{equation}
	By an argument that mirrors that in Lemma~\ref{lem:indep:GTconverge},
	\begin{equation} \label{eq:Geps:as}
		\max_{i \in [n]}
		\left| \left[(\I - \beta \G)^{-1} \G \bm \varepsilon \right]_i \right|
		\rightarrow 0 ~~~\text{ almost surely. }
	\end{equation}

	Multiplying both sides of Equation~\eqref{eq:lim-red} by $\G$,
	\begin{equation*} \begin{aligned}
			\G \Y
			 & = \frac{ \alpha }{1-\beta} \onevec_n
			+ \left(\I - \beta \G\right)^{-1} \G \T \gamma
			+ \left(\I - \beta \G\right)^{-1} \G^2 \T \delta
			+ \left(\I - \beta \G\right)^{-1} \G \bm \varepsilon .
		\end{aligned} \end{equation*}
	Applying the triangle inequality and Cauchy-Schwarz,
	\begin{equation*} \begin{aligned}
			\max_{i \in [n]} \left| \left[ \G \Y \right]_i - \frac{ \alpha }{1-\beta}
			- \frac{ \tau( \gamma + \delta ) }{ 1-\beta } \right|
			 & \le
			\| \gamma \| \left\| \left(\I - \beta \G\right)^{-1} \right\|
			\max_{i \in [n]} \left| \left[ \G \T - \tau \right]_i \right|                   \\
			 & ~~~~~~+
			\| \delta \| \max_{i \in [n]} \left|
			\left(\I - \beta \G\right)^{-1} \left( \G^2 \T - \tau \onevec_n \right) \right| \\
			 & ~~~~~~+ \max_{i \in [n]} \left| \left[
				\left(\I - \beta \G\right)^{-1} \G \bm \varepsilon
				\right]_i \right|.
		\end{aligned} \end{equation*}
	Taking $\eta = [\alpha + \tau ( \gamma + \delta )]/(1-\beta)$
	and applying Lemma~\ref{lem:IbgyG:invertible} along with Equations~\eqref{eq:GT:as},~\eqref{eq:G2T:as} and~\eqref{eq:Geps:as}, we have shown that
	\begin{equation*}
		\max_{i \in [n]} \left| \left[\G \Y \right]_i - \eta \right|
		\rightarrow 0 ~~~\text{ almost surely,}
	\end{equation*}
	completing the proof.
\end{proof}

\section{Proof of Theorem~\ref{thm:LIM:minimax}} \label{apx:LIM:minimax}

In this section, we provide proof details for Theorem~\ref{thm:LIM:minimax}, establishing estimation lower-bounds on the intercept, contagion and interference terms in the linear-in-means model of Equation~\eqref{eq:lim-mv}.
Our main technical tool is a basic information theoretic minimax lower bound, which appears as Theorem~2.5 in \cite{tsybakov2009}, which we reproduce here for ease of reference.

\begin{theorem} \label{thm:Tsybakov:2.5}
        Let $\Theta$ be a parameter set endowed with a semi-distance $d( \cdot,\cdot)$.
        Suppose that $\theta_0,\theta_1,\dots,\theta_M \in \Theta$ for $M \ge 2$, and write $\bbP_j = \bbP_{\theta_j}$ for $j=0,1,\dots,M$.
        Suppose further that
        \begin{enumerate}
                \item $d( \theta_j ,\theta_k ) \ge 2s > 0$ for all $0\le j < k \le M$,
                \item $\bbP_{\theta_j} \ll \bbP_{\theta_0}$ for all $j=1,2,\dots,M$, and
                \item $\frac{1}{M} \sum_{j=1}^M
                              \KL\left( \bbP_{\theta_j} \| \bbP_{\theta_0} \right)
                              \le c_K \ln M$,
        \end{enumerate}
        where $c_K \in (0,1/8)$. Then
        \begin{equation*}
                \inf_{\thetahat} \sup_{\theta \in \Theta} ~
                \bbP_{\theta}\!\left[ d(\thetahat,\theta) \ge s \right]
                \ge \frac{ \sqrt{M} }{ 1 + \sqrt{M} }
                \left( 1 - 2 c_K - \sqrt{\frac{2 c_K}{\ln M} } \right) ,
        \end{equation*}
        where the infimum is over all estimators.
\end{theorem}
We will also make use of the following basic fact about Kullback-Leibler divergence between Gaussian distributions, namely that for independent $p$-dimensional Gaussians $Z_1 \sim \calN(\bm \mu_1, \bm \Sigma_1)$ and $Z_2 \sim \calN(\bm \mu_2, \bm \Sigma_2)$,
\begin{equation} \label{eq:KL:gaussian}
        \KL( \calN_2 \| \calN_1 )
        = \frac{1}{2}
        \left[ \trace \bm \Sigma_1^{-1} \bm \Sigma_2 - p
                + (\bm \mu_1 - \bm \mu_2)^\top \bm \Sigma_1^{-1} (\bm \mu_1 - \bm \mu_2)
                + \ln \frac{ \det\bm  \Sigma_1 }{ \det \bm \Sigma_2 } \right] .
\end{equation}

To establish Theorem~\ref{thm:LIM:minimax}, we separately derive estimation lower-bounds for $\alpha$, $\beta$ and $\delta$ below in Section~\ref{apx:subsec:minimax:coefs}, and combine them in Section~\ref{apx:subsec:minimax:combine} to complete our proof.

\subsection{Lower-bounds for \texorpdfstring{$\alpha$, $\beta$}{a, b} and \texorpdfstring{$\delta$}{d}}
\label{apx:subsec:minimax:coefs}

Here, we separately prove estimation lower-bounds for the intercept, contagion and interference terms in Equation~\eqref{eq:lim-mv}.
We will then combine these below in Section~\ref{apx:subsec:minimax:combine} to establish Theorem~\ref{thm:LIM:minimax}.

\begin{lemma} \label{lem:beta:minimax}
	Under the model in Equation~\eqref{eq:lim-mv}, suppose that $\varepsilon \sim \calN(0, \sigmaeps^2 \I)$ and suppose that the entries of $\T$ are drawn i.i.d.~according to a distribution with mean $\tau \in \R$.
	There exist positive constants $c_\beta$ and $c_\beta'$ such that
	\begin{equation*}
		\inf_{\thetahat} \sup_{\theta \in \ThetaLIM}
		\bbP_{\theta}\left[
			\left| \betahat - \beta \right| \ge \frac{c_\beta}{\| \G \|_F} \right]
		\ge c_\beta' > 0,
	\end{equation*}
	where the infimum is over all estimators.
\end{lemma}
\begin{proof}
	We begin by defining a semi-distance on $\ThetaLIM$.
	Letting $\theta = (\alpha,\beta,\gamma,\delta)$ and $\theta' = (\alpha',\beta',\gamma',\delta')$, define
	\begin{equation} \label{eq:def:dbeta}
		d_\beta( \theta, \theta' ) = | \theta_2 - \theta'_2 |
		= | \beta - \beta' | .
	\end{equation}
	Define $\theta^{(0)}, \theta^{(+)}, \theta^{(-)} \in \ThetaLIM$ according to
	\begin{equation*}
		\theta^{(0)} = (0,0,    0,0),~~ \theta^{(+)} = (0,\beta,0,0)
		~\text{ and }~ \theta^{(-)} = (0,-\beta,0,0),
	\end{equation*}
	where $\beta > 0$ will be specified below.

	Let $\Y^{(0)},\Y^{(+)}$ and $\Y^{(-)}$ be distributed according to the linear-in-means model with coefficients $\theta^{(0)},\theta^{(+)}$ and $\theta^{(-)}$, respectively, and common $\calN(0,\sigmaeps^2 \I)$ noise distribution.
	That is,
	\begin{equation*} \begin{aligned}
			\Y^{(0)} & = \bm \varepsilon^{(0)},                                   \\
			\Y^{(+)} & = \beta \G \Y^{(+)} + \bm \varepsilon^{(+)}, ~\text{ and } \\
			\Y^{(-)} & = -\beta \G \Y^{(+)} + \bm \varepsilon^{(-)},
		\end{aligned} \end{equation*}
	where $\bm \varepsilon^{(0)}, \bm \varepsilon^{(+)}$ and $\bm \varepsilon^{(-)}$ are distributed independently from an $n$-dimensional Gaussian $\calN(0,\sigmaeps^2 \I)$.
	It follows that
	\begin{equation*} \begin{aligned}
			\Y^{(0)} & \sim \calN\left( 0, \sigma^2 \I \right) ,                                 \\
			\Y^{(+)} & \sim \calN\left( 0, \sigma^2 (\I - \beta \G)^{-1} (\I - \beta \G)^{-\top}
			\right) ~\text{ and }                                                                \\
			\Y^{(-)} & \sim \calN\left( 0, \sigma^2 (\I + \beta \G)^{-1} (\I + \beta \G)^{-\top}
			\right) .
		\end{aligned} \end{equation*}

	Recalling Equation~\ref{eq:KL:gaussian},
	\begin{equation} \label{eq:beta:KLstart} \begin{aligned}
			\KL\left( \Y^{(+)} \| \Y^{(0)} \right)
			 & = \frac{1}{2}\!
			\Big[ \trace (\I - \beta \G)^{-1} (\I - \beta \G)^{-\top} \!-\! n
				+ 2\ln \det (\I - \beta \G) \Big] .
		\end{aligned} \end{equation}
	Applying the Neumann expansion,
	\begin{equation*}
		(\I - \beta \G)^{-1} (\I - \beta \G)^{-\top}
		= \sum_{q=0}^\infty \sum_{r=0}^\infty \beta^{q+r}
		\G^q [\G^\top]^r ,
	\end{equation*}
	so that
	\begin{equation*} \begin{aligned}
			\trace (\I - \beta \G)^{-1} (\I - \beta \G)^{-\top}
			 & = \sum_{q=0}^\infty \sum_{r=0}^\infty \beta^{q+r}
			\trace \G^q [\G^\top]^r                              \\
			 & = n + \beta \trace \G + \beta \trace \G^\top
			+ \sum_{q=1}^\infty \sum_{r=1}^\infty \beta^{q+r}
			\trace \G^q [\G^\top]^r .
		\end{aligned} \end{equation*}
	Since the network is hollow by assumption, we have $\trace \G = \trace \G^\top = 0$, and thus
	\begin{equation*} \begin{aligned}
			\trace (\I - \beta \G)^{-1} (\I - \beta \G)^{-\top}
			 & = n + \sum_{q=1}^\infty \sum_{r=1}^\infty \beta^{q+r} \trace \G^q [\G^\top]^r \\
			 & = n
			+ \beta^2 \sum_{q=0}^\infty \sum_{r=0}^\infty
			\beta^{q+r} \trace \G^q \G \G^\top [\G^\top]^r .
		\end{aligned} \end{equation*}
	Applying Cauchy-Schwarz and using the fact that $\| \G \| \le 1$,
	\begin{equation*}
		\left|  \trace \G^q \G \G^\top [\G^\top]^r \right|
		\le \left\| \G^q \G \right\|_F \left\| [\G^\top]^q \G^\top \right\|_F
		\le \left\| \G \right\|^{q+r} \left\| \G \right\|_F^2
		\le \left\| \G \right\|_F^2,
	\end{equation*}
	and we conclude that
	\begin{equation*}
		\left| \trace (\I - \beta \G)^{-1} (\I - \beta \G)^{-\top} \right|
		\le n
		+ \frac{ \beta^2 }{ \left(1-\beta\right)^2 } \left\| \G \right\|_F^2 .
	\end{equation*}
	Applying this bound to Equation~\eqref{eq:beta:KLstart},
	\begin{equation} \label{eq:beta:KL:detcheckpt}
		\KL\left( \Y^{(+)} \| \Y^{(0)} \right)
		\le \frac{1}{2} \left[
			\frac{ \beta^2 \left\| \G \right\|_F^2 }{ \left(1-\beta\right)^2 }
			+ 2 \ln \det (\I - \beta \G) \right] .
	\end{equation}

	Applying basic properties of the determinant and recalling that
	\begin{equation*}
		1 \ge \lambda_1 \ge \lambda_2 \ge \cdots \ge \lambda_n \ge -1
	\end{equation*}
	are the eigenvalues of $\G$,
	\begin{equation*}
		\ln \det (\I - \beta \G)
		= \sum_{i=1}^n \ln \left( 1 - \beta \lambda_i \right) .
	\end{equation*}
	Since $\ln (1+x ) \le x$ for all $x > -1$,
	\begin{equation*}
		\ln \det (\I - \beta \G)
		\le -\beta \sum_{i=1}^n \lambda_i
		= -\beta \trace \G
		= 0 .
	\end{equation*}
	where the second equality follows from the fact that the network is hollow by assumption.
	Substituting into Equation~\eqref{eq:beta:KL:detcheckpt},
	\begin{equation*}
		\KL\left( \Y^{(+)} \| \Y^{(0)} \right)
		\le
		\frac{ \beta^2 }{ 2\left(1-\beta\right)^2 } \left\| \G \right\|_F^2 .
	\end{equation*}

	By nearly the same argument, accounting for changing the sign of $\beta$,
	\begin{equation*}
		\KL\left( \Y^{(-)} \| \Y^{(0)} \right)
		\le
		\frac{ \beta^2 }{ 2\left(1+\beta\right)^2 } \left\| \G \right\|_F^2
		\le
		\frac{ \beta^2 }{ 2\left(1-\beta\right)^2 } \left\| \G \right\|_F^2,
	\end{equation*}
	where the second inequality comes from our assumption that $\beta > 0$.
	Combining the above two displays,
	\begin{equation*}
		\frac{1}{2}\left[ \KL\left( \Y^{(+)} \| \Y^{(0)} \right)
			+ \KL\left( \Y^{(-)} \| \Y^{(0)} \right) \right]
		\le
		\frac{ \beta^2 }{ 2 \left(1-\beta\right)^2 } \left\| \G \right\|_F^2 .
	\end{equation*}
	Setting $\beta = c/\| \G \|_F$ for $c \in (0,1/2)$ to be specified below, we have
	\begin{equation*}
		\frac{1}{2}\left[ \KL\left( \Y^{(+)} \| \Y^{(0)} \right)
			+ \KL\left( \Y^{(-)} \| \Y^{(0)} \right) \right]
		\le
		\frac{c^2}{2} \left( \frac{ \| \G \|_F^2 }{ \| \G \|_F - c } \right)^2 .
	\end{equation*}
	Since $\| \G \|_F \ge 1$ and $c \in (0,1/2)$, we have $\| \G \|_F - c \ge \| \G \|_F/2$, and thus
	\begin{equation*}
		\frac{1}{2}\left[ \KL\left( \Y^{(+)} \| \Y^{(0)} \right)
			+ \KL\left( \Y^{(-)} \| \Y^{(0)} \right) \right]
		\le
		2 c^2 .
	\end{equation*}
	Choosing $c = \sqrt{2\ln 2}/8$ ensures that
	\begin{equation}  \label{eq:beta:KLsum}
		\frac{1}{2}\left[ \KL\left( \Y^{(+)} \| \Y^{(0)} \right)
			+ \KL\left( \Y^{(-)} \| \Y^{(0)} \right) \right]
		\le \frac{\ln 2}{16}.
	\end{equation}

	With our choice of $\beta = c/\| \G \|_F$ and our definition of $d_\beta(\cdot,\cdot)$ in Equation~\eqref{eq:def:dbeta},
	\begin{equation} \label{eq:beta:distance}
		d_\beta\left( \theta^{(0)}, \theta^{(+)} \right)
		=
		d_\beta\left( \theta^{(0)}, \theta^{(-)} \right)
		= \frac{ \sqrt{ \ln 2 } }{ 4 \| \G \|_F } .
	\end{equation}

	Equations~\eqref{eq:beta:KLsum} and~\eqref{eq:beta:distance} along with basic properties of the Gaussian ensure the conditions of Theorem~\ref{thm:Tsybakov:2.5} are met with $c_K=1/16$, and $s = \sqrt{ \ln 2 }/8\| \G \|_F$, and thus
	\begin{equation*}
		\inf_{\thetahat} \sup_{\theta \in \ThetaLIM}
		\bbP_{\theta}\left[ \left| \betahat - \beta \right|
			\ge \frac{ \sqrt{ \ln 2 } }{ 8\| \G \|_F } \right]
		\ge \frac{ \sqrt{2} }{1+\sqrt{2}}
		\left(1 - \frac{1}{8} - \sqrt{ \frac{ 1 }{ 8 \ln 2 }} \right) .
	\end{equation*}
	Taking $c_\beta'$ equal to the right-hand side and choosing $c_\beta = \sqrt{ \ln 2 }/8$ completes the proof.
\end{proof}

Our results establishing estimation lower-bounds for $\beta$ and $\delta$ will require a concentration inequality for the vector $\G \T$, which we now state.

\begin{lemma} \label{lem:GTtau:UB}
	Under the assumptions of Theorem~\ref{thm:LIM:minimax},
	there exists a constant $c_\tau > 0$ such that
	\begin{equation*}
		\bbP\left[ \left\| \G \T - \tau \onevec_n \right\| \le c_\tau \| \G \|_F \right]
		\ge \frac{1}{2} .
	\end{equation*}
\end{lemma}
\begin{proof}
	Writing $\bm \Tdot = \T - \tau \onevec_n$ for ease of notation, we first observe that since $G \onevec_n = \onevec_n$,
	\begin{equation*}
		\G \T - \tau \onevec_n = \G (\T - \tau \onevec_n) = \G \bm  \Tdot.
	\end{equation*}
	Expanding the Euclidean norm,
	$\left\| \G \T - \tau \onevec_n \right\|^2 = \bm \Tdot^\top \G^\top \G \bm \Tdot$,
	and thus
	\begin{equation} \label{eq:GTtau:expecnorm}
		\bbE \left\| \G \T - \tau \onevec_n \right\|^2
		= \sigmaT^2 \trace \G^\top \G
		= \sigmaT^2 \| \G \|_F^2.
	\end{equation}
	Applying Lemma~\ref{lem:Zquad:conc} with $M =\G^\top \G$ and choosing appropriate constants in the proof, with probability greater than $1/2$,
	\begin{equation*}
		\left\| \G \T - \tau \onevec_n \right\|^2
		\le \bbE ~\Tdot^\top\G^\top \G \Tdot
		+ C \sigmaT^2 \|\G^\top \G \|_F
		= \sigmaT^2 \| \G \|_F^2 + C \sigmaT^2 \|\G^\top \G \|_F
	\end{equation*}
	Using the trivial upper bound $\|\G^\top \G \|_F \le \| \G \|_F^2$ and taking square roots, with probability greater than $1/2$,
	\begin{equation} \label{eq:GTtau:UB}
		\left\| \G \T - \tau \onevec_n \right\| \le C \sigmaT \| \G \|_F .
	\end{equation}
	Defining $c_\tau = C \sigmaT$ completes the proof.
\end{proof}

\begin{lemma} \label{lem:delta:minimax}
	Under the model in Equation~\eqref{eq:lim-mv}, suppose that $\varepsilon \sim \calN(0, \sigmaeps^2 \I)$ and suppose that the entries of $\T$ are drawn i.i.d.~according to a distribution with mean $\tau \in \R$.
	There exist positive constants $c_\delta$ and $c_\delta'$ such that for any estimator $\thetahat = (\alphahat,\betahat,\gammahat,\deltahat)$ of $\theta = (\alpha,\beta,\gamma,\delta)$,
	\begin{equation*}
		\bbP_{\theta}\left[
			\left| \deltahat - \delta \right| \ge \frac{c_\delta}{\| \G \|_F} \right]
		\ge c_\delta'.
	\end{equation*}
\end{lemma}
\begin{proof}
	We begin by defining a semi-distance on $\ThetaLIM$.
	Letting $\theta = (\alpha,\beta,\gamma,\delta)$ and $\theta' = (\alpha',\beta',\gamma',\delta')$, define
	\begin{equation} \label{eq:def:ddelta}
		d_\delta( \theta, \theta' ) = | \theta_4 - \theta'_4 |
		= | \delta - \delta' | .
	\end{equation}
	Define $\theta^{(0)}, \theta^{(+)}, \theta^{(-)} \in \R^4$ according to
	\begin{equation} \label{eq:def:thetas4delta} \begin{aligned}
			\theta^{(0)} & = (      0, 0, 0, 0 ),
			~~\theta^{(+)} = ( \alpha, 0, 0, \delta)
			~\text{ and }~\theta^{(-)} = ( -\alpha,0, 0,-\delta),
		\end{aligned} \end{equation}
	where $\alpha$ and $\delta$ will be specified below.

	Let $\Y^{(0)},\Y^{(+)}$ and $\Y^{(-)}$ be distributed according to the linear-in-means model with coefficients $\theta^{(0)},\theta^{(+)}$ and $\theta^{(-)}$, respectively, and common $\calN(0,\sigmaeps^2 \I)$ noise distribution.
	That is,
	\begin{equation*} \begin{aligned}
			\Y^{(0)} & = \bm \varepsilon^{(0)},                                                 \\
			\Y^{(+)} & = \alpha \onevec_n + \delta \G \T + \bm \varepsilon^{(+)}, ~\text{ and } \\
			\Y^{(-)} & = -\alpha \onevec_n - \delta \G \T + \bm \varepsilon^{(+)},              \\
		\end{aligned} \end{equation*}
	where $\bm \varepsilon^{(0)}, \bm \varepsilon^{(+)}$ and $\bm \varepsilon^{(-)}$ are independently generated according to an $n$-dimensional Gaussian $\calN(0,\sigmaeps^2 \I)$.
	It follows that, conditional on $\G$ and $\T$,
	\begin{equation*} \begin{aligned}
			\Y^{(0)} & \sim \calN\left( 0, \sigmaeps^2 \I \right) ,                                 \\
			\Y^{(+)} & \sim \calN\left( \alpha \onevec_n + \delta \G \T , \sigmaeps^2 \I \right)
			~\text{ and }                                                                           \\
			\Y^{(-)} & \sim \calN\left( -\alpha \onevec_n - \delta \G \T , \sigmaeps^2 \I \right) .
		\end{aligned} \end{equation*}

	Applying Equation~\eqref{eq:KL:gaussian}, writing $\KL_{G,T}$ to stress that all probabilities are conditional on $\G$ and $\T$,
	\begin{equation*}
		\KL_{G,T}\left( \Y^{(+)} \| \Y^{(0)} \right)
		= \frac{1}{2 \sigmaeps^2} \left\| \alpha \onevec_n + \delta \G \T \right\|^2 .
	\end{equation*}
	Setting
	\begin{equation} \label{eq:delta:pickparams}
		\alpha = \frac{c_0 ~\tau}{\| \G \|_F}
		~\text{ and }~
		\delta = \frac{c_0}{\| \G \|_F},
	\end{equation}
	for $c_0 > 0$ a constant to be specified below, we have
	\begin{equation*}
		\KL_{G,T}\left( \Y^{(+)} \| \Y^{(0)} \right)
		= \frac{c_0^2}{2 \sigmaeps^2}
		\frac{ \left\| \tau \onevec_n - \G \T \right\|^2 }{ \| \G \|_F^2 } .
	\end{equation*}
	The same argument yields
	\begin{equation*}
		\KL_{G,T}\left( \Y^{(-)} \| \Y^{(0)} \right)
		= \frac{c_0^2}{2 \sigmaeps^2}
		\frac{ \left\| \tau \onevec_n - \G \T \right\|^2 }{ \| \G \|_F^2 } .
	\end{equation*}
	Combining the above two displays,
	\begin{equation} \label{eq:delta:KLGT}
		\frac{1}{2}
		\left[ \KL_{G,T}\left( \Y^{(+)} \| \Y^{(0)} \right)
			+ \KL_{G,T}\left( \Y^{(-)} \| \Y^{(0)} \right) \right]
		= \frac{c_0^2}{2 \sigmaeps^2}
		\frac{ \left\| \tau \onevec_n - \G \T \right\|^2 }{ \| \G \|_F^2 } .
	\end{equation}

	Define $E_\delta$ to be the event
	\begin{equation*} %
		E_\delta = \left\{ \left\| \tau \onevec_n - \G \T \right\| \le c_\tau \| \G \|_F \right\},
	\end{equation*}
	where $c_\tau > 0$ is the constant guaranteed by Lemma~\ref{lem:GTtau:UB}.
	Since $\varepsilon$ is independent of $(G,T)$ and does not depend on our choice of $\theta \in \ThetaLIM$, we have for any $q \in \R$,
	\begin{equation} \label{eq:delta:Esplit}
		\begin{aligned}
			\inf_{\thetahat} \sup_{\theta \in \ThetaLIM}
			\bbP_{\theta}\!\left[ d_\delta(\thetahat,\theta) \ge q \right]
			 & \ge
			\inf_{\thetahat} \sup_{\theta \in \ThetaLIM}
			\bbP_{\theta}\!\left[ d_\delta(\thetahat,\theta) \ge q,
			E_\delta \right]     \\
			 & =\Pr[ E_\delta ]~
			\inf_{\thetahat} \sup_{\theta \in \ThetaLIM}
			\bbP_{\theta}\!\left[ d_\delta(\thetahat,\theta) \ge q
				~\Big|~\! E_\delta \right] .
		\end{aligned} \end{equation}

	Write $\KL_{E_\delta}$ to denote the KL-divergence conditional on the event $E_\delta$, i.e.,
	\begin{equation} \label{eq:conditionalKL}
		\KL_{E_\delta}(\Y^{(+)} \| \Y^{(0)} ) =
		\bbE\left[ \ln \frac{ f_{(+)}(\Y^{(+)}) }{ f_{(0)}(\Y^{(+)}) }
			\Bigg | E_\delta \right],
	\end{equation}
	where $f_{(+)}(x)$ and $f_{(0)}(x)$ denote the densities of $\Y^{(+)}$ and $\Y^{(0)}$, respectively.
	Combining Equation~\eqref{eq:delta:KLGT} with the definition of the event $E_\delta$,
	\begin{equation*}
		\frac{1}{2}
		\left[ \KL_{E_\delta}\left( \Y^{(+)} \| \Y^{(0)} \right)
			+ \KL_{E_\delta}\left( \Y^{(-)} \| \Y^{(0)} \right) \right]
		\le \frac{c_0^2 c_\tau^2}{2 \sigmaeps^2} .
	\end{equation*}
	Setting $c_0^2= (\sigmaeps^2 \ln 2)/8c_\tau^2$,
	\begin{equation*} %
		\frac{1}{2}
		\left[ \KL_{E_\delta}\left( \Y^{(+)} \| \Y^{(0)} \right)
			+ \KL_{E_\delta}\left( \Y^{(-)} \| \Y^{(0)} \right) \right]
		\le \frac{ 1 }{ 16 } \ln 2 .
	\end{equation*}
	Recalling our definition of $d_\delta(\cdot,\cdot)$ from Equation~\eqref{eq:def:ddelta} and recalling our parameter choices from Equation~\eqref{eq:def:thetas4delta},
	\begin{equation*}
		d_\delta\left( \theta^{(+)}, \theta^{(0)} \right)
		=
		d_\delta\left( \theta^{(-)}, \theta^{(0)} \right)
		=
		\sqrt{ \frac{ \sigmaeps^2 \ln 2 }{ 8 c_\tau^2 } } \frac{ 1}{ \| \G \|_F } .
	\end{equation*}
	We note further that by basic properties of the normal,
	\begin{equation*}
		\bbP_{\theta^{(+)}}[ ~\cdot \mid E_\delta ]
		\ll \bbP_{\theta^{(0)}}[ ~\cdot \mid E_\delta ]
		~\text{ and }~
		\bbP_{\theta^{(-)}}[ ~\cdot \mid E_\delta ]
		\ll \bbP_{\theta^{(0)}}[ ~\cdot \mid E_\delta ].
	\end{equation*}
	The above three displays ensure that we may apply a conditional version of Theorem~\ref{thm:Tsybakov:2.5} with
	\begin{equation*}
		s= \sqrt{ \frac{ \sigmaeps^2 \ln 2 }{ 32 c_\tau^2 } } \frac{ 1}{ \| \G \|_F }
	\end{equation*}
	and $c_K = 1/16$ to write
	\begin{equation*}
		\inf_{\thetahat} \sup_{\theta \in \ThetaLIM}
		\bbP_{\theta}\!\left[ d_\delta(\thetahat,\theta) \ge
			\sqrt{ \frac{ \sigmaeps^2 \ln 2 }{ 32 c_\tau^2 } }
			\frac{ 1}{ \| \G \|_F } \Bigg | E_\delta \right]
		\ge \frac{ \sqrt{2} }{ 1 + \sqrt{2} } \!
		\left(\! 1 - \frac{1}{8} - \sqrt{\frac{1}{8 \ln 2} } \right) \! .
	\end{equation*}
	Thus, taking
	\begin{equation*}
		c_\delta = \sqrt{ \frac{ \sigmaeps^2 \ln 2 }{ 32 c_\tau^2 } },
	\end{equation*}
	we have shown that
	\begin{equation*}
		\inf_{\thetahat} \sup_{\theta \in \ThetaLIM}
		\bbP_{\theta}\left[ d_\delta(\thetahat,\theta)
			\ge \frac{ c_\delta }{ \| \G \|_F }
			\mid E_\delta \right]
		\ge \frac{ \sqrt{2} }{ 1 + \sqrt{2} }
		\left( 1 - \frac{1}{8} - \sqrt{\frac{1}{8 \ln 2} } \right)
		> 0.
	\end{equation*}
	Plugging this into Equation~\eqref{eq:delta:Esplit} and recalling our definition of $d_\delta(\cdot,\cdot)$ from Equation~\eqref{eq:def:ddelta},
	\begin{equation*} \begin{aligned}
			\inf_{\thetahat} \sup_{\theta \in \Theta}
			\bbP_{\theta}\left[ \left| \deltahat-\delta \right|
				\ge \frac{ c_\delta }{ \| \G \|_F } \right]
			 & \ge \frac{ \sqrt{2} }{ 1 + \sqrt{2} }
			\left( 1 - \frac{1}{8} - \sqrt{\frac{1}{8 \ln 2} } \right)
			\Pr[ E_\delta ]                                      \\
			 & \ge \frac{1}{2} \frac{ \sqrt{2} }{ 1 + \sqrt{2} }
			\left( 1 - \frac{1}{8} - \sqrt{\frac{1}{8 \ln 2} } \right) ,
		\end{aligned} \end{equation*}
	where the second inequality follows from Lemma~\ref{lem:GTtau:UB}.
	Setting
	\begin{equation*}
		c_\delta' = \frac{1}{2} \frac{ \sqrt{2} }{ 1 + \sqrt{2} }
		\left( 1 - \frac{1}{8} - \sqrt{\frac{1}{8 \ln 2} } \right)
	\end{equation*}
	completes the proof.
\end{proof}

\begin{lemma} \label{lem:alpha:minimax}
	Under the model in Equation~\eqref{eq:lim-mv}, suppose that $\varepsilon \sim \calN(0, \sigmaeps^2 \I)$ and suppose that the entries of $\T$ are drawn i.i.d.~according to a distribution with non-zero mean $\tau \in \R$.
	There exist positive constants $c_\alpha$ and $c_\alpha'$ such that for any estimator $\thetahat = (\alphahat,\betahat,\gammahat,\deltahat)$ of $\theta = (\alpha,\beta,\gamma,\delta) \in \ThetaLIM$,
	\begin{equation*}
		\bbP_{\theta}\left[
			\left| \alphahat - \alpha \right| \ge \frac{c_\alpha'}{\| \G \|_F} \right]
		\ge c_\alpha' > 0.
	\end{equation*}
\end{lemma}
\begin{proof}
	We begin by defining a semi-distance on $\ThetaLIM$.
	Letting $\theta = (\alpha,\beta,\gamma,\delta)$ and $\theta' = (\alpha',\beta',\gamma',\delta')$, define
	\begin{equation} \label{eq:def:dalpha}
		d_\alpha( \theta, \theta' ) = | \theta_1 - \theta'_1 |
		= | \alpha - \alpha' | .
	\end{equation}
	Define $\theta^{(0)}, \theta^{(+)}, \theta^{(-)} \in \R^4$ according to
	\begin{equation} \label{eq:def:thetas4alpha} \begin{aligned}
			\theta^{(0)} & = (      0, 0, 0, 0 )                    \\
			\theta^{(+)} & = ( \alpha, 0, 0, \delta ) ~\text{ and } \\
			\theta^{(-)} & = (-\alpha, 0, 0, -\delta )              \\
		\end{aligned} \end{equation}
	where $\alpha > 0$ and $\delta > 0$ will be specified below.

	Let $\Y^{(0)},\Y^{(+)}$ and $\Y^{(-)}$ be distributed according to the linear-in-means model with coefficients $\theta^{0},\theta^{(+)}$ and $\theta^{(-)}$, respectively, and common $\calN(0,\sigmaeps^2 \I)$ noise distribution.
	That is,
	\begin{equation*} \begin{aligned}
			\Y^{(0)} & = \bm \varepsilon^{(0)},                                                 \\
			\Y^{(+)} & = \alpha \onevec_n + \delta \G \T + \bm \varepsilon^{(+)}, ~\text{ and } \\
			\Y^{(-)} & = -\alpha \onevec_n - \delta \G \T + \bm \varepsilon^{(-)},              \\
		\end{aligned} \end{equation*}
	where $\bm \varepsilon^{(0)}, \bm \varepsilon^{(+)}$ and $\bm \varepsilon^{(-)}$ are independently generated according to an $n$-dimensional Gaussian $\calN(0,\sigmaeps^2 \I)$.
	It follows that, conditional on $\G$ and $\T$,
	\begin{equation*} \begin{aligned}
			\Y^{(0)} & \sim \calN\left( 0, \sigmaeps^2 \I \right) ,                                 \\
			\Y^{(+)} & \sim \calN\left( \alpha \onevec_n + \delta \G \T , \sigmaeps^2 \I \right)
			~\text{ and }                                                                           \\
			\Y^{(-)} & \sim \calN\left( -\alpha \onevec_n - \delta \G \T , \sigmaeps^2 \I \right) .
		\end{aligned} \end{equation*}

	Applying Equation~\eqref{eq:KL:gaussian}, writing $\KL_{G,T}$ to stress that all probabilities are conditional on $\G$ and $\T$,
	\begin{equation*}
		\KL_{G,T}\left( \Y^{(+)} \| \Y^{(0)} \right)
		= \frac{1}{2 \sigmaeps^2} \left\| \alpha \onevec_n + \delta \G \T \right\|^2 .
	\end{equation*}
	Setting
	\begin{equation} \label{eq:alpha:paramchoices}
		\delta = \frac{c_0}{\| \G \|_F} ~\text{ and }~ \alpha = \frac{c_0 \tau}{\| \G \|_F} ,
	\end{equation}
	where $c_0 > 0$ is to be specified below,
	\begin{equation*}
		\KL_{G,T}\left( \Y^{(+)} \| \Y^{(0)} \right)
		= \frac{c_0^2}{2 \sigmaeps^2 \| \G \|_F^2} \left\| \tau \onevec_n - \G \T \right\|^2 .
	\end{equation*}
	By the same argument,
	\begin{equation*}
		\KL_{G,T}\left( \Y^{(-)} \| \Y^{(0)} \right)
		= \frac{c_0^2}{2 \sigmaeps^2 \| \G \|_F^2} \left\| \tau \onevec_n - \G \T \right\|^2 .
	\end{equation*}
	Combining the above two displays,
	\begin{equation} \label{eq:alpha:KLbound}
		\frac{1}{2}
		\left[ \KL_{G,T}\left( \Y^{(+)} \| \Y^{(0)} \right)
			+ \KL_{G,T}\left( \Y^{(-)} \| \Y^{(0)} \right) \right]
		= \frac{c_0^2}{2 \sigmaeps^2 \| \G \|_F^2} \left\| \tau \onevec_n - \G \T \right\|^2 .
	\end{equation}

	As in the proof of Lemma~\ref{lem:delta:minimax}, we employ a conditional version of Theorem~\ref{thm:Tsybakov:2.5}.
	Define the event $E_\alpha$ according to
	\begin{equation} \label{eq:def:Ealpha}
		E_\alpha = \left\{ \left\| \G \T - \tau \onevec_n \right\|
		\le c_\tau \| \G \|_F \right\} ,
	\end{equation}
	where $c_\tau$ is the positive constant guaranteed by Lemma~\ref{lem:GTtau:UB}.
	We observe that on the event $E_\alpha$,
	\begin{equation} \label{eq:alpha:onEalpha}
		\frac{ \sqrt{2 \sigmaeps^2 \ln 2} ~|\tau| }{ 8 \left\| \tau \onevec_n - \G \T \right\|}
		\ge
		\frac{ \sqrt{2 \sigmaeps^2 \ln 2} ~|\tau| }{ 8 c_\tau \| \G \|_F } .
	\end{equation}
	Since $\varepsilon$ is independent of $(G,T)$ and does not depend on our choice of $\theta \in \ThetaLIM$, we have for any $q \in \R$,
	\begin{equation} \label{eq:alpha:Esplit}
		\begin{aligned}
			\inf_{\thetahat} \sup_{\theta \in \ThetaLIM}
			\bbP_{\theta}\!\left[ d_\alpha(\thetahat,\theta) \ge q \right]
			 & \ge
			\inf_{\thetahat} \sup_{\theta \in \ThetaLIM}
			\bbP_{\theta}\!\left[ d_\alpha(\thetahat,\theta) \ge q,
			E_\alpha \right]     \\
			 & =\Pr[ E_\alpha ]~
			\inf_{\thetahat} \sup_{\theta \in \ThetaLIM}
			\bbP_{\theta}\!\left[ d_\alpha(\thetahat,\theta) \ge q
				~\Big|~\! E_\alpha \right] .
		\end{aligned} \end{equation}
	Following the notation established in Equation~\eqref{eq:conditionalKL}, our bound in Equation~\eqref{eq:alpha:KLbound} implies
	\begin{equation*}
		\frac{1}{2}
		\left[ \KL_{E_\alpha}\left( \Y^{(+)} \| \Y^{(0)} \right)
			+ \KL_{E_\alpha}\left( \Y^{(-)} \| \Y^{(0)} \right) \right]
		\le
		\frac{c_0^2 c_\tau^2}{2 \sigmaeps^2}
	\end{equation*}
	and choosing
	\begin{equation} \label{eq:alpha:c0}
		c_0^2 = \frac{ \sigmaeps^2 \ln 2 }{ 8 c_\tau^2 }
	\end{equation}
	ensures that
	\begin{equation*}
		\frac{1}{2}
		\left[ \KL_{E_\alpha}\left( \Y^{(+)} \| \Y^{(0)} \right)
			+ \KL_{E_\alpha}\left( \Y^{(-)} \| \Y^{(0)} \right) \right]
		\le
		\frac{ \ln 2 }{ 16 } .
	\end{equation*}

	With our choice of $\alpha$ and $\delta$ in Equation~\eqref{eq:alpha:paramchoices}, our parameter choices from Equation~\eqref{eq:def:thetas4alpha} imply
	\begin{equation} \label{eq:delta:thetadists}
		d_\alpha\left( \theta^{(+)}, \theta^{(0)} \right)
		=
		d_\alpha\left( \theta^{(-)}, \theta^{(0)} \right)
		=
		\frac{ c_0 |\tau| }{ \| \G \|_F } .
	\end{equation}
	Observing that the distributions of $\Y^{(0)}, \Y^{(+)}$ and $\Y^{(-)}$ are mutually absolutely continuous with respect to one another, we may apply a conditional version of Theorem~\ref{thm:Tsybakov:2.5} with $c_K=1/16$ and
	\begin{equation*}
		s = \frac{ c_0 |\tau| }{ 2 \| \G \|_F } .
	\end{equation*}
	It follows that, recalling the definition of $d_\alpha(\cdot,\cdot)$ from Equation~\eqref{eq:def:dalpha}
	\begin{equation*}
		\inf_{\thetahat} \sup_{\theta \in \ThetaLIM}
		\bbP_{\theta}\!\left[ \left| \alphahat - \alpha \right|
			\ge \frac{ c_0 |\tau| }{ 2 \| \G \|_F }
			~\Big|~\! E_\alpha \right]
		\ge \frac{ \sqrt{2} }{ 1 + \sqrt{2} }
		\left( 1 - \frac{1}{8} - \sqrt{\frac{1}{8 \ln 2} } \right).
	\end{equation*}
	Applying this lower-bound to Equation~\eqref{eq:alpha:Esplit} with $q = c_0 |\tau| / 2 \| \G \|_F$,
	\begin{equation*} \begin{aligned}
			\inf_{\thetahat} \sup_{\theta \in \ThetaLIM}
			\bbP_{\theta}\!\left[ d_\alpha(\thetahat,\theta)
				\ge \frac{ c_0 |\tau| }{ 2 \| \G \|_F } \right]
			 & \ge \Pr[ E_\alpha ]~ \frac{ \sqrt{2} }{ 1 + \sqrt{2} }
			\left( 1 - \frac{1}{8} - \sqrt{\frac{1}{8 \ln 2} } \right) \\
			 & \ge \frac{1}{2} \frac{ \sqrt{2} }{ 1 + \sqrt{2} }
			\left( 1 - \frac{1}{8} - \sqrt{\frac{1}{8 \ln 2} } \right) ,
		\end{aligned} \end{equation*}
	where the second inequality follows from Lemma~\ref{lem:GTtau:UB}.
	Substituting our choice of $c_0$ from Equation~\eqref{eq:alpha:c0},
	\begin{equation*} \begin{aligned}
			\inf_{\thetahat} \sup_{\theta \in \ThetaLIM}
			\bbP_{\theta}\!\left[ d_\alpha(\thetahat,\theta)
				\ge \frac{ |\tau| \sigmaeps^2 \ln 2 }{ 16 c_\tau^2 \| \G \|_F } \right]
			 & \ge \frac{1}{2} \frac{ \sqrt{2} }{ 1 + \sqrt{2} }
			\left( 1 - \frac{1}{8} - \sqrt{\frac{1}{8 \ln 2} } \right) .
		\end{aligned} \end{equation*}
	Choosing
	\begin{equation*}
		c_\alpha' \frac{1}{2} \frac{ \sqrt{2} }{ 1 + \sqrt{2} }
		\left( 1 - \frac{1}{8} - \sqrt{\frac{1}{8 \ln 2} } \right) > 0
	\end{equation*}
	and
	\begin{equation*}
		c_\alpha = \frac{ |\tau| \sigmaeps^2 \ln 2 }{ 16 c_\tau^2 }
	\end{equation*}
	completes the proof, once we recall that $\tau \neq 0$ by assumption.
\end{proof}

\subsection{Combining Coefficient Lower-bounds}
\label{apx:subsec:minimax:combine}

With the above results in hand, we are ready to prove our estimation lower-bound in Theorem~\ref{thm:LIM:minimax}.

\begin{proof}[Proof of Theorem~\ref{thm:LIM:minimax}]
	Applying Lemmas \ref{lem:beta:minimax} and~\ref{lem:delta:minimax}, there exist positive constants $c_\beta,c_\beta',c_\delta$ and $c_\delta'$ such that
	\begin{equation*} \begin{aligned}
			\inf_{\thetahat} \sup_{\theta \in \ThetaLIM}
			\bbP_{\theta}\left[
				\left| \betahat - \beta \right| \ge \frac{c_\beta}{\| \G \|_F} \right]
			 & \ge c_\beta'    \\
			~\text{ and }~
			\inf_{\thetahat} \sup_{\theta \in \ThetaLIM}
			\bbP_{\theta}\left[
				\left| \deltahat - \delta \right| \ge \frac{1}{\| \G \|_F} \right]
			 & \ge c_\delta' .
		\end{aligned} \end{equation*}
	Taking $c_0 = \min\{c_\beta', c_\delta' \}$ yields Equation~\eqref{eq:minimax:betadelta}.
	Applying Lemma~\ref{lem:alpha:minimax}, when $\tau \neq 0$, we have
	\begin{equation*}
		\bbP_{\theta}\left[
			\left| \alphahat - \alpha \right| \ge \frac{c_\alpha}{\| \G \|_F} \right]
		\ge c_\alpha',
	\end{equation*}
	completing the proof.
\end{proof}

\section{Proof of Theorem~\ref{thm:rdpg}} \label{app:partial-id-details}

In this section, we present the technical details alluded to in the statement of Theorem~\ref{thm:rdpg} in the main text and provide a proof of the result.

\subsection{Technical Assumptions} \label{apx:partialid:assums}

We begin by collecting the technical conditions required for Theorem~\ref{thm:rdpg}.
These pertain to the behavior of the latent position distribution $F$ and its interaction with the sparsity parameter $\rho$ and the subgamma parameters $\nu$ and $b$, which control the moment decay conditions of the edge-level noise.
We take $F$ to be fixed with respect to $n$, but we allow both the subgamma parameters $(\nu,b) = (\nu_n,b_n)$ and the sparsity parameter $\rho = \rho_n$ to vary with $n$, to capture the widely observed phenomenon of sparse networks (i.e., the expected number of edges grows at a rate slower than $n^2$). While these parameters are allowed to depend on $n$, we suppress the subscripts on $\nu, b$ and $\rho$ in the sequel for the sake of readability.

Our first assumption ensures that the network $\A$ is suitably dense.
\begin{assumption} \label{assum:growth:sparsity}
	The edge-level $(\nu,b)$-subgamma parameters and the sparsity $\rho$ are such that                                                                              \begin{equation} \label{eq:growth:sparsityLB}                                           \rho = \omega\left( \frac{ \log^2 n }{ n^{1/2} } \right)                \end{equation}
	and
	\begin{equation} \label{eq:nubrho:ratio}                                                \frac{ \nu + b^2 }{ \rho } = \Theta( 1 ).
	\end{equation}                                                          \end{assumption}

\begin{remark}
	In the context of a binary network, the restriction on $\rho$ in Equation~\eqref{eq:growth:sparsityLB} is equivalent to requiring all degrees in the network to grow at rates faster than $\sqrt{n} \log^2 n$.
	This rate is larger than the more typical requirement in the random dot product graph literature that the degrees grow at $\omega(\log^c n )$ rates \citep[see, e.g.,][]{rubin-delanchy2022,levin2022a}.
	While a more careful analysis might yield a less strict lower bound, we observe that this lower bound matches, up to the logarithmic factor, the lower bound required for convergence of ordinary least squares estimates considered in \cite{lee2002}. It may be of interest to pursue the asymptotics of $\G \X$ and $\G \Y$ in sparser regimes using concentration inequalities specialized to the binary case \citep{tropp2015,lei2015}, but we do not pursue that here.
\end{remark}

Our remaining assumptions concern the latent position distribution $F$.
In essence, Assumptions~\ref{assum:Fsparse:interact},~\ref{assum:F:extremes:norho} and~\ref{assum:F:momentratio} ensure that the latent positions do not give rise to overly sparse networks and that $F$ is not so heavy-tailed as to preclude convergence of the quadratic terms in the latent positions.

\begin{assumption} \label{assum:Fsparse:interact}
	Let $\mu \in \R^d$ denote the mean of the latent position distribution $F$.
	The latent position distribution $F$ and the sparsity parameter $\rho$ are such that
	\begin{equation} \label{eq:iplb:rho:newer}
		\min_{i \in 1, ..., n} |\X_i^\top \bm \mu|
		= \omega\left( \frac{ \log^2 n }{ \sqrt{ n } \rho } \right)
		~\text{ almost surely. }
	\end{equation}
\end{assumption}

\begin{assumption} \label{assum:F:extremes:norho}
	The latent position distribution $F$ is such that
	\begin{equation} \label{eq:growth:maxnorm}
		\max_{i \in 1, ..., n} \| \X_i \| = o( n^{1/2} )
		~\text{ almost surely.}
	\end{equation}
\end{assumption}

\begin{assumption} \label{assum:F:momentratio}
	The latent position distribution $F$ is such that                               \begin{equation} \label{eq:F:secondmoment}
		\bbE \| \X_1 \|^2 < \infty.
	\end{equation}
\end{assumption}

With Assumptions~\ref{assum:growth:sparsity}, ~\ref{assum:Fsparse:interact}, ~\ref{assum:F:extremes:norho} and~\ref{assum:F:momentratio} in hand, we are ready to state our main result, with the help of some additional notation.
Denote by $\secmm$ the second moment matrix of $F$
\begin{equation} \label{eq:def:secmm}
	\secmm = \E{\X_1 \X_1^\top}.
\end{equation}
Additionally, define the random diagonal matrix
\begin{equation} \label{eq:def:H}
	\bm H = \diag(\X_1^\top \bm \mu,\X_2^\top \bm \mu, \dots,\X_n^\top \bm \mu ) \in \R^{n \times n}.
\end{equation}
and the population-level matrix
\begin{equation} \label{eq:def:Gamma}
	\bm \Gamma
	= \left(I - \beta \, \E{\frac{\X_1\X_1^\top }{\X_1^\top \bm \mu }}\right)^{-1} - I.
\end{equation}
We take $\secmm$ and $\bm \Gamma$ to be fixed as a function of $n$.

Under these assumptions, the key result is the following lemma, which describes the asymptotic behavior of the $\G \X$ and $\G \Y$ columns of the design matrix.

\begin{lemma} \label{lem:rdpg:bramoulleconverge}
	Under Assumptions~\ref{assum:growth:sparsity}, ~\ref{assum:Fsparse:interact}, ~\ref{assum:F:extremes:norho}, and~\ref{assum:F:momentratio}, suppose that $(\A, \X) \sim \RDPG(F,n)$ with sparsity $\rho$ and that $F$ is such that $\X \in \bbR^{n \times d}$ is rank $d$ with probability $1$.
	Let $\varepsilon$ be a vector of mean zero, i.i.d.~$(\nueps,\beps)$-subgamma random variables, with $(\nueps,\beps)$ not depending on $n$,
	and let
	\begin{equation}
		\Y = \alpha \onevec_n + \beta \G \Y + \X \bm \gamma + \G \X \bm \delta + \bm \varepsilon
	\end{equation}
	for $\alpha, \beta \in \R$ and $\bm \gamma, \bm \delta \in \R^d$.
	Let $\gammatilde = \secmm \bm \gamma +  \bm \Gamma \secmm (\beta \bm \gamma + \bm \delta)$. Then
	\begin{align*}
		\norm*{\G \X - \paren*{\bm H^{-1} \X \secmm}}_{\tti}                                           & = o(1) ~\text{ almost surely.} \\
		                                                                                               & \text{ and }                   \\
		\norm*{\G \Y - \paren*{\frac{\alpha}{1 - \beta} \onevec_n + \bm H^{-1} \X \gammatilde}}_{\tti} & = o(1) ~\text{ almost surely.}
	\end{align*}
\end{lemma}

See Appendix~\ref{proof:thm:rdpg:bramoulleconverge} for a proof.
Given these results, we next seek to understand the asymptotic behavior of the design matrix.
Let
\begin{equation*}
	\W_n = \begin{bmatrix}
		\onevec_n & \G \Y & \X & \G \X
	\end{bmatrix}
	~\text{and}~
	\W =
	\begin{bmatrix}
		\onevec_n & \, \left(\frac{\alpha}{1 - \beta} \onevec_n + \bm H^{-1} \X \gammatilde \right) & \, \X & \, \bm H^{-1} \X \secmm
	\end{bmatrix}.
\end{equation*}
Later, we will show that in the large-$n$ limit, $\W_n$ and $\W$ become arbitrarily close in spectral norm, from which we will argue that $\W_n^\top \W_n / n$ converges to the same covariance matrix as $\W^\top \W /n$. Then, we will be interested in understanding the rank of $\W$, and thus $\W^\top \W / n$.

Several preliminary results that will help us characterize $\rank \W$. The following result, for instance, clarifies when $\X$ and $\G \X$ are linearly independent in the asymptotic limit. This requires a slightly awkward condition on heterogeneity in the expected degree distribution of the network.

\begin{proposition}
	\label{prop:XHX-rank}
	Let $\bm \mu \in \R^d$ and suppose that $\Y_1,\Y_2,\dots,\Y_d,\Z_1,\Z_2,\dots,\Z_d \in \R^d$ are rows of $\X \in \R^{n \times d}$ such that $\Y_1,\Y_2,\dots,\Y_d$ are linearly independent and $\Z_1,\Z_2,\dots,\Z_d$ are linearly independent.
	Collecting these vectors in the rows of $\Y \in \R^{d \times d}$ and $\Z \in \R^{d \times d}$, respectively, define
	\begin{equation*}
		\bm H_Y = \diag\left( \Y_1^\top \bm \mu, \Y_2^\top \bm \mu, \dots, \Y_d^\top \bm \mu \right)
		~~~\text{ and }~~~
		\bm H_Z = \diag\left( \Z_1^\top \bm \mu, \Z_2^\top \bm \mu, \dots, \Z_d^\top \bm \mu \right).
	\end{equation*}
	Provided that
	\begin{equation*}
		\Z^{-1} \bm H_Z^{-1} \Z - \Y^{-1} \bm H_Y^{-1} \Y \in \R^{d \times d}
	\end{equation*}
	is invertible, then the matrix
	\begin{equation} \label{eq:targetmx}
		\M = \begin{bmatrix} \X & \bm H^{-1} \X \end{bmatrix} \in \R^{n \times 2d}
	\end{equation}
	has rank $2d$.
\end{proposition}

Intuitively, distinguishing the direct effect $\bm \gamma$ from the interference term $\bm \delta$ depends heavily on the presence of degree heterogeneity within the network. In particular, the (scaled) degree normalization matrix $n \rho \D^{-1}$ converges to $\bm H^{-1}$ in the large-$n$ limit, which encodes the expected degree of each node, conditional on the latent positions $\X$. Thus, $\W$ obtains rank $2d$ when information contained in the latent positions $\X$ and expected degree-scaled latent positions $\bm H^{-1} \X$ are linearly independent. This should generally be the case in, for example, degree-corrected stochastic blockmodels.

Note that Proposition \ref{prop:XHX-rank} further implies that $\X$ is linearly independent of the intercept $\onevec_n$ by our previous observation that $\bm H^{-1} \X$ is always collinear with the intercept.

\begin{proof}[Proof of Proposition~\ref{prop:XHX-rank}]
	We begin by observing that it will suffice to show that
	\begin{equation*}
		\bm{\Mtilde} = \begin{bmatrix}
			\Y ~ & ~ \bm H_Y^{-1} \Y \\
			\Z ~ & ~ \bm H_Z^{-1} \Z\end{bmatrix}
		\in \R^{2d \times 2d}
	\end{equation*}
	has full rank, since $\bm{\Mtilde} \in \R^{2d \times 2d}$ comprises $2d$ rows of the matrix $\bm M$ defined in eq.~\eqref{eq:targetmx}.
	Since $\Y$ is invertible by linear independence of $\Y_1,\Y_2,\dots,\Y_d$, $\Mtilde$ is invertible if and only if the Schur complement
	\begin{equation*}
		\bm H_Z^{-1} \Z - \Z \Y^{-1} \bm H_Y^{-1} \Y \in \R^{d \times d}
	\end{equation*}
	is invertible.
	Multiplying by appropriate matrices, invertibility of the above matrix is equivalent to invertibility of
	\begin{equation*}
		\Z^{-1} \bm H_Z^{-1} \Z - \Y^{-1} \bm H_Y^{-1} \Y,
	\end{equation*}
	completing the proof.
\end{proof}

\begin{proposition}
	\label{prop:w-rank-2d}
	Suppose the conditions of Proposition \ref{prop:XHX-rank} hold. Then the rank of asymptotic design matrix $\W \in \R^{n \times (2d + 2)}$ is $2d$.
\end{proposition}

\begin{proof}
	First we show $\rank \W \le 2d$. This follows from the fact that $\bm H^{-1} \X \secmm$ is collinear with the intercept, by definition of $\bm H$ and $\secmm$. Similarly, $\bm H^{-1} \X \secmm$ is collinear with $\frac{\alpha}{1 - \beta} \onevec_n + \bm H^{-1} \X \gammatilde$. Thus, in the asymptotic limit, $\G \X$ is collinear with both the intercept $\onevec_n$ and $\G \Y$ and the maximum rank of $\W$ is $2d + 2 - 2 = 2d$. By Proposition \ref{prop:XHX-rank}, $\begin{bmatrix} \X & \bm H^{-1} \X \secmm \end{bmatrix}$ has rank $2d$, since $\secmm$ is invertible by definition. Thus $\rank \W \ge 2d$, completing the proof.
\end{proof}

\begin{proposition}
	\label{prop:WnWn-to-WW}
	Under Assumptions~\ref{assum:growth:sparsity},~\ref{assum:Fsparse:interact},~\ref{assum:F:extremes:norho},~and~\ref{assum:F:momentratio},
	\begin{equation*}
		\left\| \frac{\W_n^\top \W_n}{n} - \frac{\W^\top \W}{n} \right\| = o(1) \quad \text{almost surely}.
	\end{equation*}
\end{proposition}

\begin{proof}
	We will show that
	\begin{equation*}
		\left\| \frac{\W_n^\top \W_n}{n} - \frac{\W^\top \W}{n} \right\|_\tti = o(1) \quad \text{almost surely},
	\end{equation*}
	which implies the corresponding bound on the spectral norm. Adding and subtracting appropriate terms and applying the triangle inequality, together with properties of the two-to-infinity-norm we have
	\begin{align*}
		\left\| \frac{\W_n^\top \W_n}{n} - \frac{\W^\top \W}{n} \right\|_\tti
		 & \le \frac{2}{n} \lVert \W^\top (\W - \W_n) \rVert_\tti + \frac{1}{n} \lVert \W - \W_n \rVert^2_\tti               \\
		 & \le \frac{2}{n} \lVert \W \rVert_\infty \lVert \W - \W_n \rVert_\tti + \frac{1}{n} \lVert \W - \W_n \rVert^2_\tti
	\end{align*}
	where $\lVert \cdot \rVert_\infty$ denotes the maximum row sum of a matrix. It thus sufficient to observe that $\lVert \W \rVert_\infty = o(n)$, as $\norm*{\W - \W_n}_\tti = o(1)$ almost surely by Lemma~\ref{lem:rdpg:bramoulleconverge}. This fact follows from Assumptions~\ref{assum:Fsparse:interact}~and~\ref{assum:F:extremes:norho}, the definition of $\W$, and the reduced-form characterization of $\G \Y$.
\end{proof}

We can now prove Theorem~\ref{thm:rdpg}.

\begin{proof}[Proof of Theorem~\ref{thm:rdpg}]
	Let $\bm \Sigma = \W^\top \W / n$. By Proposition~\ref{prop:WnWn-to-WW},
	\begin{equation*}
		\left\| \frac{\W_n^\top \W_n}{n} - \bm \Sigma \right\| = o(1) \quad \text{almost surely}.
	\end{equation*}
	By Proposition~\ref{prop:w-rank-2d}, $\rank \bm \Sigma = 2d$, since $\rank \W = \rank \W^\top \W / n = \rank \bm \Sigma$.
\end{proof}

\begin{remark}[Which coefficients to drop]
	Theorem~\ref{thm:rdpg} does not offer any guidance on which two columns to drop in order to ensure that $\bm \Sigma$ is full rank. In our simulations on stochastic blockmodels, the precise choice of columns does not seem to matter, but in general, the distribution of $\X$ and the coefficients $\alpha, \beta, \bm \gamma$ and $\bm \delta$ could force a particular set of choices. For instance, if $\mu = c \cdot \gammatilde$ for some $c \in \R$, it would be mandatory to drop either the intercept or the contagion column to ensure that $\bm \Sigma$ obtains rank $2d$. In practice, variance inflation factors can be used to guide the choice of which two coefficients to drop from the full regression.
\end{remark}

\begin{remark}[Asymptotic collinearity in stochastic blockmodels without degree correction]
	The condition of Proposition \ref{prop:XHX-rank} requires that $\X$ has at least $2d$ distinct rows. We illustrate one common scenario where this is not the case. Suppose $F$ is a mixture distribution over $d$ points $\Z_1, \Z_2, \dots, \Z_d \in \R^d$ and $\Z_1, \Z_2, \dots, \Z_d$ are linearly independent. $F$ could thus present represent the distribution of the latent positions of $\X$ for any stochastic blockmodel with full-rank mixing matrix $\B$.

	Let $\bm \Delta = (\Z_1^\top, \Z_2^\top, ..., \Z_d^\top)^\top \in \R^{d \times d}$, such that $\bm \Delta$ is full-rank by hypothesis. Suppose that $n_1, n_2,\dots, n_d$ points are sampled from each of the respective atoms $\Z_1,\Z_2,\dots,\Z_d$. Then, without loss of generality, we can reorder the rows of $\X$ and write
	\begin{equation*}
		\X =
		\begin{bmatrix}
			\onevec_{n_1} & 0             & \dots  & 0             \\
			0             & \onevec_{n_2} & \dots  & 0             \\
			\vdots        & \vdots        & \ddots & \vdots        \\
			0             & 0             & \dots  & \onevec_{n_d}
		\end{bmatrix}
		\bm \Delta
		\text{ and }
		\bm H^{-1} \X =
		\begin{bmatrix}
			(\Z_1^\top \bm \mu)^{-1} \onevec_{n_1} & 0                                      & \dots  & 0                                      \\
			0                                      & (\Z_2^\top \bm \mu)^{-1} \onevec_{n_2} & \dots  & 0                                      \\
			\vdots                                 & \vdots                                 & \ddots & \vdots                                 \\
			0                                      & 0                                      & \dots  & (\Z_d^\top \bm \mu)^{-1} \onevec_{n_d}
		\end{bmatrix}
		\bm \Delta
	\end{equation*}
	such that $\X$ and $\bm H^{-1} \X$ are collinear and $\begin{bmatrix} \X & \bm H^{-1} \X \end{bmatrix} \in \R^{n \times 2d}$ only has rank $d$. Thus, the columns corresponding to direct effect $\gamma$ and the interference effect $\delta$ are asymptotically collinear in stochastic blockmodels without degree correction.
\end{remark}

\subsection{Degree concentration} \label{apx:degree}

The following results relate to control of the degree and associated quantities (e.g., their conditional expectations) to be used in our technical results supporting Theorem~\ref{thm:rdpg}.

\begin{lemma} \label{lem:HinvX:constant}
	Let $(\A, \X) \sim \RDPG( F, n)$ with $(\nu,b)$-subgamma edges and sparsity parameter $\rho$.
	Then there exists a constant $c_F > 0$ depending on $F$ but not on $n$ such that
	\begin{equation*}
		\left\| \bm H^{-1} \X \right\|_{\tti} \le c_F
		~~~\text{ almost surely.}
	\end{equation*}
\end{lemma}
\begin{proof}
	Recalling the definition of $\bm H \in \R^{n \times n}$ from Equation~\eqref{eq:def:H},
	\begin{equation} \label{eq:HinvX:recip}
		\left\| \bm H^{-1} \X \right\|_{\tti}
		= \max_{i \in [n]} \frac{ \|\X_i\| }{ |\X_i^\top \bm \mu| }
		= \left( \min_{i \in [n]} \frac{  |\X_i^\top \bm \mu| }{ \|\X_i\| } \right)^{-1}.
	\end{equation}

	Since $\X_1^\top\X_2 \ge 0$ with probability $1$, we may assume without loss of generality that the support of $F$ is contained in the positive orthant and that $\bm \mu$ has all entries strictly positive.
	Defining $\calS_{\ge 0}$ to be the intersection of the unit sphere with the non-negative orthant, i.e.,
	\begin{equation*}
		\calS_{\ge 0}
		= \left\{ \bm u \in \R^d : \| \bm u\|=1, u_k \ge 0 ~\text{for all}~k \in [d] \right\},
	\end{equation*}
	we have
	\begin{equation*}
		\min_{i \in [n]} \frac{\X_i^\top \bm \mu }{ \|\X_i\|}
		= \min_{i \in [n]} \frac{\X_i^\top }{ \|\X_i\|} \mu
		\ge \inf_{ \bm u \in \calS_{\ge 0} } \bm u^\top \bm \mu
		\ge \min_{k \in [d]} \mu_k.
	\end{equation*}
	Applying this to Equation~\eqref{eq:HinvX:recip} and noting that the right-hand side does not depend on $n$,
	\begin{equation*}
		\left\| \bm H^{-1} \X \right\|_{\tti}
		\le \frac{ 1 }{ \min_{k \in [d]} \mu_k } ~\text{ almost surely.}                \end{equation*}
	Defining $c_F$ to be this right-hand side completes the proof.
\end{proof}

\begin{corollary} \label{cor:boundedexpectation}
	Suppose that $F$ is a distribution on $\R^d$ obeying the assumptions in Definition~\ref{def:rdpg} and suppose that $F$ obeys the growth assumption in Equation~\eqref{eq:F:secondmoment}.
	Then
	\begin{equation*}
		\bbE \frac{ \|\X_1 \|^2 }{ (\X_i^\top \bm \mu)^2 }
		~~~\text{ and }~~~
		\bbE \frac{ \|\X_1 \|^4 }{ (\X_i^\top \bm \mu)^2 }
	\end{equation*}
	exist and are finite.
\end{corollary}
\begin{proof}
	By an argument largely identical to that in Lemma~\ref{lem:HinvX:constant},
	\begin{equation*}
		\sup_{\bm x \in \supp F} \frac{ \|\bm x\| }{ \bm x^\top \bm \mu }
		< \infty,
	\end{equation*}
	from which
	\begin{equation*}
		\bbE \frac{ \|\X_1 \|^2 }{ (\X_i^\top \bm \mu)^2 }
		\le \left( \sup_{\bm x \in \supp F} \frac{ \| \bm x \| }{ \bm x^\top \bm \mu } \right)^2
		< \infty.
	\end{equation*}
	Similarly, applying our assumption in Equation~\eqref{eq:F:secondmoment}, %
	\begin{equation*}
		\bbE \frac{ \|\X_1 \|^4 }{ (\X_i^\top \bm \mu)^2 }
		\le \left( \sup_{\bm x \in \supp F} \frac{ \| \bm x \| }{ \bm x^\top \bm \mu } \right)^2
		\bbE \|\X_1 \|^2
		< \infty,
	\end{equation*}
	completing the proof.
\end{proof}

\begin{lemma} \label{lem:degconc}
	Let $(\A, \X) \sim \RDPG( F, n )$ with $(\nu,b)$-subgamma edges and sparsity parameter $\rho$.
	Denote the (conditional) expected degree by
	\begin{equation} \label{eq:def:delta}
		\delta_i = \bbE[ d_i \mid \X ]
		= \rho \sum_{j : j \neq i}\X_i^\top\X_j.
	\end{equation}
	Then with probability $1 - O(n^{-2})$,
	\begin{equation*}
		\max_{i \in [n]} \left| d_i - \delta_i \right|
		\le
		C \sqrt{ \nu + b^2} \sqrt{n} \log n.
	\end{equation*}
\end{lemma}
\begin{proof}
	Fix $i \in [n]$. We observe that
	\begin{equation*}
		d_i - \delta_i
		= \sum_{j \in [n]\setminus\{i\} } (\A_{ij} - \rho\X_i^\top\X_j)
	\end{equation*}
	is, conditional on $\X$, a sum of $(\nu,b)$-subgamma random variables.
	Applying Lemma~\ref{lem:sgsum} with suitably chosen constants, it holds with probability at least $1-2n^{-3}$ that
	\begin{equation*}
		\left| d_i - \delta_i \right|
		\le C \sqrt{\nu + b^2} \sqrt{n} \log n.
	\end{equation*}
	A union bound over $i \in [n]$ completes the result.
\end{proof}

\begin{lemma} \label{lem:degLB}
	Let $(\A, \X) \sim \RDPG( F, n )$ with $(\nu,b)$-subgamma edges and sparsity parameter $\rho$ with $F$ and the sparsity parameter $\rho$ obeying the growth assumption in Equation~\eqref{eq:iplb:rho:newer}.

	Define the minimum expected degree
	\begin{equation} \label{eq:def:deltamin}
		\deltamin = \min_{i \in [n]} \delta_i,
	\end{equation}
	where $\delta_i$ is as defined in Equation~\eqref{eq:def:delta}.
	Then
	\begin{equation*}
		\min_{i \in [n]} d_i = \Omega( \deltamin ).
	\end{equation*}
\end{lemma}
\begin{proof}
	Applying Lemma~\ref{lem:degconc}, it holds with high probability that
	\begin{equation} \label{eq:degLB:step1}
		\min_{i \in [n]} d_i \ge \deltamin - C \sqrt{\nu+b^2} n^{1/2} \log n.
	\end{equation}
	By Lemma~\ref{lem:deltaLB}, which is proved below,
	\begin{equation*}
		\deltamin = \Omega( n \rho \min_{i \in [n]}\X_i^\top \bm \mu ).
	\end{equation*}
	By our growth assumption in Equation~\eqref{eq:iplb:rho:newer} and the fact that $\rho \le 1$,
	\begin{equation*}
		n \rho \min_{i \in [n]}\X_i^\top \bm \mu
		= \omega( \sqrt{n \rho} \log n )
		= \omega( \sqrt{ \nu+b^2} \sqrt{n} \log n ).
	\end{equation*}
	It follows that $\deltamin = \omega( \sqrt{\nu+b^2} n^{1/2} \log n )$.
	Applying this to Equation~\eqref{eq:degLB:step1} completes the proof.
\end{proof}

\begin{lemma} \label{lem:degrecip:conc}
	Suppose that $(\A, \X) \sim \RDPG(F,n)$ with $(\nu,b)$-subgamma edges and sparsity parameter $\rho$, obeying Assumption~\ref{assum:growth:sparsity} and the growth assumption in Equation~\eqref{eq:iplb:rho:newer}.

	With $\delta_i$ as defined in Equation~\eqref{eq:def:delta}, it holds with probability at least $1-O(n^{-2})$ that for all $i \in [n]$,
	\begin{equation*}
		\left| \frac{1}{d_i} - \frac{1}{\delta_i} \right|
		\le
		\frac{ C \sqrt{\nu + b^2 } \sqrt{n} \log n }{ \delta_i^2 }.
	\end{equation*}
\end{lemma}
\begin{proof}
	By Lemma~\ref{lem:deltaLB},
	\begin{equation*}
		\deltamin = \Omega( n \rho \min_{i \in [n]}\X_i^\top \bm \mu ).
	\end{equation*}
	Combining this with our growth assumption in Equation~\eqref{eq:iplb:rho:newer} and the fact that $\rho \le 1$, then applying our growth assumption in Equation~\eqref{eq:nubrho:ratio},
	\begin{equation*}
		\deltamin = \omega\left( \sqrt{ \nu+b^2 } \sqrt{n} \log n \right).
	\end{equation*}
	It follows that, applying Lemma~\ref{lem:degconc} and trivially bounding $\delta_i \ge \deltamin$, it holds for all $i \in [n]$ that
	\begin{equation} \label{eq:DiLB}
		d_i \ge \delta_i - \left| d_i - \delta_i \right|
		\ge \delta_i - C \sqrt{ \nu + b^2} n^{1/2} \log n
		\ge \delta_i\left( 1 - o( 1 ) \right).
	\end{equation}
	Applying Lemma~\ref{lem:degconc} once more, it holds for all $i \in [n]$ that
	\begin{equation*}
		\left| \frac{1}{d_i} - \frac{1}{\delta_i} \right|
		= \frac{ | d_i - \delta_i | }{ d_i \delta_i }
		\le \frac{ C \sqrt{ \nu + b^2} n^{1/2} \log n }{ d_i \delta_i }.
	\end{equation*}
	Lower-bounding $d_i$ using Equation~\eqref{eq:DiLB} completes the proof.
\end{proof}

\begin{lemma} \label{lem:ratio}
	Let $(\A, \X) \sim \RDPG( F, n )$ with $(\nu,b)$-subgamma edges and sparsity parameter $\rho$, with $F$ obeying Assumption~\ref{assum:F:extremes:norho}.
	Let $r,q \in [0,\infty)$ be such that the expectation
	\begin{equation*}
		\bm \Xi_{r,q} = \bbE \frac{ \|\X_1 \|^r }{ (\X_1^\top \bm \mu )^q }
	\end{equation*}
	exists and is finite. Then
	\begin{equation*}
		\sum_{i=1}^n \frac{ n^{q-1} \rho^{q} \|\X_i \|^r }{ \delta_i^q }
		\le \left( 1 + o(1) \right) \bm \Xi_{r,q}.
	\end{equation*}
\end{lemma}
\begin{proof}
	By Lemma~\ref{lem:deltaconc}, it holds uniformly over $i \in [n]$ that
	\begin{equation*}
		\min_{i \in [n]} \frac{ \delta_i }{ n \rho\X_i^\top \bm \mu }
		\ge 1 - o(1),
	\end{equation*}
	from which it follows that uniformly over $i \in [n]$,
	\begin{equation*}
		\frac{n \rho }{\delta_i} \le \frac{1}{\X_i^\top \bm \mu}\left(1 + o(1) \right).
	\end{equation*}
	Therefore,
	\begin{equation*}
		\sum_{i=1}^n \frac{ n^{q-1} \rho^q \|\X_i \|^r }{ \delta_i^q }
		= \frac{1}{n} \sum_{i=1}^n \|\X_i\|^r
		\left( \frac{ n \rho }{ \delta_i } \right)^q
		\le \frac{\left(1+o(1) \right)^q}{n} \sum_{i=1}^n
		\frac{ \|\X_i \|^r }{ \left(\X_i^\top \bm \mu \right)^q  }.
	\end{equation*}
	Applying the law of large numbers completes the proof.
\end{proof}

\begin{lemma} \label{lem:deltaconc}
	Let $(\A, \X) \sim \RDPG( F, n )$ with $(\nu,b)$-subgamma edges and sparsity parameter $\rho$, with $F$ obeying Assumption~\ref{assum:F:extremes:norho}.
	Then with high probability it holds that
	\begin{equation} \label{eq:deltaconc:1}
		\max_{i \in [n]} \frac{ \left| \delta_i - n\rho\X_i^\top \bm \mu \right| }
		{ \delta_i }
		= o( 1 )~\text{ almost surely.}
	\end{equation}
	and
	\begin{equation} \label{eq:deltaconc:2}
		\max_{i \in [n]} \frac{ n \rho\X_i^\top \bm \mu }{ \delta_i }
		= \Theta( 1 )~\text{ almost surely.}
	\end{equation}
\end{lemma}
\begin{proof}
	Recalling the definition of $\delta_i$ from Equation~\eqref{eq:def:delta}, for any $i \in [n]$,
	\begin{equation*}
		\frac{ \delta_i }{ n \rho } -\X_i^\top \bm \mu
		=
		\frac{1}{n} \sum_{j : j\neq i}\X_i^\top\X_j -\X_i^\top \bm \mu
		=\X_i^\top \left( \Xbar - \bm \mu \right)
		- \frac{ \|\X_i \|^2 }{ n }.
	\end{equation*}
	Applying the triangle inequality, Cauchy-Schwarz and standard concentration inequalities~\citep{boucheron2013,vershynin2020},
	\begin{equation} \label{eq:deltaconc:triangle}
		\left| \frac{ \delta_i }{ n \rho } -\X_i^\top \bm \mu \right|
		\le
		\|\X_i \| \| \Xbar - \bm \mu \| + \frac{ \|\X_i\|^2 }{n}
		\le
		\frac{ C \|\X_i \| \log n }{ \sqrt{n} } + \frac{ \|\X_i\|^2 }{n} .
	\end{equation}
	It follows that
	\begin{equation*}
		\max_{i \in [n]}
		\frac{ \left| \delta - n \rho\X_i^\top \bm \mu \right| }{ n \rho\X_i^\top \bm \mu }
		\le
		\frac{ C \log n }{ \sqrt{n} }
		\max_{i \in[n]} \frac{ \|\X_i \| }{\X_i^\top \bm \mu }
		+ \frac{ C }{ n } \max_{i \in [n]} \|\X_i\|^2.
	\end{equation*}
	Applying Lemma~\ref{lem:HinvX:constant} and our growth assumption in Equation~\eqref{eq:growth:maxnorm},
	\begin{equation} \label{eq:deltabound}
		\max_{i \in [n]} \frac{ \left| \delta - n \rho\X_i^\top \bm \mu \right| }
		{ n \rho\X_i^\top \bm \mu }
		= o(1)~\text{ almost surely.}
	\end{equation}

	Noting that for $i \in [n]$,
	\begin{equation*}
		\frac{ n \rho\X_i^\top \bm \mu }{ \delta_i }
		= \left( 1 +
		\frac{ \delta_i - n \rho\X_i^\top \bm \mu }{ n \rho\X_i^\top \bm \mu } \right)^{-1},
	\end{equation*}
	Equation~\eqref{eq:deltabound} implies that
	\begin{equation*}
		\max_{i \in [n]} \frac{ n \rho\X_i^\top \bm \mu }{ \delta_i }
		= 1 + o(1)~\text{ almost surely,}
	\end{equation*}
	establishing Equation~\eqref{eq:deltaconc:2}.

	Multiplying through by appropriate quantities, for any $i \in [n]$,
	\begin{equation*}
		\frac{ |\delta_i - n \rho\X_i^\top \bm \mu| }{ \delta_i }
		= \frac{ |\delta_i - n \rho\X_i^\top \bm \mu| }{ n \rho\X_i^\top \bm \mu }
		\frac{ n \rho\X_i^\top \bm \mu }{ \delta_i }.
	\end{equation*}
	Taking the maximum over $i \in [n]$ followed by an application of Equations~\eqref{eq:deltaconc:2} and~\eqref{eq:deltabound} yields Equation~\eqref{eq:deltaconc:1}, completing the proof.
\end{proof}

\begin{lemma} \label{lem:deltaLB}
	Let $(\A, \X) \sim \RDPG( F, n )$ with sparsity parameter $\rho$ and recall the definition of $\deltamin$ from Equation~\eqref{eq:def:deltamin}.
	Suppose that $F$ is such that the growth assumption in Equation~\eqref{eq:iplb:rho:newer} holds.
	Then
	\begin{equation*}
		\deltamin
		= \Omega\left( n \rho \min_{i \in [n]}\X_i^\top \bm \mu \right).
	\end{equation*}
\end{lemma}
\begin{proof}
	We recall that for $i \in [n]$,
	\begin{equation*}
		\delta_i = \rho \sum_{j : j \neq i}\X_i^\top\X_j
		= (n\rho)\X_i^\top \bm \mu + (n \rho)\X_i^\top ( \Xbar - \bm \mu ) - \rho \|\X_i \|^2.
	\end{equation*}
	Taking the minimum over all $i \in [n]$,
	\begin{equation} \label{eq:deltamin:lb}
		\deltamin
		\ge
		(n\rho)\left( \min_{i \in [n]}\X_i^\top \bm \mu \right)
		+
		(n \rho) \left( \min_{i \in [n]}\X_i^\top ( \Xbar - \bm \mu ) \right)
		+ \rho \min_{i \in [n]} \|\X_i \|^2.
	\end{equation}

	By standard concentration inequalities \citep{boucheron2013}, $\| \Xbar - \bm \mu \| = O( n^{-1/2} \log n )$, and thus
	\begin{equation} \label{eq:XiERR:firstUB}
		\left| (n \rho) \min_{i \in [n]}\X_i^\top ( \Xbar - \bm \mu ) \right|
		\le
		(n \rho) \left\| \Xbar - \bm \mu \right\| \min_{i \in [n]} \|\X_i \|
		\le
		C \rho \left( \sqrt{n} \log n \right) \min_{i \in [n]} \|\X_i \|.
	\end{equation}
	Bounding the minimum by the average and appealing to the law of large numbers,
	\begin{equation*}
		\left( \sqrt{n} \log n \right) \min_{i \in [n]} \|\X_i \|
		\le
		\left( \sqrt{n} \log n \right) \frac{1}{n} \sum_{i \in [n]} \|\X_i \|
		= O( \sqrt{n} \log n ).
	\end{equation*}
	Applying our growth assumption in Equation~\eqref{eq:iplb:rho:newer} and using the fact that $\rho \le 1$ by assumption,
	\begin{equation*}
		\left( \sqrt{n} \log n \right) \min_{i \in [n]} \|\X_i \|
		= o( n \min_{i \in [n]}\X_i^\top \bm \mu ).
	\end{equation*}
	Applying this to Equation~\eqref{eq:XiERR:firstUB},
	\begin{equation} \label{eq:minnorm:littleoh}
		(n \rho) \min_{i \in [n]}\X_i^\top ( \Xbar - \bm \mu )
		= o( n \rho \min_{i \in [n]}\X_i^\top \bm \mu ).
	\end{equation}

	Again applying the law of large numbers,
	\begin{equation*}
		\min_{i \in [n]} \|\X_i \|^2 \le \frac{1}{n} \sum_{i=1}^n \|\X_i\|^2 = O( 1 ).
	\end{equation*}
	Since Equation~\eqref{eq:minnorm:littleoh} trivially implies $n \min_i\X_i^\top \bm \mu = \omega( n^{-1} )$, it follows that
	\begin{equation*}
		\min_{i \in [n]} \|\X_i \|^2 = o( n \min_{i \in [n]}\X_i^\top \bm \mu ).
	\end{equation*}
	Applying the above display and Equation~\eqref{eq:minnorm:littleoh} to Equation~\eqref{eq:deltamin:lb},
	\begin{equation*}
		\deltamin \ge
		\left(1- o(1)\right) (n\rho) \min_{i \in [n]}\X_i^\top \bm \mu ,
	\end{equation*}
	completing the proof.
\end{proof}

\subsection{Proof of Lemma~\ref{lem:rdpg:bramoulleconverge}}
\label{proof:thm:rdpg:bramoulleconverge}

\begin{proof}
	By Lemma~\ref{lem:rdpg:GXconverge}, established in Appendix~\ref{apx:rdpg:GXconverge},
	\begin{equation} \label{eq:rdpg:GX:tti}
		\left\| \G \X - \bm H^{-1} \X \secmm \right\|_{\tti}
		= o(1) ~\text{ almost surely.}
	\end{equation}

	Multiplying by $\G$ in Equation~\eqref{eq:lim-red},
	\begin{equation} \label{eq:model:bramoulle:repeat}
		\G \Y
		=
		\frac{\alpha }{1-\beta} \onevec_n
		+ (\I - \beta \G)^{-1} \G \X \bm \gamma
		+ (\I - \beta \G)^{-1} \G^2 \X \bm\delta
		+ (\I - \beta \G)^{-1} \G \bm \varepsilon,
	\end{equation}
	where we have used the fact that $\G$ commutes with $(\I - \beta \G)^{-1}$, and the fact that
	\begin{equation*}
		(\I - \beta \G)^{-1} \onevec_n = \frac{1}{1-\beta} \onevec_n.
	\end{equation*}

	By Lemma~\ref{lem:rdpg:contagion:tti}, established in Section~\ref{apx:subsec:rdpg:contagion:tti},                                                                    \begin{equation} \label{eq:rdpg:contagion:tti}
		\left\| (\I - \beta \G)^{-1} \G^2 \X - \bm H^{-1} \X \bm \Gamma \secmm \right\|_{\tti}         = o(1) ~\text{ almost surely.}
	\end{equation}

	By Lemma~\ref{lem:rdpg:interference:tti}, also established in Section~\ref{apx:subsec:rdpg:contagion:tti},
	\begin{equation} \label{eq:rdpg:interference:tti}
		\left\| \left(\I - \beta \G\right)^{-1} \G \X
		- \bm H^{-1} \X \left( I + \beta \bm \Gamma \right) \secmm
		\right\|_{\tti}
		= o(1)~\text{ almost surely.}
	\end{equation}

	Applying basic properties of the $(\tti)$-norm and Lemmas~\ref{lem:Gsubgamma} and~\ref{lem:IbgyG:invertible},
	\begin{equation} \label{eq:rdpg:contagion:epsterm:prelim} \begin{aligned}
			\max_{i \in [n]}
			\left| \left[(\I - \beta \G)^{-1} \G \bm \varepsilon \right]_{i} \right|
			 & \le \left\| (\I - \beta \G)^{-1} \right\|_{\infty}
			\max_{i \in [n]}\left| \left[\G \bm \varepsilon\right]_i \right| \\
			 & \le \frac{ C }{1-|\beta|}
			\left[ \nueps \max_{i \in [n]}
				\sqrt{ \sum_{j=1}^n \frac{ \A_{ij}^2 }{ d_i^2 } \log^2 n }
				+ \beps \max_{i \in [n]} \max_{j \in [n]} \frac{ \A_{ij} }{ d_i }
				\log n
				\right].
		\end{aligned} \end{equation}

	Applying Lemma~\ref{lem:sgmax}, it holds with high probability that for all $i \in [n]$,
	\begin{equation*}
		\sum_{j=1}^n \A_{ij}^2
		\le 2\sum_{j=1}^n ( \rho\X_i^\top\X_j )^2 + 2n (\nu+b^2) \log^2 n .
	\end{equation*}
	Using this fact and applying Lemma~\ref{lem:degLB},
	\begin{equation*}
		\max_{i \in [n]} \sum_{j=1}^n \frac{ \A_{ij}^2 }{ d_i^2 }
		\le \frac{ 2n }{ \deltamin^2 }
		\left[ \frac{\rho^2 }{n} \sum_{j=1}^n (\X_i^\top\X_j )^2
			+ (\nu+b^2) \log^2 n \right].
	\end{equation*}
	Applying the law of large numbers, our growth assumption in Equation~\eqref{eq:nubrho:ratio} and the fact that $\rho \le 1$ by assumption, it holds with high probability that
	\begin{equation*}
		\max_{i \in [n]} \sum_{j=1}^n \frac{ \A_{ij}^2 }{ d_i^2 }
		\le \frac{Cn}{\deltamin} \left( \rho^2 + (\nu + b^2) \log^2 n \right)
		\le \frac{C\rho n \log^2 n}{ \deltamin^2 }.
	\end{equation*}
	Applying Lemma~\ref{lem:deltaLB}, it holds with high probability that
	\begin{equation*}
		\max_{i \in [n]} \sqrt{ \sum_{j=1}^n \frac{ \A_{ij}^2 }{ d_i^2 } }
		\le \frac{ C \log^2 n }
		{ n \rho \left( \min_{i \in [n]}\X_i^\top \bm \mu \right)^2 }.
	\end{equation*}
	Invoking our growth assumption in Equation~\eqref{eq:iplb:rho:newer} and the fact that $\rho \le 1$,
	\begin{equation} \label{eq:nuterm:converge}
		\max_{i \in [n]} \sqrt{ \sum_{j=1}^n \frac{ \A_{ij}^2 }{ d_i^2 } \log^2 n}
		= o( 1 ) ~\text{ almost surely.}
	\end{equation}

	By Lemmas~\ref{lem:sgmax} and~\ref{lem:degLB}, recalling the definition of $\deltamin$ from Equation~\eqref{eq:def:deltamin},
	\begin{equation*}
		\max_{i \in [n]} \max_{j \in [n]} \frac{ \A_{ij} \log n }{ d_i }
		\le \frac{ C \sqrt{\nu+b} \log^2 n }{ \deltamin }.
	\end{equation*}
	Further applying Lemma~\ref{lem:deltaLB} followed by our assumptions in Equations~\eqref{eq:nubrho:ratio} and~\eqref{eq:iplb:rho:newer} and the fact that $\rho \le 1$,
	\begin{equation} \label{eq:bterm:converge}
		\max_{i \in [n]} \max_{j \in [n]} \frac{ \A_{ij} \log n }{ d_i }
		\le \frac{ C \log^2 n }{ n \sqrt{\rho}  \min_{i \in [n]}\X_i^\top \bm \mu }
		= o( 1 )~\text{ almost surely.}
	\end{equation}

	Applying Equations~\eqref{eq:nuterm:converge} and~\eqref{eq:bterm:converge} to Equation~\eqref{eq:rdpg:contagion:epsterm:prelim},
	\begin{equation} \label{eq:rdpg:contagion:epsterm}
		\max_{i \in [n]}
		\left| \left[ (\I - \beta \G)^{-1} \G \bm \varepsilon \right]_{i} \right|
		= o( 1 ) ~\text{ almost surely.}
	\end{equation}

	Recalling Equation~\eqref{eq:model:bramoulle:repeat}, applying the triangle inequality and Equations~\eqref{eq:rdpg:GX:tti},~\eqref{eq:rdpg:contagion:tti},~\eqref{eq:rdpg:interference:tti} and~\eqref{eq:rdpg:contagion:epsterm},
	\begin{equation*} \begin{aligned}
			\Big\| \G \Y & - \frac{\alpha }{1-\beta} \onevec_n
			- \bm H^{-1} \X \left( I + \beta \bm \Gamma \right) \secmm \bm \gamma
			- \bm H^{-1} \X \bm \Gamma \secmm \bm \delta \Big\|_{\tti}        \\
			             & ~~~~~~\le
			\|\gamma\| \left\| \left(\I - \beta \G \right)^{-1} \G \X
			-\bm H^{-1} \X \left( I + \beta \bm \Gamma \right) \secmm
			\right\|_{\tti}                                                   \\
			             & ~~~~~~~~~~~~+ \|\bm \delta\|
			\left\| (\I - \beta \G)^{-1} \G^2 \X
			- \bm H^{-1} \X \bm \Gamma \secmm \right\|_{\tti}
			+ \left\| (\I - \beta \G)^{-1} \G \bm \varepsilon \right\|_{\tti} \\
			             & ~~~~~~= o(1) ~\text{ almost surely},
		\end{aligned} \end{equation*}
	as we set out to show.
\end{proof}

\subsection{Convergence of \texorpdfstring{$GX$}{GX} term} \label{apx:rdpg:GXconverge}

Here we establish the uniform entrywise convergence of $\G \X$ to an appropriate limit object $\bm H^{-1} \X \secmm$, as used in our proof of Lemma~\ref{lem:rdpg:bramoulleconverge}. %

We require the following technical lemma.

\begin{lemma} \label{lem:DinvX:tti}
	Let $(\A, \X) \sim \RDPG( F, n)$ with $(\nu,b)$-subgamma edges and sparsity parameter $\rho$ and suppose that Assumptions~\ref{assum:growth:sparsity},~\ref{assum:Fsparse:interact} and~\ref{assum:F:extremes:norho} hold.
	Then
	\begin{equation*}
		\left\| n \rho \D^{-1} \X - \bm H^{-1} \X \right\|_{\tti} = o( 1 )
		~~~\text{ almost surely.}
	\end{equation*}
\end{lemma}
\begin{proof}
	By Lemma~\ref{lem:degrecip:conc}, with high probability,
	\begin{equation*} \begin{aligned}
			\left\| n \rho \D^{-1} \X - n \rho \calD^{-1} \X \right\|_{\tti}
			 & =
			\max_{i \in [n]} n \rho \left| \frac{1}{d_i} - \frac{1}{\delta_i} \right|
			\|\X_i \| \\
			 & \le
			C \rho \sqrt{\nu + b^2 } \left( n^{3/2}  \log n \right)
			\max_{i \in [n]} \frac{ \|\X_i \| }{ \delta_i^2 }.
		\end{aligned} \end{equation*}
	Multiplying through by appropriate quantities, applying Lemmas~\ref{lem:deltaconc} and~\ref{lem:HinvX:constant}, and applying our growth assumption in Equation~\eqref{eq:nubrho:ratio},
	\begin{equation*} \begin{aligned}
			\left\| n \rho \D^{-1} \X - n \rho \calD^{-1} \X \right\|_{\tti}
			 & \le
			\frac{ C \sqrt{\nu + b^2} \log n }{ \rho \sqrt{n} }
			\left( \max_{i \in [n]} \frac{ \|\X_i\| }{ (\X_i^\top \bm \mu)^2 } \right)
			\left( \max_{i \in [n]} \frac{ n \rho\X_i^\top \bm \mu }{ \delta_i } \right)^2 \\
			 & \le
			\frac{ C \log n }{ \sqrt{n \rho} }
			\left( \max_{i \in [n]} \frac{ 1}{\X_i^\top \bm \mu } \right).
		\end{aligned} \end{equation*}
	Applying our growth assumption in Equation~\eqref{eq:iplb:rho:newer} and using the fact that $\rho \le 1$,
	\begin{equation*}
		\left\| n \rho \D^{-1} \X - n \rho \calD^{-1} \X \right\|_{\tti}
		= o( 1)~\text{ almost surely,}
	\end{equation*}
	so that by the triangle inequality,
	\begin{equation} \label{eq:DinvX:checkpt}
		\left\|  n \rho \D^{-1} \X - \bm H^{-1} \X \right\|_{\tti}
		\le
		\left\| n \rho \calD^{-1} \X - H^{-1} \right\|_{\tti} + o(1).
	\end{equation}

	By definition and basic properties of the norm,
	\begin{equation*} \begin{aligned}
			\left\| n \rho \calD^{-1} \X - \bm H^{-1} \X \right\|_{\tti}
			 & \le \max_{i \in [n]}
			\left| \frac{n \rho}{\delta_i} -  \frac{ 1 }{\X_i^\top \bm \mu } \right|
			\|\X_i \|
			= n \rho \max_{i \in [n]}
			\left| \frac{1}{\delta_i} -  \frac{ 1 }{ n \rho\X_i^\top \bm \mu } \right|
			\|\X_i \|               \\
			 & \le \max_{i \in [n]}
			\frac{ \left| \delta_i - n \rho\X_i^\top \bm \mu \right| \|\X_i \| }
			{\X_i^\top \bm \mu \delta_i }.
		\end{aligned} \end{equation*}
	Applying Lemma~\ref{lem:HinvX:constant},
	\begin{equation*}
		\left\| n \rho \calD^{-1} \X - \bm H^{-1} \X \right\|_{\tti}
		\le
		\left( \max_{i \in [n]} \frac{ \|\X_i \| }{\X_i^\top \bm \mu } \right)
		\left( \max_{i \in [n]} \frac{  \left| \delta_i - n \rho\X_i^\top \bm \mu \right| }
		{ \delta_i } \right)
		\le
		C \left( \max_{i \in [n]}
		\frac{ \left| \delta_i - n \rho\X_i^\top \bm \mu \right| }{\delta_i} \right).
	\end{equation*}
	Applying Lemma~\ref{lem:deltaconc},
	\begin{equation*}
		\left\| n \rho \calD^{-1} \X - \bm H^{-1} \X \right\|_{\tti}
		= o( 1 ) ~~~\text{ almost surely.}
	\end{equation*}
	Applying this to Equation~\eqref{eq:DinvX:checkpt} completes the proof.
\end{proof}

With Lemma~\ref{lem:DinvX:tti} in hand, we are ready to prove our convergence result.

\begin{lemma} \label{lem:rdpg:GXconverge}
	Let $(\A, \X) \sim \RDPG( F, n)$ with $(\nu,b)$-subgamma edges and sparsity parameter $\rho$, and suppose that Assumptions~\ref{assum:growth:sparsity},~\ref{assum:Fsparse:interact} and~\ref{assum:F:extremes:norho} hold.
	Then
	\begin{equation*}
		\left\| \G \X - \bm H^{-1} \X \secmm \right\|_{\tti} = o(1) ~\text{ almost surely.}
	\end{equation*}
\end{lemma}
\begin{proof}
	Recalling $\G = \D^{-1} \A$, by basic properties of the $(\tti)$-norm,
	\begin{equation*}
		\left\| \G \X - \D^{-1} \P \X \right\|_{\tti}
		\le \left\| \D^{-1} \right\|_\infty \left\| (\A-\P)\X \right\|_{\tti}.
	\end{equation*}
	Recalling the definition of $\deltamin$ from Equation~\eqref{eq:def:deltamin} and applying Lemmas~\ref{lem:vectorbernstein} and~\ref{lem:degLB},
	\begin{equation*}
		\left\| \G \X - \D^{-1} \P \X \right\|_{\tti}
		\le \frac{ C \sqrt{\nu+b^2} \log n }{ \deltamin }
		\left( \sum_{j=1}^n \|\X_j \|^2 \right)^{1/2}.
	\end{equation*}
	Multiplying through by appropriate quantities and applying the law of large numbers,
	\begin{equation*}
		\left\| \G \X - \D^{-1} \P \X \right\|_{\tti}
		\le \frac{ C \sqrt{\nu+b^2} \sqrt{n} \log n }{ \deltamin }.
	\end{equation*}
	Further applying Lemma~\ref{lem:deltaLB},
	\begin{equation*}
		\left\| \G \X - \D^{-1} \P \X \right\|_{\tti}
		\le \frac{ C \sqrt{\nu+b^2} \log n }
		{ \sqrt{n} \rho \min_{i \in [n]}\X_i^\top \bm \mu }.
	\end{equation*}
	Our growth assumptions in Equations~\eqref{eq:nubrho:ratio} and~\eqref{eq:iplb:rho:newer} and the fact that $\rho \le 1$ then imply
	\begin{equation} \label{eq:rdpg:GXconverge:tri1}
		\left\| \G \X - \D^{-1} \P \X \right\|_{\tti}
		\rightarrow 0 ~~~\text{ almost surely.}
	\end{equation}

	Recalling that $\P =\rho \X \X^\top$ and using basic properties of the $(\tti)$-norm,
	\begin{equation*} \begin{aligned}
			\left\| \D^{-1} \P \X - \bm H^{-1} \X                            \X^\top \X \right\|_{\tti}
			 & = \left\| (n \rho \D^{-1} \X)\frac{ \X^\top \X}{n}
			- \bm H^{-1} \X \frac{ \X^\top \X }{n} \right\|_{\tti}         \\
			 & \le \left\| \rho \D^{-1} \X - \bm H^{-1} \X \right\|_{\tti}
			\left\| \frac{ \X^\top \X }{ n } \right\|.
		\end{aligned} \end{equation*}
	By the law of large numbers, $n^{-1} \X^\top \X \rightarrow \secmm$ almost surely, and thus, since $\secmm$ is constant with respect to $n$,
	\begin{equation*}
		\left\| \D^{-1} \P \X - \bm H^{-1} \X \X^\top \X \right\|_{\tti}
		\le C \left\| n \rho \D^{-1} \X - \bm H^{-1} \X \right\|_{\tti}.
	\end{equation*}
	Applying Lemma~\ref{lem:DinvX:tti},
	\begin{equation} \label{eq:rdpg:GXconverge:tri2}
		\left\| \D^{-1} \P \X - \bm H^{-1} \X \frac{ \X^\top \X }{ n } \right\|_{\tti} = o( 1 ).
	\end{equation}

	Once more applying basic properties of the $(\tti)$-norm and standard (multivariate) concentration inequalities \citep{vershynin2020,boucheron2013},
	\begin{equation} \label{eq:rdpg:GXconverge:tri3}
		\begin{aligned}
			\left\| \bm H^{-1} \X \frac{ \X^\top \X }{ n } - \bm H^{-1} \X \secmm \right\|_{\tti}
			 & \le \left\| \bm H^{-1} \X \right\|_{\tti}
			\left\| \frac{ \X^\top \X}{n} - \secmm \right\|
			\le \frac{ C \log n }{ \sqrt{n} }\left\| \bm H^{-1} \X \right\|_{\tti} \\
			 & = o( 1 ) ~\text{ almost surely, }
		\end{aligned} \end{equation}
	where the final equality follows from Lemma~\ref{lem:HinvX:constant}.
	Applying the triangle inequality followed by
	Equations~\eqref{eq:rdpg:GXconverge:tri1},~\eqref{eq:rdpg:GXconverge:tri2} and~\eqref{eq:rdpg:GXconverge:tri3} completes the proof.
\end{proof}

\subsection{Convergence of \texorpdfstring{$\X^\top \D^{-1} \X$}{XDX} term}

Our results, in particular Lemma~\ref{lem:rdpg:contagion:tti}, rely upon the convergence of $\X^\top (\I - \beta \G)^{-1} \D^{-1} \X$ to a population analogue.
This convergence in turn relies on convergence of $\rho \X^\top \D^{-1} \X$, which we establish below.

\begin{lemma} \label{lem:XcalDinvX2Expec:converge}
	Let $(\A, \X) \sim \RDPG(F,n)$ with $(\nu,b)$-subgamma edges and sparsity factor $\rho$, and suppose that
	Assumptions~\ref{assum:F:extremes:norho} and~\ref{assum:F:momentratio}
	hold. Then
	\begin{equation*}
		\left\| \rho \X^\top \calD^{-1} \X
		- \bbE \frac{\X_1\X_1^\top }{\X_1^\top \bm \mu } \right\|
		=o( 1 )
		~~~\text{ almost surely. }
	\end{equation*}
\end{lemma}
\begin{proof}
	By the triangle inequality,
	\begin{equation*}
		\left\| \rho \X^\top \calD^{-1} \X
		- \bbE \frac{\X_1\X_1^\top }{\X_1^\top \bm \mu } \right\|
		\le \left\| \rho \X^\top \calD^{-1} \X - \X^\top(n \bm H)^{-1} \X \right\|
		+
		\left\| \X^\top(n \bm H)^{-1}
		- \bbE \frac{\X_1\X_1^\top }{\X_1^\top \bm \mu } \right\|.
	\end{equation*}
	Applying Lemmas~\ref{lem:XcalDinvXtoXHX:converge}
	and~\ref{lem:XHinvXtoExpec:converge},
	\begin{equation*}
		\left\| \rho \X^\top \calD^{-1} \X
		- \bbE \frac{\X_1\X_1^\top }{\X_1^\top \bm \mu } \right\|
		= o(1)
		~\text{ almost surely, }
	\end{equation*}
	completing the proof.
\end{proof}

\begin{lemma} \label{lem:XcalDinvXtoXHX:converge}
	Let $(\A, \X) \sim \RDPG(F,n)$ with $(\nu,b)$-subgamma edges and sparsity factor $\rho$, and suppose that
	Assumptions~\ref{assum:F:extremes:norho} %
	and~\ref{assum:F:momentratio} %
	hold.  Then
	\begin{equation*}
		\left\| \rho \X^\top \calD^{-1} \X - \X^\top(n \bm H)^{-1} \X \right\|
		= O\left( \frac{ \log n }{ \sqrt{n} } \right)
		~~~\text{ almost surely. }
	\end{equation*}
	where $\bm H$ is as defined in Equation~\eqref{eq:def:H}.
\end{lemma}
\begin{proof}
	Expanding the matrix products and applying the triangle inequality,
	\begin{equation*} \begin{aligned}
			\left\| \rho \X^\top \calD^{-1} \X - \X^\top(n \bm H)^{-1} \X \right\|
			 & \le
			\sum_{i=1}^n \|\X_i \|^2 \left| \frac{ \rho }{ \delta_i }
			- \frac{ 1 }{ n\X_i^\top \bm \mu } \right|
			=
			\sum_{i=1}^n \frac{ \|\X_i \|^2 }{ n\X_i^\top \bm \mu }
			\frac{ | \delta_i - n \rho\X_i^\top \bm \mu | }{ \delta_i } \\
			 & \le
			\left( \max_{i \in [n]}
			\frac{ | \delta_i - n \rho\X_i^\top \bm \mu | }{ \delta_i }
			\right)
			\frac{1}{n} \sum_{i=1}^n \frac{ \|\X_i \|^2 }{\X_i^\top \bm \mu }.
		\end{aligned} \end{equation*}
	Applying Lemma~\ref{lem:deltaconc} to control the maximum and applying the law of large numbers to control the sample mean,
	\begin{equation*}
		\left\| \rho \X^\top \calD^{-1} \X - \X^\top(n \bm H)^{-1} \X \right\|
		= o(1)~\text{ almost surely,}
	\end{equation*}
	completing the proof.
\end{proof}

\begin{lemma} \label{lem:XHinvXtoExpec:converge}
	Let $(\A, \X) \sim \RDPG(F,n)$ with $(\nu,b)$-subgamma edges and sparsity factor $\rho$ and suppose that Assumption~\ref{assum:F:momentratio} holds.
	Letting $\bm H \in \R^{n \times n}$ be as defined in
	Equation~\eqref{eq:def:H},
	\begin{equation*}
		\left\| \X^\top(n \bm H)^{-1} \X - \bbE \frac{\X_1\X_1^\top }{\X_1^\top \bm \mu } \right\|
		= o( 1 )~\text{ almost surely.}
	\end{equation*}
\end{lemma}
\begin{proof}
	For ease of notation, write
	\begin{equation*}
		\bm \Xi = \bbE \frac{\X_1\X_1^\top }{\X_1^\top \bm \mu }.
	\end{equation*}
	Expanding the matrix products,
	\begin{equation*}
		X^\top (n \bm H)^{-1} \X - \bm \Xi
		=
		\frac{1}{n} \sum_{i=1}^n \left( \frac{\X_i\X_i^\top }{\X_i^\top \bm \mu } - \bm \Xi \right).
	\end{equation*}
	By the strong law of large numbers,
	\begin{equation*}
		\left\| \X^\top(n \bm H)^{-1} \X - \bm \Xi \right\|
		= \left\| \frac{1}{n} \sum_{i=1}^n \left( \frac{\X_i\X_i^\top }{\X_i^\top \bm \mu }
		- \bm \Xi\right) \right\|
		= o(1)~~\text{ almost surely},
	\end{equation*}
	as we set out to show.
\end{proof}

\begin{lemma} \label{lem:XDinvXtoXcalDinvX:converge}
	Let $(\A, \X) \sim \RDPG(F,n)$ with $(\nu,b)$-subgamma edges and sparsity factor $\rho$ and suppose that
	Assumptions~\ref{assum:F:momentratio},~\ref{assum:growth:sparsity} and~\ref{assum:Fsparse:interact} hold.
	Then
	\begin{equation*}
		\left\| \rho \X^\top \D^{-1} \X - \rho \X^\top \calD^{-1} \X \right\|
		= O\left( \frac{ \log n }{ \sqrt{n \rho} } \right)
		~~~\text{ almost surely. }
	\end{equation*}
\end{lemma}
\begin{proof}
	Expanding the matrix-vector products and applying the triangle inequality,
	\begin{equation*}
		\| \rho \X^\top \D^{-1} \X - \rho \X^\top \calD^{-1} \X \|
		\le
		\rho \sum_{i=1}^n \left| \frac{1}{d_i} - \frac{1}{\delta_i} \right|
		\left\|\X_i \right\|^2.
	\end{equation*}
	Applying Lemma~\ref{lem:degrecip:conc}, it holds with high probability that
	\begin{equation*}
		\| \rho \X^\top \D^{-1} \X - \rho \X^\top \calD^{-1} \X \|
		\le
		C \rho \left( \sum_{i=1}^n \frac{ \|\X_i \|^2 }{ \delta_i^2 } \right)
		\left( \sqrt{\nu + b^2} n^{1/2} \log n \right)
	\end{equation*}
	Multiplying through by appropriate quantities and applying Lemma~\ref{lem:ratio} with $r=q=2$ (noting that Lemma~\ref{lem:ratio} applies, thanks to Corollary~\ref{cor:boundedexpectation}),
	\begin{equation*}
		\| \rho \X^\top \D^{-1} \X - \rho \X^\top \calD^{-1} \X \|
		\le \frac{ C \sqrt{ \nu + b^2 } \log n }{ \sqrt{n} \rho }
		\left( \bbE \frac{ \|\X_1\|^2 }{ (\X_1^\top \bm \mu)^2 } + o(1) \right)
	\end{equation*}
	with high probability.
	Applying our growth assumptions in Equation~\eqref{eq:nubrho:ratio} and using the fact that $F$ is assumed fixed with respect to $n$ completes the proof.
\end{proof}

\subsection{Convergence of Contagion Term} \label{apx:subsec:rdpg:contagion:tti}

\begin{lemma} \label{lem:EDinvX:tti}
	Let $(\A, \X) \sim \RDPG( F, n)$ with $(\nu,b)$-subgamma edges and sparsity parameter $\rho$, and suppose that
	Assumptions~\ref{assum:growth:sparsity}, %
	~\ref{assum:Fsparse:interact},
	~\ref{assum:F:extremes:norho} %
	and~\ref{assum:F:momentratio} %
	hold.  Then
	\begin{equation*}
		\left\| (\A - \P) \D^{-1} \X \right\|_{\tti}
		= O\left( \frac{ \log^2 n }{ \sqrt{n} \rho } \right)
		~~~\text{ almost surely.}
	\end{equation*}
\end{lemma}
\begin{proof}
	Applying the triangle inequality,
	\begin{equation} \label{eq:EDinvX:tti:tri}
		\left\| (\A - \P) \D^{-1} \X \right\|_{\tti}
		\le
		\left\| (\A - \P) \calD^{-1} \X \right\|_{\tti}
		+ \left\| (\A - \P)( \D^{-1} - \calD^{-1}) \X \right\|_{\tti}.
	\end{equation}
	Applying Lemma~\ref{lem:vectorbernstein} and a union bound over all $i \in [n]$, it holds with high probability that
	\begin{equation*} \begin{aligned}
			\left\| (\A - \P) \calD^{-1} \X \right\|_{\tti}
			 & = \max_{i \in [n]} \left\| \sum_{j=1}^n \frac{ (\A - \P)_{ij}\X_j }
			{ \delta_j } \right\|
			\le \left( \sum_{j=1}^n \frac{ \|\X_j\|^2 }{ \delta_j^2 }
			\right)^{1/2} \sqrt{ \nu+b^2} \log n                                   \\
			 & \le \frac{ C\sqrt{\nu+b^2} \log n }{ \sqrt{n} \rho }
			\left( \sum_{j=1}^n
			\frac{ n \rho^2 \|\X_j\|^2 }{ \delta_j^2 }
			\right)^{1/2}.
		\end{aligned} \end{equation*}
	Note that by Corollary~\ref{cor:boundedexpectation}, Lemma~\ref{lem:ratio} applies with $q=r=2$. Combining this with our assumption that $(\nu+b^2) = \Theta( \rho )$,
	\begin{equation} \label{eq:EDinvX:tti:step1}
		\left\| (\A - \P) \calD^{-1} \X \right\|_{\tti}
		\le \frac{ C \log n }{ \sqrt{ n \rho } }.
	\end{equation}

	Applying the definition of the $(\tti)$-norm followed by Lemmas~\ref{lem:sgmax} and~\ref{lem:degrecip:conc},
	\begin{equation*} \begin{aligned}
			\left\| (\A - \P)( \D^{-1} - \calD^{-1}) \X \right\|_{\tti}
			 & = \max_{i \in [n]}
			\left\| \sum_{j=1}^n (\A - \P)_{ij}
			\left( \frac{1}{d_j} - \frac{1}{\delta_j} \right)\X_j
			\right\|                                                          \\
			 & \le C \left\| \sum_{j=1}^n \frac{\X_j }{ \delta_j^2 } \right\|
			(\nu+b^2) \sqrt{n} \log^2 n.
		\end{aligned} \end{equation*}
	Applying Lemma~\ref{lem:ratio} with $r=1,q=2$,
	\begin{equation*}
		\left\| (\A - \P)( \D^{-1} - \calD^{-1}) \X \right\|_{\tti}
		\le \frac{ C (\nu+b^2) \sqrt{n} \log^2 n }{ n \rho^2 }
		\le \frac{ C \log^2 n }{ \sqrt{n} \rho },
	\end{equation*}
	where we have again used our growth assumption in Equation~\eqref{eq:nubrho:ratio}.
	Applying this and Equation~\eqref{eq:EDinvX:tti:step1} to Equation~\eqref{eq:EDinvX:tti:tri},
	\begin{equation*}
		\left\| (\A - \P) \D^{-1} \X \right\|_{\tti}
		\le \frac{ C \log n }{ \sqrt{n \rho } }
		+ \frac{ C \log^2 n }{ \sqrt{n} \rho }
		= O\left( \frac{ \log^2 n }{ \sqrt{n} \rho } \right),
	\end{equation*}
	completing the proof.
\end{proof}

\begin{lemma} \label{lem:Gqp2XtoDinvVersion}
	Let $(\A, \X) \sim \RDPG( F, n)$ with $(\nu,b)$-subgamma edges and sparsity parameter $\rho$ and suppose that
	Assumptions~\ref{assum:growth:sparsity},~\ref{assum:Fsparse:interact},~\ref{assum:F:extremes:norho} and~\ref{assum:F:momentratio} hold.

	Let $q \ge 0$ be an integer. Then for any constant $\tau > 0$,
	\begin{equation*}
		\left\| \G^{q+2} \X
		- \rho \D^{-1} \X \left( \rho \X^\top \D^{-1} \X \right)^q \X^\top \X
		\right\|_{\tti}
		\le
		\frac{ C (q+1) (1+\tau)^q \sqrt{n} \log^2 n }{ \deltamin }.
	\end{equation*}
\end{lemma}
\begin{proof}
	Recalling that $\G = \D^{-1} \A$ and $\P =\rho \X \X^\top$,
	\begin{equation} \label{eq:Gqp2X:step1}
		\G^{q+2} \X
		= \G^{q+1} (\rho \D^{-1} \X) (\X^\top \X) + \G^{q+1} \D^{-1} (\A - \P) \X.
	\end{equation}
	Applying basic properties of the $(\tti)$-norm
	followed by Lemma~\ref{lem:degLB} and using the fact that $\G^{q+1}$ is row-stochastic,
	\begin{equation*}
		\left\| \G^{q+1} \D^{-1} (\A - \P) \X \right\|_{\tti}
		\le
		\left\| \G^{q+1} \right\|_\infty \left\| \D^{-1} \right\|_\infty
		\left\| (\A - \P) \X \right\|_{\tti}
		\le
		\frac{ C }{ \deltamin } \left\| (\A - \P) \X \right\|_{\tti}.
	\end{equation*}
	Further applying Lemma~\ref{lem:vectorbernstein} to control $\| (\A - \P) \X \|_{\tti}$,
	\begin{equation*}
		\left\| \G^{q+1} \D^{-1} (\A - \P) \X \right\|_{\tti}
		\le \frac{ C \sqrt{\nu+b^2} \log n  }{ \deltamin }
		\left( \sum_{i=1}^n \|\X_i\|^2 \right)^{1/2}.
	\end{equation*}
	Multiplying through by appropriate quantities and applying the law of large numbers,
	\begin{equation*}
		\left\| \G^{q+1} \D^{-1} (\A - \P) \X \right\|_{\tti}
		\le \frac{ C \sqrt{\nu+b^2} \sqrt{n} \log n  }{ \deltamin }.
	\end{equation*}
	Rearranging Equation~\eqref{eq:Gqp2X:step1} and applying this bound,
	\begin{equation} \label{eq:Gqp2X:bound1}
		\left\| \G^{q+2} \X - \G^{q+1} (\rho \D^{-1} \X) (\X^\top \X) \right\|_{\tti}
		\le \frac{ C \sqrt{\nu+b^2} \sqrt{n} \log n  }{ \deltamin }.
	\end{equation}
	Again recalling $\G = \D^{-1} \A$ and adding and subtracting appropriate quantities,
	\begin{equation} \label{eq:Gqp2X:step2} \begin{aligned}
			\G^{q+1} (\rho \D^{-1} \X) (\X^\top \X)
			 & = \G^q \rho \D^{-1} \X (\rho \X^\top \D^{-1} \X) (\X^\top \X)     \\
			 & ~~~~~~~~~+ \G^q \D^{-1} (\A - \P) (\rho \D^{-1} \X) (\X^\top \X).
		\end{aligned} \end{equation}
	Again applying basic properties of the $(\tti)$-norm,
	\begin{equation*}
		\left\| \G^q \D^{-1} (\A - \P) (\rho \D^{-1} \X) (\X^\top \X) \right\|_{\tti}
		\le
		\rho \left\| \G^q \right\|_{\infty} \left\| \D^{-1} \right\|_{\infty}
		\left\| (\A - \P) \D^{-1} \X \right\|_{\tti} \left\| \X \right\|^2.
	\end{equation*}
	Using the fact that $\G^q$ is row-stochastic and applying
	Lemmas~\ref{lem:degLB},~\ref{lem:Xspec} and~\ref{lem:EDinvX:tti},
	\begin{equation*}
		\left\| \G^q \D^{-1} (\A - \P) (\rho \D^{-1} \X) (\X^\top \X) \right\|_{\tti}
		\le
		\frac{ C \sqrt{n} \log n }{ \deltamin } .
	\end{equation*}
	Rearranging Equation~\eqref{eq:Gqp2X:step2}, applying the triangle inequality and using this bound,
	\begin{equation*}
		\left\| \G^{q+1} (\rho \D^{-1} \X) (\X^\top \X)
		- \G^q (\rho \D^{-1} \X) \left( \rho \X^\top \D^{-1} \X \right)
		(\X^\top \X) \right\|_{\tti}
		\le \frac{ C \sqrt{n} \log n }{ \deltamin }.
	\end{equation*}
	Applying the triangle inequality and combining this bound with Equation~\eqref{eq:Gqp2X:bound1},
	\begin{equation*}
		\left\| \G^{q+2} \X - \G^q (\rho \D^{-1} \X) ( \rho \X^\top \D^{-1} \X ) (\X^\top \X)
		\right\|_{\tti}
		\le \frac{ C\left( 1 + \sqrt{\nu+b^2} \right) \sqrt{n} \log^2 n }
		{ \deltamin }.
	\end{equation*}
	Applying our growth assumption in Equation~\eqref{eq:nubrho:ratio} and using the fact that $\rho = O(1)$ trivially, it holds with high probability that
	\begin{equation} \label{eq:Gqp2X:basecase}
		\left\| \G^{q+2} \X - \G^q (\rho \D^{-1} \X) ( \rho \X^\top \D^{-1} \X ) (\X^\top \X)
		\right\|_{\tti}
		\le \frac{ C \sqrt{n} \log^2 n }{ \deltamin }.
	\end{equation}

	If $q=0$, our proof is complete, so suppose that $q \ge 1$.
	An argument essentially identical to that following Equation~\eqref{eq:Gqp2X:step2} yields
	\begin{equation*} \begin{aligned}
			 & \left\| \G^q (\rho \D^{-1} \X)(\rho \X^\top \D^{-1} \X)( \X^\top \X)
			- \G^{q-1} (\rho \D^{-1} \X)(\rho \X^\top \D^{-1} \X)^2( \X^\top \X)
			\right\|_{\tti}                                                         \\
			 & ~~~~~~\le
			\left\| \G^{q-1} \D^{-1} (\A - \P) (\rho \D^{-1} \X)
			(\rho \X^\top \D^{-1} \X)( \X^\top \X) \right\|_{\tti}                  \\
			 & ~~~~~~\le
			\rho \left\| \D^{-1} \right\|_{\infty}
			\left\| (\A - \P) \D^{-1} \right\|_{\tti}
			\left\| \rho \X^\top \D^{-1} \X \right\| \left\| \X^\top \X \right\|    \\
			 & ~~~~~~\le \frac{ C \sqrt{n} \log^2 n }{ \deltamin }
			\left\| \rho \X^\top \D^{-1} \X \right\|.
		\end{aligned} \end{equation*}
	Applying the triangle inequality and recursively repeating this argument,
	\begin{equation*} \begin{aligned}
			 & \left\| \G^q (\rho \D^{-1} \X)(\rho \X^\top \D^{-1} \X)( \X^\top \X)
			- (\rho \D^{-1} \X)(\rho \X^\top \D^{-1} \X)^{q+1}( \X^\top \X)
			\right\|_{\tti}                                                         \\
			 & ~~~~~~~~~\le
			\frac{ C \sqrt{n} \log^2 n }{ \deltamin }
			\sum_{m=1}^q \left\| \rho \X^\top \D^{-1} \X \right\|^m.
		\end{aligned} \end{equation*}
	Applying the triangle inequality and using Equation~\eqref{eq:Gqp2X:basecase},
	\begin{equation} \label{eq:Gqp2X:recurse}
		\left\| \G^{q+2} \X
		- (\rho \D^{-1} \X)(\rho \X^\top \D^{-1} \X)^{q+1}( \X^\top \X) \right\|_{\tti}
		\le
		\frac{ C \sqrt{n} \log^2 n }{ \deltamin }
		\sum_{m=0}^q \left\| \rho \X^\top \D^{-1} \X \right\|^m.
	\end{equation}

	Applying Lemma~\ref{lem:XDinvXtoXcalDinvX:converge} and using our growth assumption in Equation~\eqref{eq:growth:sparsityLB}, since $\tau > 0$ is constant by assumptio is constant by assumption, it holds with high probability that for all suitably large $n$ that
	\begin{equation*}
		\left\| \rho \X^\top \D^{-1} \X \right\|^m
		\le \left( \left\| \rho \X^\top \calD^{-1} \X \right\| + \tau \right)^m
		\le \left( 1 + \tau \right)^m,
	\end{equation*}
	where the second inequality follows from Lemma~\ref{lem:XcalDinvX:spectral}.
	Applying this bound to Equation~\eqref{eq:Gqp2X:recurse} and trivially upper-bounding the sum,
	\begin{equation*} \begin{aligned}
			\left\| \G^{q+2} \X
			- (\rho \D^{-1} \X)(\rho \X^\top \D^{-1} \X)^{q+1}( \X^\top \X) \right\|_{\tti}
			 & \le \frac{ C \sqrt{n} \log^2 n }{ \deltamin }
			\sum_{m=0}^q (1+\tau)^m                                            \\
			 & \le \frac{ C (q+1) (1+\tau)^q \sqrt{n} \log^2 n }{ \deltamin },
		\end{aligned} \end{equation*}
	as we set out to show.
\end{proof}

\begin{lemma} \label{lem:rdpg:contagion:bigconverge}
	Let $(\A, \X) \sim \RDPG( F, n)$ with $(\nu,b)$-subgamma edges and sparsity parameter $\rho$, and suppose that
	Assumptions~\ref{assum:F:momentratio},~\ref{assum:growth:sparsity},~\ref{assum:Fsparse:interact} and~\ref{assum:F:extremes:norho} hold
	as well as the growth assumption in
	Equation~\eqref{eq:iplb:rho:newer}. %
	Then
	\begin{equation*}
		\left\| \rho \D^{-1} \X
		\left( I - \beta \rho \X^\top \D^{-1} \X \right)^{-1} \X^\top \X
		- \bm H^{-1} \X \bm \Gamma \secmm \right\|_{\tti} = o( 1 ).
	\end{equation*}
\end{lemma}
\begin{proof}
	Recursively applying Lemma~\ref{lem:XDinvXtoXcalDinvX:converge}, for any $q \ge 0$,
	\begin{equation*} \begin{aligned}
			\left\| \left( \beta \rho \X^\top \D^{-1} \X \right)^q
			- \left( \beta \rho \X^\top \calD^{-1} \X \right)^q \right\|
			 & \le \frac{ C |\beta|^q \log n }{ \sqrt{n \rho} }
			\sum_{m=0}^{q-1}
			\left( \frac{ \log n }{ \sqrt{n \rho} } \right)^m
			\le \frac{ C q |\beta|^q \log n }{ \sqrt{n \rho} },
		\end{aligned} \end{equation*}
	where the second inequality follows from our assumption in Equation~\eqref{eq:growth:sparsityLB}.
	Applying the Neumann expansion followed by the triangle inequality, it follows that
	\begin{equation} \label{eq:neumann:converge:DinvTocalDinv}
		\begin{aligned}
			\left\| \left( I\! -\! \beta \rho \X^\top \D^{-1} \X \right)^{-1}
			\!- \left( I\! -\! \beta \rho \X^\top \calD^{-1} \X \right)^{-1}
			\right\|
			 & \le \frac{ C \log n }{ \sqrt{n \rho} }
			\sum_{q=0}^\infty q |\beta|^q
			= O\left( \frac{  \log n }{ \sqrt{n \rho} } \right),
		\end{aligned} \end{equation}
	where we have used the fact that $|\beta| < 1$.

	Multiplying through by appropriate quantities
	and using properties of the $(\tti)$-norm,
	\begin{equation*} \begin{aligned}
			 & \left\| \rho \D^{-1} \X
			\left( I - \beta \rho \X^\top \D^{-1} \X \right)^{-1} \X^\top \X
			- \rho \D^{-1} \X
			\left( I - \beta \rho \X^\top \calD^{-1} \X \right)^{-1} \X^\top \X
			\right\|_{\tti}                           \\
			 & ~~~~~~ \le
			\left\| n \rho \D^{-1} \X \right\|_{\tti}
			\left\|  \left( I - \beta \rho \X^\top \D^{-1} \X \right)^{-1}
			- \left( I - \beta \rho \X^\top \calD^{-1} \X \right)^{-1} \right\|
			\left\| \frac{ \X^\top \X }{ n } \right\| \\
			 & ~~~~~~ \le
			\frac{ C \sqrt{n} \log^2 n }{ \deltamin }
			\left\| n \rho \D^{-1} \X \right\|_{\tti} \left\| \frac{ \X^\top \X }{n} \right\|.
		\end{aligned} \end{equation*}
	Applying the law of large numbers to the sample covariance $n^{-1} \X^\top \X$, it follows that
	\begin{equation*} \begin{aligned}
			 & \left\| \rho \D^{-1} \X
			\left( I - \beta \rho \X^\top \D^{-1} \X \right)^{-1} \X^\top \X
			- \rho \D^{-1} \X
			\left( I - \beta \rho \X^\top \calD^{-1} \X \right)^{-1} \X^\top \X
			\right\|_{\tti}                                        \\
			 & ~~~~~~\le \frac{ C \sqrt{n} \log^2 n }{ \deltamin }
			\left\| n \rho \D^{-1} \X \right\|_{\tti}.
		\end{aligned} \end{equation*}
	Applying Lemma~\ref{lem:DinvX:tti} followed by Lemma~\ref{lem:HinvX:constant},
	\begin{equation*} \begin{aligned}
			 & \left\| \rho \D^{-1} \X
			\left( I - \beta \rho \X^\top \D^{-1} \X \right)^{-1} \X^\top \X
			- \rho \D^{-1} \X
			\left( I - \beta \rho \X^\top \calD^{-1} \X \right)^{-1} \X^\top \X
			\right\|_{\tti}                                        \\
			 & ~~~~~~\le \frac{ C \sqrt{n} \log^2 n }{ \deltamin }
			\left[ \left\| \bm H^{-1} \X \right\|_{\tti} + o(1) \right]
			\le \frac{ C \sqrt{n} \log^2 n }{ \deltamin }.
		\end{aligned} \end{equation*}
	Applying Lemma~\ref{lem:deltaLB} followed by our growth assumption in Equation~\eqref{eq:iplb:rho:newer},
	\begin{equation} \label{eq:rdpg:contagion:tti:tri2}
		\begin{aligned}
			 & \left\| \rho \D^{-1} \X
			\left( I - \beta \rho \X^\top \D^{-1} \X \right)^{-1} \X^\top \X
			- \rho \D^{-1} \X
			\left( I - \beta \rho \X^\top \calD^{-1} \X \right)^{-1} \X^\top \X
			\right\|_{\tti}            \\
			 & ~~~~~~\le
			\frac{ C \log^2 n }{ \sqrt{n} \rho \min_{i \in [n]}\X_i^\top \bm \mu }
			= o(1)~\text{ almost surely.}
		\end{aligned} \end{equation}

	We note that Lemma~\ref{lem:XcalDinvX:spectral} along with the fact that $|\beta \rho| < 1$ implies that
	\begin{equation*}
		\left\| \left( I -\beta \rho \X^\top \calD^{-1} \X \right)^{-1} \right\|
		= O( 1 ).
	\end{equation*}
	It follows that, by basic properties of the $(\tti)$-norm,
	\begin{equation*} \begin{aligned}
			 & \left\| \left( n \rho \D^{-1} \X - \bm H^{-1} \X \right)
			\left( I - \beta \rho \X^\top \calD^{-1} \X \right)^{-1}
			\frac{ \X^\top \X }{ n } \right\|_{\tti}                                  \\
			 & ~~~~~~~~~\le \left\| n \rho \D^{-1} \X - \bm H^{-1} \X \right\|_{\tti}
			\left\| \frac{ \X^\top \X }{ n } \right\|.
		\end{aligned} \end{equation*}
	Applying Lemma~\ref{lem:DinvX:tti} and the law of large numbers,
	\begin{equation} \label{eq:rdpg:contagion:tti:tri3}
		\left\| n \rho \D^{-1} \X\!
		\left( \!I\! -\! \beta \rho \X^\top \calD^{-1} \X \right)^{-1}\!
		\frac{ \X^\top\! \X }{ n }
		-
		\bm H^{-1} \X\!
		\left( \!I \!-\! \beta \rho \X^\top \calD^{-1}\! \X \right)^{-1}\!
		\frac{ \X^\top\X }{ n }
		\right\|_{\tti}
		= o( 1 ).
	\end{equation}

	By a similar argument, this time applying the law of large numbers and using Lemma~\ref{lem:HinvX:constant} to control $\bm H^{-1} \X$,
	\begin{equation} \label{eq:rdpg:contagion:tti:tri4}
		\begin{aligned}
			 & \left\| \bm H^{-1} \X \left( I - \beta \rho \X^\top \calD^{-1} \X \right)^{-1}
			\frac{ \X^\top \X }{ n }
			- \bm H^{-1} \X \left( I - \beta \rho \X^\top \calD^{-1} \X \right)^{-1}
			\secmm \right\|_{\tti}                                                            \\
			 & ~~~~~~\le
			\left\| \bm H^{-1} \X \right\|_{\tti}
			\left\| \frac{ \X^\top \X }{ n } - \secmm \right\|
			= o( 1 )~\text{ almost surely.}
		\end{aligned} \end{equation}

	Applying basic properties of the $(\tti)$-norm,
	\begin{equation*} \begin{aligned}
			 & \left\| \bm H^{-1} \X \left( I - \beta \rho \X^\top \calD^{-1} \X \right)^{-1}
			\secmm
			- \bm H^{-1} \X
			\left( I - \beta \bbE \frac{\X_1\X_1^\top }{\X_1^\top \bm \mu }
			\right)^{-1} \secmm
			\right\|_{\tti}                                                                   \\
			 & ~~~~~~\le
			\left\| \bm H^{-1} \X \right\|_{\tti}
			\left\| \left( I - \beta \rho \X^\top \calD^{-1} \X \right)^{-1}
			- \left( I - \beta \bbE \frac{\X_1\X_1^\top }{\X_1^\top \bm \mu }
			\right)^{-1} \right\|
			\left\| \secmm \right\|.
		\end{aligned} \end{equation*}
	Lemmas~\ref{lem:HinvX:constant} and~\ref{lem:XcalDinvX2Expec:converge}, along with the continuous mapping theorem imply that almost surely,
	\begin{equation} \label{eq:rdpg:contagion:tti:tri5}
		\left\| \bm H^{-1} \X \left( I - \beta \rho \X^\top \calD^{-1} \X \right)^{-1}
		\secmm
		- \bm H^{-1} \X
		\left( I - \beta \bbE \frac{\X_1\X_1^\top }{\X_1^\top \bm \mu }
		\right)^{-1} \secmm
		\right\|_{\tti}
		= o( 1 ).
	\end{equation}

	Applying the triangle inequality and combining
	Equations~\eqref{eq:rdpg:contagion:tti:tri2},~\eqref{eq:rdpg:contagion:tti:tri3},~\eqref{eq:rdpg:contagion:tti:tri4} and~\eqref{eq:rdpg:contagion:tti:tri5},
	\begin{equation*}
		\left\| \rho \D^{-1} \X \left( I - \beta \rho \X^\top \D^{-1} \X \right)^{-1}
		\X^\top \X
		- \bm H^{-1} \X \left( I - \beta \bbE \frac{\X_1\X_1^\top }{\X_1^\top \bm \mu } \right)^{-1}
		\secmm \right\|_{\tti}
		= o( 1 )
	\end{equation*}
	almost surely, completing the proof.
\end{proof}

\begin{lemma} \label{lem:rdpg:contagion:tti}
	Let $(\A, \X) \sim \RDPG( F, n)$ with $(\nu,b)$-subgamma edges and sparsity parameter $\rho$
	and suppose that
	Assumptions~\ref{assum:growth:sparsity},~\ref{assum:Fsparse:interact},~\ref{assum:F:extremes:norho} and~\ref{assum:F:momentratio} hold.
	Then
	\begin{equation*}
		\left\| \left(\I - \beta \G \right)^{-1} \G^2 \X
		- \bm H^{-1} \X \bm \Gamma \secmm \right\|_{\tti}
		= o(1)~\text{ almost surely.}
	\end{equation*}
\end{lemma}
\begin{proof}
	Applying Lemma~\ref{lem:Gqp2XtoDinvVersion}, and letting $\tau > 0$ be a constant of our choosing to be specified below, it holds with high probability that for all $q \ge 0$,
	\begin{equation} \label{eq:betaq:bound}
		\left\| \beta^q \G^{q+2} \X
		- \beta^q \rho \D^{-1} \X \left( \rho \X^\top \D^{-1} \X \right)^q \X^\top \X
		\right\|_{\tti}
		\le
		\frac{ C |\beta|^q (q+1)(1+\tau)^q \sqrt{n} \log^2 n }{ \deltamin }.
	\end{equation}
	Recalling that by the Neumann expansion we have
	\begin{equation*}
		\left(\I - \beta \G\right)^{-1} \G^2 \X = \sum_{q=0}^\infty \beta^q \G^{q+2} \X,
	\end{equation*}
	Applying the triangle inequality and using the bound in Equation~\eqref{eq:betaq:bound} yields that
	\begin{equation*} \begin{aligned}
			 & \left\| \left(\I - \beta \G\right)^{-1} \G^2 \X
			- \sum_{q=0}^\infty \beta^q \rho \D^{-1} \X
			\left( \rho \X^\top \D^{-1} \X \right)^q \X^\top \X \right\|_{\tti} \\
			 & ~~~~~~\le
			\sum_{q=0}^\infty |\beta|^q
			\left\| \G^{q+2} \X
			- \rho \D^{-1} \X \left( \rho \X^\top \D^{-1} \X \right)^q \X^\top \X
			\right\|_{\tti}                                                     \\
			 & ~~~~~~\le
			\frac{ C \sqrt{n} \log^2 n }{ \deltamin }
			\sum_{q=0}^\infty (q+1) (1+\tau)^q|\beta|^q .
		\end{aligned} \end{equation*}
	Choosing $\tau$ small enough that $(1+\tau)|\beta| < 1$, the infinite sum converges to a quantity depending only on constants $\tau$ and $\beta$, and it follows that
	\begin{equation*}
		\left\| \left(\I - \beta \G\right)^{-1} \G^2 \X
		- \sum_{q=0}^\infty \beta^q \rho \D^{-1} \X
		\left( \rho \X^\top \D^{-1} \X \right)^q \X^\top \X \right\|_{\tti}
		\le \frac{ C \sqrt{n} \log^2 n }{ \deltamin }.
	\end{equation*}
	Applying Lemma~\ref{lem:deltaLB} followed by our growth assumption in Equation~\eqref{eq:iplb:rho:newer},
	\begin{equation} \label{eq:rdpg:contagion:tti:tri1:prelim}
		\left\| \left(\I - \beta \G\right)^{-1} \G^2 \X
		- \sum_{q=0}^\infty \beta^q \rho \D^{-1} \X
		\left( \rho \X^\top \D^{-1} \X \right)^q \X^\top \X \right\|_{\tti}
		= o(1) ~\text{ almost surely.}
	\end{equation}

	By Lemmas~\ref{lem:XcalDinvX:spectral} and~\ref{lem:XDinvXtoXcalDinvX:converge}, for any constant $\tau > 0$, it holds for all suitably large $n$ that
	\begin{equation*}
		\left\| \rho \X^\top \D^{-1} \X \right\|^m
		\le \left( \left\| \rho \X^\top \calD^{-1} \X \right\| + \tau \right)^m
		\le \left( 1 + \tau \right)^m,
	\end{equation*}
	Again choosing $\tau$ small enough that $(1+\tau)|\beta| < 1$, we have
	\begin{equation*}
		\left\| \rho \beta \X^\top\D^{-1} \X \right\|
		\le |\beta| \left\| \rho \X^\top \D^{-1} \X \right\|
		\le |\beta|(1+\tau) < 1,
	\end{equation*}
	so that the Neumann expansion converges, and
	\begin{equation*} \begin{aligned}
			\sum_{q=0}^\infty \beta^q \rho \D^{-1} \X
			\left( \rho \X^\top \D^{-1} \X \right)^q \X^\top \X
			 & = \rho \D^{-1} \X
			\left[
				\sum_{q=0}^\infty \left( \beta \rho \X^\top \D^{-1} \X \right)^q
			\right] \X^\top \X   \\
			 & = \rho \D^{-1} \X
			\left( I - \beta \rho \X^\top \D^{-1} \X \right)^{-1} \X^\top \X.
		\end{aligned} \end{equation*}

	Plugging this into Equation~\eqref{eq:rdpg:contagion:tti:tri1:prelim},
	it holds with high probability that
	\begin{equation*}
		\left\| \left(\I - \beta \G\right)^{-1} \G^2 \X
		- \rho \D^{-1} \X \left( I - \beta \rho \X^\top \D^{-1} \X \right)^{-1}
		\X^\top \X \right\|_{\tti}
		\le \frac{ C \sqrt{n} \log^2 n }{ \deltamin }.
	\end{equation*}
	Applying Lemma~\ref{lem:deltaLB} followed by our growth assumption in Equation~\eqref{eq:iplb:rho:newer},
	\begin{equation*}
		\left\| \left( \I - \beta \G \right)^{-1} G^2 \X
		- \rho \D^{-1} \X \left( I - \beta \rho \X^\top \D^{-1} \X \right)^{-1}
		\X^\top \X \right\|_{\tti}
		= o( 1 ) ~\text{ almost surely. }
	\end{equation*}

	Lemma~\ref{lem:rdpg:contagion:bigconverge} implies that
	\begin{equation*}
		\left\| \rho \D^{-1} \X \left( I - \beta \rho \X^\top \D^{-1} \X \right)^{-1}
		\X^\top \X
		- \bm H^{-1} \X \bm \Gamma \secmm \right\|_{\tti}
		= o( 1 ) ~\text{ almost surely. }
	\end{equation*}
	The triangle inequality and combining the above two displays completes the proof.
\end{proof}

\begin{lemma} \label{lem:rdpg:interference:tti}
	Let $(\A, \X) \sim \RDPG( F, n)$ with $(\nu,b)$-subgamma edges and sparsity parameter $\rho$
	and suppose that
	Assumptions~\ref{assum:growth:sparsity},~\ref{assum:Fsparse:interact},~\ref{assum:F:extremes:norho} and~\ref{assum:F:momentratio} hold.
	Then
	\begin{equation*}
		\left\| \left(\I - \beta \G \right)^{-1} \G \X
		- \bm H^{-1} \X \left( I + \beta \bm \Gamma \right) \secmm \right\|_{\tti} \\
		= o(1)~\text{ almost surely.}
	\end{equation*}
\end{lemma}
\begin{proof}
	Applying the Neumann expansion,
	\begin{equation*}
		\left(\I - \beta \G \right)^{-1} \G \X
		= \G \X + \sum_{q=1}^\infty \beta^q \G^{q+1} \X
		= \G \X + \beta \left( \I - \beta \G \right)^{-1} \G^2 \X.
	\end{equation*}
	Applying the triangle inequality followed by Lemmas~\ref{lem:rdpg:GXconverge} and~\ref{lem:rdpg:contagion:tti},
	\begin{equation*} \begin{aligned}
			 & \left\| \left(\I - \beta \G \right)^{-1} \G \X
			- \bm H^{-1} \X \left( I + \beta \bm \Gamma \right) \secmm \right\|_{\tti}             \\
			 & ~~~~~~~~~\le
			\left\| \G \X - \bm H^{-1} \X \secmm \right\|_{\tti}
			+ |\beta|
			\left\| (\I - \beta \G)^{-1} \G^2 \X - \bm H^{-1} \X \bm \Gamma \secmm \right\|_{\tti} \\
			 & ~~~~~~~~~= o(1)~\text{ almost surely,}
		\end{aligned} \end{equation*}
	as we set out to show.
\end{proof}

\renewcommand{\harvardurl}{\url}
\bibliography{references}

\end{document}